\newif\ifacm
\acmtrue  
\ifacm
  \documentclass[acmsmall,screen]{acmart}
  \AtBeginDocument{%
  }
    \newcommand{\subfigtextwidth}{1.0\textwidth}
\else
  \documentclass[10pt,journal]{IEEEtran}
    \newcommand{\subfigtextwidth}{0.475\textwidth}
\fi

\usepackage{amsmath, mathtools, amsfonts} 
\usepackage{mathrsfs}

\usepackage{xcolor}  
\usepackage{array, multirow, multicol, makecell, tabularx, booktabs}
\usepackage{ragged2e}
\usepackage{caption, subcaption}

\usepackage{tikz, pgfplots}
\usetikzlibrary{arrows.meta, automata, positioning, shapes.geometric, shadows, shapes.symbols, calc}

\usepackage{siunitx}
\usepackage[normalem]{ulem} 
\usepackage{url, nameref, balance, lipsum} 

\usepackage{listings}

\usepackage[ruled, vlined, linesnumbered]{algorithm2e} 

\SetCommentSty{mycommfont}

\newcommand{\mR}{\mathsf{R}}

\newcommand{\ta}{\mathcal{A}}

\newcommand{\mdp}{\mathcal{M}}

\newcommand{\los}{\ell}

\newcommand{\schr}{\pi}
\newcommand{\argmax}{\arg\max}
\newcommand{\argmin}{\arg\min}
\newcommand{\vect}[1]{\boldsymbol{#1}}
\newcommand{\tolrew}[2]{{\mR}_{#1}^{#2}}
\newcommand{\Polytope}[1]{[\![{\mR}_{#1}]\!]}
\newcommand{\expctl}{\mathbf{Exp}_{\mathit{CTL}}}
\newcommand{\probld}{\mathbf{Prob}_{\mathit{LD}}}
\newcommand{\prhw}{\mathbf{PR}_{\mathit{HW}}}
\newcommand{\prttc}{\mathbf{PR}_{\mathit{TTC}}}
\newcommand{\pchw}{\mathbf{P50}_{\mathit{HW}}}
\newcommand{\pcttc}{\mathbf{P10}_{\mathit{TTC}}}

\usepackage{amsthm} 
\usepackage{etoolbox} 

\makeatletter
\@ifclassloaded{letter}
{
}
{
    
    \newtheorem{theorem}{Theorem}[section]
    \newtheorem{assumption}[theorem]{Assumption}
    
    \newtheorem{proposition}[theorem]{Proposition}
    \newtheorem{definition}[theorem]{Definition} 
    \newtheorem{lemma}[theorem]{Lemma} 
    \newtheorem{claim}[theorem]{Claim} 
    
   \newcommand{\rrev}[1]{{#1}} 
    \newcommand{\rev}[1]{{#1}} 
    \newcommand{\orev}[1]{{#1}} 
    
    \usepackage[colorinlistoftodos,color=red!30,textwidth=1.2cm,textsize=scriptsize,textcolor=blue,bordercolor=white]{todonotes}
    
    \makeatletter
    \@mparswitchfalse
    \makeatother
    \normalmarginpar
}
\makeatother

\begin{document}
\title{Fault-Tolerant Design and Multi-Objective Model Checking for Real-Time Deep Reinforcement Learning Systems}

\ifacm
\setcopyright{acmlicensed}
\copyrightyear{2026}
\acmYear{2026}
\acmDOI{XXXXXXX.XXXXXXX}

\acmJournal{TOSEM}
\acmVolume{00}
\acmNumber{0}
\acmArticle{000}
\acmMonth{0}

\author{Guoxin Su}
\affiliation{%
  \institution{University of Wollongong}
  \city{Wollongong}
  \country{Australia}}
\email{guoxin@uow.edu.com}
\orcid{0000-0002-2087-4894}

\author{Thomas Robinson}
\affiliation{%
  \institution{Fortescue Zero}
  \city{Kidlington}
  \country{UK}}
  \orcid{0000-0002-6150-2587}

\author{Hoa Khanh Dam}
\affiliation{%
 \institution{University of Wollongong}
 \city{Wollongong}
  \country{Australia}}
  \orcid{0000-0003-4246-0526}

\author{Li Liu}
\affiliation{%
  \institution{Chongqing University}
  \city{Chongqing}
  \country{China}}
 \orcid{https://orcid.org/0000-0002-4776-5292}

\author{David S.~Rosenblum}
\affiliation{%
  \institution{George Mason University}
  \city{Fairfax}
  \state{Virginia}
  \country{USA}}
  \orcid{0000-0003-1685-4206}

\renewcommand{\shortauthors}{Su et al.}

\begin{abstract}
Deep reinforcement learning (DRL) has emerged as a powerful paradigm for solving complex decision-making problems. However, DRL-based systems still face significant dependability challenges particularly in real-time environments due to the simulation-to-reality gap, out-of-distribution observations, and the critical impact of latency. \rev{Latency-induced faults\rrev{, in particular,}  can lead to unsafe or unstable behaviour, yet existing} fault-tolerance approaches to DRL systems lack formal methods to rigorously analyse and optimise performance and safety simultaneously in real-time settings. To address this, we propose a formal framework for designing and analysing real-time switching mechanisms between DRL agents and alternative controllers. Our approach leverages Timed Automata (TAs) for explicit switch logic design, which is then syntactically converted to a Markov Decision Process (MDP) for formal analysis. We develop a novel convex query technique for multi-objective model checking, enabling the optimisation of soft performance objectives while ensuring hard safety constraints for MDPs. Furthermore, we present MOPMC, a GPU-accelerated software tool implementing this technique, demonstrating superior scalability in both model size and objective numbers. 
\end{abstract}



\received{\today}
\received[revised]{dd/mm/yyyy}
\received[accepted]{25/02/2026}

\maketitle

\fi

\section{Introduction}
\ifacm Deep \fi
reinforcement learning (DRL) has emerged as a powerful paradigm for solving complex sequential decision-making problems, with notable successes in domains like autonomous driving
, robotic control
, and industrial automation
. By combining deep neural networks with the trial-and-error learning process of reinforcement learning, DRL enables agents to derive optimal control policies directly from environmental \rrev{interactions}, often surpassing traditional methods where explicit modelling is intractable.

However, the deployment of DRL is hindered by inherent challenges such as the simulation-to-reality (sim-to-real) gap~\cite{zhao2020sim,Da2025} and out-of-distribution (OOD) observations \cite{Haider2024}.
These issues can lead to unpredictable behaviours that compromise system dependability.
Consequently, significant research effort has been dedicated to engineering \emph{dependable} DRL systems,  where properties such as safety and reliability are paramount.
Current approaches to achieving DRL dependability broadly fall into two categories. The first category focuses on internal improvements to DRL algorithms by incorporating various techniques (e.g., constrained optimisation~\cite{Hasanbeig2020}, risk \rrev{aversion}~\cite{Rigter2021} and model predictive control (MPC)~\cite{zanon2020safe}).
The second category involves fault-tolerant DRL systems, which employ architectural patterns such as Simplex~\cite{Phan2020,Cai2025,Maderbacher2025} and protection components  (often referred to as ``{shields}'')~\cite{Alshiekh2018,Jansen2020,Koenighofer2020,Elsayed2021,Dunlap2023}. 
These patterns can offer dual benefits to DRL systems: (i) they can enhance overall system dependability by providing assurance for critical properties; and (ii) they can improve \rrev{the} dependability of reinforcement learning by overriding interactions of the DRL agent with the environment when necessary.

Our work specifically focuses on fault-tolerant DRL systems operating in \emph{real-time} environments.
While DRL applications in real-time control are pervasive across domains such as autonomous driving and robotic control, the unique challenges of ensuring dependability in these contexts are largely under-explored in the literature.
For these systems, fault-tolerance is concerned with not only accurate computation results but also the time at which the results are computed~\cite{Reghenzani2023}.
\rev{DRL models which are trained end-to-end are usually sensitive to time variations between observations and actions~\cite{Thodoroff22a}.} 
However, latency arising from various sources (including sensing, communication, computation and actuation~\cite{Liu2023}) often manifests as rare but critical events. These events often present \rrev{themselves} as out-of-distribution observations to a DRL agent, potentially leading to catastrophic failures that the agent, with its training experience alone, is ill-equipped to handle.

Furthermore, a fault-tolerant DRL system including heterogeneous controllers must manage multiple, often conflicting, objectives. For example, always executing the actions of the DRL controller can yield the highest average-case performance but risks violating a safety constraint. Conversely, always triggering a simple, safe controller guarantees safety but may severely degrade performance. This highlights a fundamental trade-off between maximising performance and assuring safety in devising optimal control strategies. 
\rev{In realistic control domains, a wide range of safety and performance metrics must be considered. For instance, the autonomous driving DRL literature reports up to \num{95} distinct evaluation metrics in the CARLA simulator~\cite{Delavari2025}.}
Balancing \rev{numerous and potentially conflicting} objectives for such DRL systems is inherently a multi-objective optimisation problem that existing techniques rarely address in a formal, real-time setting.

To address these challenges, we introduce a formal approach to designing and analysing a real-time switching mechanism between DRL agents and alternative controllers.
\rev{Our approach builds on the long-standing principle of \emph{redundancy} in fault-tolerant design. In real-time DRL systems, latency arising from sensing, communication, or computation can lead to unsafe or unstable behaviour. To mitigate these risks, we assume the presence of alternative controllers based on traditional control techniques (e.g., rule-based controllers and model predictive control), which prioritise critical properties such as safety and stability. The core challenge is to determine when and how to switch between the DRL agent and these controllers to assure or balance multiple, often conflicting, objectives under real-time constraints.}

\rev{Specifically, we leverage} \emph{timed automata} (TAs) \cite{Alur1994}, a well-established formalism for modelling systems manifesting time-dependent behaviours, to explicitly design the switch logic.
By distinguishing under-determined conditions (due to model abstraction) and non-deterministic design choices, our approach syntactically converts a TA model into a Markov Decision Process (MDP). 
We then employ multi-objective model checking~\cite{Forejt2012},  a specialised branch of probabilistic model checking~\cite{Katoen2016,Baier2019} that verifies MDPs against multiple temporal and reward-based properties simultaneously, to analyse the trade-off between conflicting objectives.

Our main contributions are as follows:
\begin{enumerate}
\item We present a \emph{framework} to support the formal design of switch-based real-time DRL systems and diverse analysis of design choices. 
\item We propose a novel multi-objective model checking \emph{technique}, called \emph{convex query}, which computes an optimal MDP  scheduler accounting for both hard objectives (whose satisfaction must be guaranteed) and soft objectives (where trade-offs are permissible).
\item We develop MOPMC, a \emph{software tool} for multi-objective model checking that implements our proposed technique. MOPMC supports convex queries and leverages GPU acceleration for policy and value iteration. Our evaluation results demonstrate MOPMC's superior scalability over state-of-the-art tools (i.e., Storm~\cite{Hensel2021} and PRISM~\cite{Kwiatkowska2011}) in model size and number of objectives.
\end{enumerate}

The remainder of the paper is organised as follows:
Section~\ref{sec:motivating-example} describes real-time DRL systems characteristics, a \rrev{motivating} example and our switch-based architecture.
Section~\ref{sec:approach} presents our TA-based design approach and MDP-based analysis framework.
Section~\ref{sec:cq} includes our core technique, namely, a convex query method for multi-objective MDPs.
\rev{Section~\ref{sec:tool} describes the main features and architecture of our novel tool.
Section~\ref{sec:case_study} evaluates our switch design approach and our tool's performance.}
Section~\ref{sec:related_work} discusses the related work.
Finally, Section~\ref{sec:conclusions} concludes the paper and outlines further work directions.

\section{Motivations and Example}\label{sec:motivating-example}

In this section, we first present key DRL systems characteristics that are central to our work (Section~\ref{sec:rt_drl_attributes}). We then use an autonomous-driving scenario to illustrate such a DRL system (Section~\ref{sec:example}). Finally, we outline our research focus: the mechanism for switching between the DRL agent and alternative controllers (Section~\ref{sec:control_switch_overview}).

\subsection{Real-Time DRL System Characteristics}\label{sec:rt_drl_attributes}
Real-time Deep Reinforcement Learning (DRL) systems frequently operate within dynamic and often unpredictable environments where ideal operational conditions are not always guaranteed. 
For these systems, correctness is determined by not only accurate computation results but also the time at which the results are computed.
We focus on one broad class of such systems that embody the following two {characteristics}: 

Firstly, whilst typically designed for ultra-low latency, these systems may experience \emph{time delay} with a small but non-negligible probability. These delays can arise from a multitude of sources inherent in complex real-time deployments, such as sensor acquisition lag, communication network congestion, computational bottlenecks in the processing unit, or even asynchronous updates within the DRL architecture itself~\cite{Liu2023}.

Secondly, the time delay can cause significant overall {performance degradation} and, even worse, {critical requirement violation} particularly when timely and accurate responses are paramount.
To be more specific, we distinguish two kinds of \rrev{performance} objectives: 
(1) \emph{hard performance objectives} that must be guaranteed (such as safety) for the system, and (2) \emph{soft performance objectives} that are desirable for the system and can be traded off only when necessary (e.g., optimal task completion and cost).
Even seemingly minor or infrequent delays can accumulate or critically impact the timeliness of decision-making, leading to a discrepancy between the perceived state of the environment and its actual real-time condition. This can profoundly degrade soft performance objectives and even breach the hard requirements.

\subsection{Example: Autonomous-Driving Agent}\label{sec:example}
\begin{figure}[!tbh]
\centering
\begin{subfigure}[b]{\subfigtextwidth}
\centering
\begin{tikzpicture}[
node distance=.1cm, 
drlagent/.style={
rectangle,
draw=blue!70!black, 
fill=blue!20,       
line width=0.8pt,   
minimum size=5mm,
text width=1.3cm,     
align=center,
font=\bfseries,     
inner sep=0mm,      
rounded corners=2pt, 
minimum height=1.1cm,
drop shadow,
},
datalabel/.style={
ellipse, fill=green!20, draw=blue!70!black, text width=2cm,     
align=center, font=\small, inner sep=0mm,      
text=black, minimum height=1.7cm, drop shadow,
},
arrowstyle/.style={
-Triangle, 
draw=gray!70!black, 
line width=0.8pt 
}
]
\node[drlagent] (dnn) {Agent};
\node[datalabel, left=of dnn, xshift=-0.5cm] (s) {
\textbf{Perception} \\ 
 {\footnotesize (kinematics, sensor data, etc.)}
};
\node[datalabel, right=of dnn, xshift=0.5cm] (a2) {
\textbf{Prediction} \\ 
 {\footnotesize(acceleration, steering, etc.)}
};
\draw[arrowstyle] (s.east) -- (dnn.west);
\draw[arrowstyle] (dnn.east) -- (a2.west);
\end{tikzpicture}
\caption{Perception and Prediction of DRL Agent  \label{fig:real-time-drl}}
\end{subfigure}
\vspace{.5em}
\newline
\begin{subfigure}[b]{\subfigtextwidth}
\centering
\begin{tikzpicture}[
    declare function={gamma(\z)=
    2.506628274631*sqrt(1/\z)+ 0.20888568*(1/\z)^(1.5)+ 0.00870357*(1/\z)^(2.5)- (174.2106599*(1/\z)^(3.5))/25920- (715.6423511*(1/\z)^(4.5))/1244160)*exp((-ln(1/\z)-1)*\z;},
    declare function={gammapdf(\x,\k,\theta) = 1/(\theta^\k)*1/(gamma(\k))*\x^(\k-1)*exp(-\x/\theta);}
]
\begin{axis}[
  no markers, domain=0:9, samples=100,
  axis x line=left, axis y line=none, xlabel=\textit{time}, ylabel=\textit{prob.},
  every axis y label/.style={at=(current axis.above origin),anchor=east},
  every axis x label/.style={at=(current axis.right of origin),anchor=north},
  height=2.75cm, width=8.6cm,
  xtick={0, 9.0}, 
  xticklabels={$0$, $\Delta$},
  enlargelimits=false, clip=false, axis on top,
  ]
\addplot [ thick,cyan!20!black,domain=0:17.9] {gammapdf(x,2,2)};
\addplot [fill=green!20, draw=none, domain=0:9] {gammapdf(x,2,2)} \closedcycle;
\addplot [very thick, fill=blue!20, draw=none, domain=9.01:18] {gammapdf(x,2,2)} \closedcycle;
\node[ font=\footnotesize, text width=3cm, align=left] at (axis cs: 4.5,.275)(t1) {\textit{Most likely, the DRL agent returns an action within time $\Delta$}};
\node[ font=\footnotesize, text width=3.2cm] at (axis cs: 14,.275)(t2) {\textit{With a non-negligible probability, the agent returns an action after $\Delta$}};
\node[] at (axis cs: 9,0.0)(m0) {};
\node[] at (axis cs: 0,0.35)(m1) {};
\node[] at (axis cs: 0,0.0)(m2) {};
\node[] at (axis cs: 9,0.35)(m3) {};
\node[] at (axis cs: 9,0.0)(m4) {};
\draw[dashed] (m3) -- (m4);
\end{axis}
\end{tikzpicture}
\caption{Perception-to-prediction time  \label{fig:time-distribution}}
\end{subfigure}
\caption{DRL Agent for Autonomous Driving}
\end{figure}
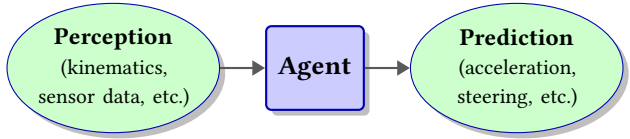
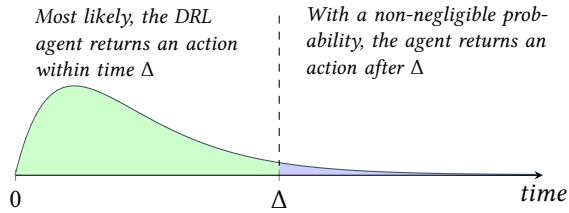

We consider an autonomous-driving system that is controlled by a DRL agent.
Typically, this DRL agent conducts two steps (c.f., Figure~\ref{fig:real-time-drl}). First, the agent perceives kinematics data such as positions, speeds and driving directions of the ego vehicle and  nearby vehicles. It then predicts control values (e.g., acceleration and steering values) for the ego vehicle.
In reality, the above two steps \rrev{occur} in real time; namely, when the agent performs the perception and prediction, the physical world  changes continuously.
In normal circumstances, these steps take a minimum time period, during which changes in the physical world are negligible.
However, \rrev{when} experiencing significant variations of the end-to-end time (also referred to as \emph{jitter}~\cite{Kopetz2022}),  the agent's decision making may be compromised. 
First, the jitter can result in performance degradation for agents trained in the normal conditions. These agents are very sensitive to the discrepancy between the normal and abnormal conditions.
Second, a delayed action can be a safety issue for deadline-driven systems. For example, if the ego vehicle is fast approaching another vehicle in the highway, an urgent action must be taken.

It is noteworthy that time variation in \rev{perception and prediction} is subject to \textit{aleatory} uncertainty (in contrast to epistemic uncertainty). In other words, the agent cannot eliminate this uncertainty even if all (including historical) data is available.  To illustrate this uncertainty, Figure~\ref{fig:time-distribution} depicts a typical long-tail distribution, illustrating a rare but non-negligible significant time delay.
It is possible to train an agent who can budget enough time to accommodate the uncertain end-to-end delay.
For example, as the exact amount of delay is unpredictable (due to aleatory uncertainty), an agent which controls the ego vehicle can provide excessive actions (e.g., low-speed \rrev{driving} and early deceleration) to avoid collisions in the abnormal conditions. However, this inevitably \rrev{sacrifices} the overall performance  (e.g., driving efficiency and comfortability) under normal conditions.

\subsection{Switch Architecture and Design Choices}\label{sec:control_switch_overview}\label{sec:determine_swtich_config}

In the traditional DRL paradigm, an agent interacts with the environment directly. 
However, as the previous section highlighted, this direct interaction in a real-time setting exposes the system to the risks of aleatory uncertainty and time delay. To address this, we advocate a long-standing philosophy in fault-tolerant system design, that is, \emph{redundancy}.

We assume that the DRL system \rrev{includes} \emph{alternative controllers} specifically engineered to be more robust on some, usually critical, performance aspects (e.g., safety and system stability) in the presence of time delay. 
These alternative controllers are often designed by more traditional methodologies, such as rule-based systems or model-predictive control. 
In our autonomous driving example, the DRL agent serves as the primary controller. It is trained to achieve smooth, energy-efficient driving and efficient lane changes, which are soft performance \rev{objectives}.
One or more rule-based controllers are designed to avoid lane departures, small headway and low time-to-collision, which are hard performance objectives.

Our focus is on the development of a rigorous approach to design a switching mechanism for coordinating \rev{these} controllers. 
This mechanism mediates the primary controller, alternative controllers and \rrev{the} environment (as illustrated in  Figure~\ref{fig:switch}).
Its core function is to decide \emph{how} and \emph{when} to switch between them in real-time, especially when the DRL agent experiences a delay.

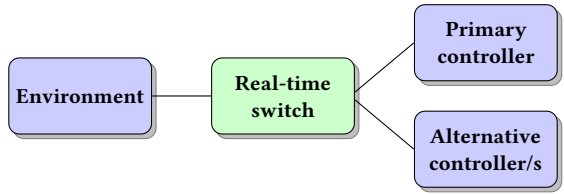
\begin{figure}[htb!]
    \centering
    {\footnotesize
\begin{tikzpicture}[
node distance=2.7cm,
every path/.style={>={Stealth},draw=black, text=black},
block/.style={rectangle,rounded corners,draw=black, fill=blue!20, text=black, inner sep=1mm, text width=1.7cm, font=\bfseries,  align=center, minimum height=1cm, drop shadow},
]
\def\bottom#1#2{\hbox{\vbox to #1{\vfill\hbox{#2}}}}
\node[block](env)[]{Environment};
\node[block, fill=green!20](sw)[right of=env, xshift=0cm, yshift=0cm] {Real-time switch};
\node[block](c0)[right of=sw, yshift=0.7cm]{Primary controller};
\node[block](c1)[right of=sw, yshift=-0.7cm]{Alternative controller/s};
\draw[-](env.east)-- (sw.west);
\draw[-] ([yshift=.05cm]sw.east)-- (c0.west);
\draw[-] ([yshift=-.05cm]sw.east)-- (c1.west);
\end{tikzpicture}
    \caption{Real-time Switch Architecture}
    \label{fig:switch}
    }
\end{figure}

One key issue in switch design is determining design choices that can meet the hard performance objectives and balance the soft performance objectives.
The switch designer can usually prescribe a space of all possible \emph{design choices}, which are essentially switch \emph{configurations} involving time-relevant parameters. 
These choices must ensure the system to meet its hard performance objectives and balance the soft ones simultaneously.
 For the autonomous car, this means \rev{guaranteeing} safety while allowing the DRL agent to maximise other performance aspects (such as driving efficiency and comfort).
To achieve these objectives, we explore two strategies of design choices:
\begin{itemize}
\item (Dynamic Strategy) 
The switch can dynamically determine configuration parameters, allowing the system to adapt its switching behaviour during operation.
\item (Randomised Strategy) The switch can determine configuration parameters probabilistically,  namely, drawing from a pre-defined set of possible parameter values according to a specified distribution.
\end{itemize}



\section{Formal Design Framework}\label{sec:architecture}\label{sec:approach}

This section presents our formal framework \rrev{of} real-time switch design.
Figure~\ref{fig:approach} illustrates the main steps of our approach and the corresponding subsections (Sections~\ref{sec:ta} to \ref{sec:formal_problem}). 


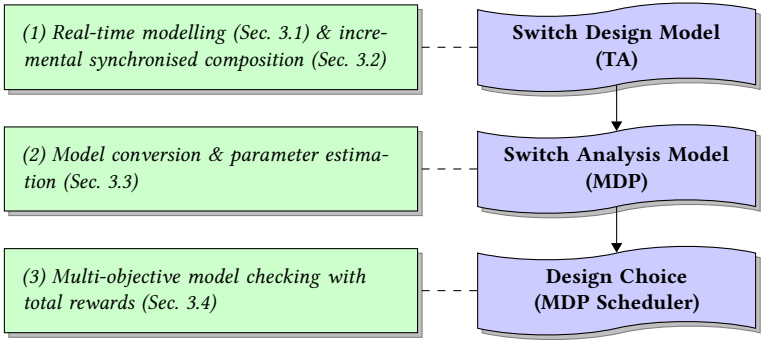
\begin{figure}[htb!]
\ifacm\newcommand{\btextwidth}{3.5cm}\fi
\ifacm\newcommand{\ptextwidth}{5cm}\fi
\ifacm\newcommand{\custxshift}{-3.8cm}\fi
    \centering
    {\footnotesize
\begin{tikzpicture}[
node distance=1.6cm,
every path/.style={draw=black, text=black},
block/.style={tape, draw=black, fill=blue!20, text=black, font=\bfseries, inner sep=1mm, text width=\btextwidth, align=center, minimum height=1cm, drop shadow},
process/.style={rectangle,draw=black, fill=green!20, text=black, font=\slshape,  inner sep=2.5mm, text width=\ptextwidth, align=left, minimum height=1cm, drop shadow},
]
\def\bottom#1#2{\hbox{\vbox to #1{\vfill\hbox{#2}}}}
\node[block](ta)[]{Switch Design Model \\  (TA)};
\node[process](modelling)[left of=ta, xshift=\custxshift ] {(1) Real-time modelling (Sec.~3.1)  \& incremental synchronised\ composition\   (Sec.~3.2)}; 
\node[block](mdp)[below of=ta, xshift=0cm] {\textbf{Switch Analysis Model }\\ (MDP) };
\node[process](conversion)[left of=mdp, xshift=\custxshift ] {(2) Model conversion\ \& parameter\ estimation  (Sec.~3.3)}; 
\node[block](objectives)[below of=mdp] {Design Choice \\ (MDP Scheduler)};
\node[process](specification)[left of=objectives, xshift=\custxshift ] {(3) Multi-objective model checking with total rewards (Sec.~3.4)}; 

\draw[-Triangle](ta.south)-- (mdp.north);
\draw[-Triangle](mdp.south)-- (objectives.north);
\draw[dashed](ta.west)--(modelling.east);
\draw[dashed](mdp.west)--(conversion.east);
\draw[dashed](objectives.west)--(specification.east);
\end{tikzpicture}
    \caption{Overview of the Switch Design Approach}
    \label{fig:approach}
    }
\end{figure}

\subsection{TA-Based Switch Modelling}\label{sec:ta}

While many existing approaches have employed discrete-time models for DRL verification and assurance \cite{Bacci2022,Koenighofer2022,Riley2022,Jansen2020}, our work specifically addresses the real-time setting.
For this purpose, we adopt {timed automata} (TAs) \cite{Alur1994} as the formal model for designing controller switches. 

A key concept in TAs is that of \emph{clocks}, which are symbolic real-valued variables representing the physical timers. Clocks increase uniformly with time and can be reset to $0$. 
For the scope of our work, we consider TAs that are \emph{augmented with Boolean variables}.
Let $B$ be a set of Boolean variables and $Z$ be a set of clocks. 
\emph{Guarding conditions}, or simply \emph{guards}, are defined by the following syntactic rules:
\begin{equation*}
\begin{aligned}
\varphi ::= & ~ \top ~|~ {b} \mid \lambda \leq \lambda \mid \neg \varphi \mid \varphi \wedge \varphi , \quad
\lambda ::=  z \mid c \\
\end{aligned}
\end{equation*}
where ${b}\in B$, $z\in Z$ and $c$ is a constant. Let ${G}$ denote the set of all such guards.
\begin{definition}[\rev{Timed} automata]
A \textit{\rev{timed} automaton} (TA) is a tuple 
$$\ta=(S, s_{0}, A,\mathit{inv}, \mathit{enb}, \mathit{tran}, \mathit{res})$$ where
\begin{itemize}
    \item $S$ is a non-empty set of states (often called locations);
    \item $s_{0}\in S$ is the initial state;
    \item $A$ is a non-empty set of actions; 
    \item $\mathit{inv}: S \mapsto {G}$ is an invariant function;
    \item $\mathit{enb} : S\times A \mapsto {G}$ is \rrev{an} {enabling} function that \rrev{maps} a state-action pair to a {guard};
     \item $\mathit{tran}: S \times A \mapsto S$ is a transition function that maps a state-action pair to its next state;    
     \item $\mathit{res}: S\times A  \mapsto 2^{{Z}}$ is a clock-resetting function that specifies a subset of clocks to be reset to 0 upon a transition.
\end{itemize}
\end{definition} 

The invariant function $\mathit{inv}$ provides conditions that must hold in a state. 
The enabling function $\mathit{enb}$ specifies the guarding condition for an action to be executable from a given state.
The transition function $\mathit{tran}$ determines the next state when performing an action. 
The clock-resetting function $\mathit{res}$ indicates which clocks are reset to $0$ as part of the transition.
If $\mathit{enb}(s,a)$ does not imply $\neg \top$ (false) (that is, enabling $a$ at $s$ is logically impossible), we say that $a$ is disabled at $s$.
Let $A(s)$ denote the subset of \emph{enabled}  actions at $s$.
In practice, only enabled actions are relevant.

For design modelling purposes, we introduce the following important assumption for TAs: 
{The guards for enabled actions at each state are either mutually exclusive or logically equivalent}.  
Formally, this assumption is expressed as follows:
\begin{assumption}\label{thm:assump}
For any $s\in S$ such that $|A(s)|>1$, one of the following two conditions holds:
\begin{enumerate}
\item $\forall a,a'\in A(s)$ such that $a\neq a'$, $\mathit{enb}(s,a) $ and  $ \mathit{enb}(s,a')$ are mutually exclusive. States satisfying this proposition are called \emph{branching states}. Let $S_\mathit{BR}\subseteq S$ denote the set of branching states.
\item $\forall a,a'\in A(s)$, $\mathit{enb}(s,a) $ and $\mathit{enb}(s,a')$ are logically equivalent. 
States satisfying this proposition are called \emph{choice states}. Let $S_\mathit{CH}\subseteq S$ denote the set of choice states.
\end{enumerate}
\end{assumption}

Assumption~\ref{thm:assump} does not impose a theoretical constraint, as any TA can be transformed to an equivalent TA that meets this assumption. However, its practical merit in system design is {significant} because it clearly distinguishes between two types of non-deterministic behaviours.
\begin{itemize}
\item Transitions from branching states represent underdetermination arising from model abstraction. These underdetermined behaviours will be resolved in system execution (when the exact values of clocks and Boolean variables are known).
\item Actions associated with \rev{choice} states represent open, non-deterministic design choices, which must be settled by the model designer.
\end{itemize}

When representing \rrev{a} TA, it is convenient to specify its \emph{edges}.
An edge in TA is denoted $s \overset{\varphi: a,C }{\longrightarrow} s'$, which means ``$\mathit{enb}(s,a) = \varphi$, $ \mathit{tran}(s,a) =s'$, and $ \mathit{res}(s, a) = C$'' where $C\subseteq Z$ (a subset of clocks).
We further simplify edge labels in common special cases:
\begin{itemize}
\item  $s\overset{\varphi:a}{\longrightarrow}s'$ abbreviates $s\overset{\varphi:a,\emptyset}{\longrightarrow}s'$ (no clock reset).
\item  $s\overset{a}{\longrightarrow}s'$ abbreviates $s\overset{\top:a}{\longrightarrow}s'$ (always enabled action $a$, no clock reset).
\item $s\overset{C}{\longrightarrow}s'$ abbreviates $s\overset{\top:\tau, C}{\longrightarrow}s'$   (always enabled anonymous action).
\end{itemize}

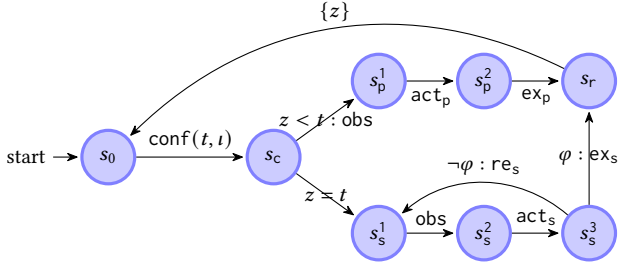
\begin{figure}[!thbp]
    \footnotesize\centering
\begin{tikzpicture}[shorten >=1pt,node distance=1.4cm,on grid,>={Stealth[round]},
    actionnode/.style={circle, draw, minimum size=1mm},
    every state/.style={draw=blue!50,very thick,fill=blue!20, inner sep=0cm}]
\node[initial left, state] (s_0){$s_0$};
\node[state] (s_c) [right=of s_0, xshift=.8cm] {$s_\mathtt{c}$};
\node[state] (s_p^m) [right=of s_c, yshift=1cm] {$s_\mathtt{p}^\mathtt{1}$};
\node[state] (s_s) [right=of s_c, yshift=-1cm] {$s_\mathtt{s}^\mathtt{1}$};
\node[state] (s_p^a) [right=of s_p^m] {$s_\mathtt{p}^\mathtt{2}$};
\node[state] (s_s^m) [right=of s_s] {$s_\mathtt{s}^\mathtt{2}$};
\node[state] (s_s^a) [right=of s_s^m] {$s_\mathtt{s}^\mathtt{3}$};
\node[state] (s_r) [right=of s_p^a] {$s_\mathtt{r}$};
\path[->] 
(s_0) edge node[above]{$\mathtt{conf}(t, \iota)$ } (s_c) 
(s_c) edge node[]{$z<t:\mathtt{obs}$} (s_p^m) edge node[]{$z = t$} (s_s)
(s_p^m) edge node[below]{$\mathtt{act_p}$} (s_p^a)
(s_p^a) edge node[below] {$\mathtt{ex_p}$} (s_r)
(s_s) edge node[above]{$\mathtt{obs}$} (s_s^m)
(s_r) edge [bend right = 35] node[above]{$\{z\}$} (s_0)
(s_s^m) edge node[above]{$\mathtt{act_s}$} (s_s^a)
(s_s^a) edge [bend right = 40] node[above]{$\neg \varphi:\mathtt{re_s}$} (s_s) edge node[]{$\varphi:\mathtt{ex_s}$} (s_r)
;
\end{tikzpicture}
\caption{$\ta_{\mathit{sw}0}$: Base Model of a Real-time Switch}\label{fig:ta_sw}
\end{figure}

Figure~\ref{fig:ta_sw} depicts an exemplar TA $\ta_{\mathit{SW}0}$ which represents a base model of a real-time switch between a primary controller (e.g., a DRL agent) and one or more secondary controllers (e.g., rule-based controllers). 
For clarity, an abstract action $\mathtt{conf}(t,\iota)$ is used to represent non-deterministic design choices, which intuitively means 
``once the time delay exceeds $t$, switch to the alternative controller $\iota$.''
%
The initial state $s_0$ is the only choice state of $\ta_{\mathit{SW}0}$, where the switch follows a concrete design choice and transitions to $s_\mathtt{c}$.
From $s_\mathtt{c}$, the switch waits for an observation. If an observation ($\mathtt{obs}$) is received before the clock $z$ reaches $t$  (namely $z<t$), the switch transitions to $s_\mathtt{p}^1$ (the primary path). However, if the observation is not received when $z=t$, it immediately moves to $s_\mathtt{s}^1$ (the secondary path).
At $s_\mathtt{p}^1$, the switch calls the DRL agent to obtain an action ($\mathtt{act_p}$) and then exits to $s_\mathtt{r}$. 
At $s_\mathtt{s}^1$, upon receiving an observation ($\mathtt{obs}$), it calls the secondary controller and obtains an action (\rev{$\mathtt{act_s}$}).
Then, depending on a Boolean condition $\varphi$ (which may represent the road condition), it exits (with transition label ``$\varphi:\mathtt{ex_s}$'') to $s_\mathtt{r}$ or returns (with transition label ``$\neg \varphi:\mathtt{re_s}$'') to $s_\mathtt{s}^1$.
Finally, from $s_\mathtt{s}^1$, the switch returns to the initial state $s_0$, resetting the clock $z$ to zero.

\subsection{Incremental Design by Synchronised Composition}\label{sec:syn_comp}

Incremental design is a crucial concept in model-driven development. 
For TAs, incremental design can be achieved via the standard synchronised composition of TAs.
\begin{definition}[Synchronised composition]
Given $\ta_i =(S_i, s_{i0}, A_i, \mathit{inv}_i, \mathit{enb}_i, \mathit{tran}_i, \mathit{res}_i) $ where $i\in\{1,2\}$, the {synchronised composition} $\ta_1\| \ta_2$ is the tuple
$$\big( S_{1}\times S_{2}, (s_{10}, s_{20}), A_1\cup A_2, \mathit{enb}_{1,2}, \mathit{tran}_{1,2}, \mathit{res}_{1,2}\big) $$ where $\rev{\mathit{inv}_{1,2} (s_1.s_2)} = \mathit{inv}_{1} (s_1)\wedge \mathit{inv}_{2} (s_2)$, and  $\mathit{enb}_{1,2}$, $\mathit{tran}_{1,2}$ and $\mathit{res}_{1,2}$ are formalised as follows: (i) If $a\in A_1\cap A_2$ (common action), then
\begin{itemize}
\item $\mathit{enb}_{1,2} (s_1, s_2, a) = \mathit{enb}_1(s_1, a) \wedge \mathit{enb}_2( s_2, a)$;
\item $\mathit{tran}_{1,2} (s_1, s_2, a) = (\mathit{tran}_1(s_1, a) , \mathit{tran}_2( s_2, a))$;
\item $\mathit{res}_{1,2} (s_1, s_2, a) = \mathit{res}_1(s_1, a) \cup \mathit{res}_2( s_2, a)$.
\end{itemize}
(ii) If $a\in A_1\backslash A_2$ (action unique to $\ta_1$), then
\begin{itemize}
\item $\mathit{enb}_{1,2} (s_1, s_2, a) = \mathit{enb}_1(s_1, a) $;
\item $\mathit{tran}_{1,2} (s_1, s_2, a) = (\mathit{tran}_1(s_1, a) , s_2)$;
\item $\mathit{res}_{1,2} (s_1, s_2, a) = \mathit{res}_1(s_1, a)$.
\end{itemize}
(iii) Otherwise (action unique to $\ta_2$),
\begin{itemize}
\item $\mathit{enb}_{1,2} (s_1, s_2, a) = \mathit{enb}_2(s_2, a) $;
\item $\mathit{tran}_{1,2} (s_1, s_2, a) = (s_1, \mathit{tran}_2(s_2, a))$;
\item $\mathit{res}_{1,2} (s_1, s_2, a) = \mathit{res}_2(s_2, a)$.
\end{itemize}
\end{definition}

\begin{figure}[!thbp]
    \footnotesize\centering
\begin{subfigure}[b]{\subfigtextwidth}
\centering
\begin{tikzpicture}[shorten >=1pt,node distance=2.29cm,on grid,>={Stealth[round]},
    actionnode/.style={circle, draw, minimum size=1mm},
    every state/.style={draw=blue!50,very thick,fill=blue!20, inner sep=0cm}]
\node[initial, state] (l_0){$s_{\mathtt{ts}}^0$};
\node[state] (l_1) [right of=l_0] {$s_{\mathtt{ts}}^1$};
\node[] (cdots)[right of=l_1] {$\cdots$};
\node[state] (l_12) [right of=cdots] {$s_{\mathtt{ts}}^n$};
\path[->] 
(l_0) edge node[above]{{$\mathtt{act_p}$ or $\mathtt{act_s}$}} (l_1) 
(l_1) edge node[above]{{$\mathtt{act_p}$ or $\mathtt{act_s}$}}(cdots) 
(cdots) edge node[above]{{$\mathtt{act_p}$ or $\mathtt{act_s}$}} (l_12)
;
\end{tikzpicture}
\caption{${\ta}_\mathit{TS}$: Timestep Counter ($n$ the maximum timestep)\label{fig:ra1}}
\end{subfigure}
\newline
\begin{subfigure}[b]{\subfigtextwidth}
\centering
\begin{tikzpicture}[shorten >=1pt,node distance=1.4cm,on grid,>={Stealth[round]},
    actionnode/.style={circle, draw, minimum size=1mm},
    every state/.style={draw=blue!50,very thick,fill=blue!20, inner sep=0cm}]
\node[initial, state] (l_0){$s_{\mathtt{cc}}^0$};
\node[state] (l_1) [right of=l_0] {$s_{\mathtt{cc}}^1$};
\node[] (cdots)[right of=l_1] {$\cdots$};
\node[state] (l_12) [right of=cdots] {$s_{\mathtt{cc}}^k$};
\path[->] 
(l_0) edge node[above]{$\mathtt{re_s}$} (l_1) 
(l_1) edge node[above]{$\mathtt{re_s}$} (cdots) 
(cdots) edge node[above]{$\mathtt{re_s}$} (l_12)
(l_12) edge[loop right] node[right]{$\mathtt{re_s}$ or $\mathtt{ex_s}$}(l_12)
;
\end{tikzpicture}
\caption{${\ta}_\mathit{CC}$: Consecutive Call Counter ($k$ the minimum number of calls)  \label{fig:ra2}}
\end{subfigure}
\newline
\begin{subfigure}[b]{\subfigtextwidth}
\centering
\begin{tikzpicture}[shorten >=1pt,node distance=1.8cm,on grid,>={Stealth[round]},
    actionnode/.style={circle, draw, minimum size=1mm},
    every state/.style={draw=blue!50,very thick,fill=blue!20, inner sep=0cm}]
\node[initial, state] (l_0){$s_{\mathtt{mr}}^0$};
\node[state] (l_00)[right of=l_0]{$s_{\mathtt{mr}}^1$};
\node[state] (l_1) [right of=l_00,  yshift=1cm] {$s_{\mathtt{mr}}^\mathtt{yes}$};
\node[state] (l_2) [right of=l_00, yshift=-1cm] {$s_{\mathtt{mr}}^\mathtt{no}$};
\path[->] 
(l_0) edge node[above] {$\mathtt{obs}$}(l_00)
(l_00) edge [bend right = 0] node[below]{$\varphi_{\mathtt{mr}}$} (l_1) edge [bend left = 0] node[above]{$\neg \varphi_{\mathtt{mr}}$} (l_2)
(l_1) edge [bend right = 10]  node[above]{$\mathtt{act_p}$ or $\mathtt{act_s}$ } (l_0)
(l_2) edge [bend left = 10]  node[below]{$\mathtt{act_p}$ or $\mathtt{act_s}$} (l_0)
;
\end{tikzpicture}
\caption{${\ta}_\mathit{MR}$: Metric Recorder ($\varphi_{\mathtt{mr}}$ a proposition of the metric) 
\label{fig:ra3}}
\end{subfigure}
\caption{TAs for  Synchronised Composition with $\ta_{\mathit{SW}0}$}
\label{fig:ta_examples}
\end{figure}
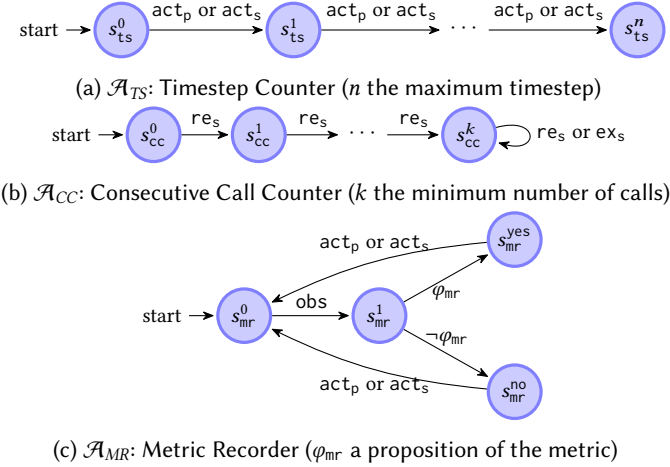

Figure~\ref{fig:ta_examples} presents three additional TAs. 
$\ta_\mathit{TS}$ (Figure~\ref{fig:ra1}) models a timestep counter, representing the progression of simulation steps up to a maximum 
$n$. $\ta_\mathit{CC}$ (Figure~\ref{fig:ra2}) models a consecutive call counter, tracking the minimum number of consecutive calls $k$ to the secondary controller.
$\ta_\mathit{MR}$ (Figure~\ref{fig:ra3}) acts as a metric recorder, capturing whether a proposition $\varphi_{\mathtt{mr}}$ of the metric is satisfied or not.
\rev{In practice, $\mathit{MR}$ can be any custom metric. For example, if $\mathit{MR}$ is a continuous variable, then $\varphi_{\mathtt{mr}}$ can be the assertion $\mathit{MR}\leq c$ where $c$ is a prescribed constant.}
These TAs can be composed with the base model $\ta_{\mathit{SW}0}$ to refine their behaviour by introducing additional statefulness and conditions related to some metrics. 
The synchronised composition $\ta_{\mathit{SW}1}=\ta_{\mathit{SW}0} \| \ta_{\mathit{TS}}  \| \ta_{\mathit{CC}}\|\ta_{\mathit{MR}}$ is the resulted incremental design model for the switch.
\rev{According to the semantics of synchronised composition, if $\ta_{\mathit{SW}1}$ performs $\mathtt{ex_p}$, $\{z\}$ or $\tau$, it can run on its own. But if it performs any other action, it must synchronise with other TAs whose action sets contain the same action.}

\subsection{Syntactic-Converted MDP as Analysis Model}\label{sec:syntactic_mdp}

We first recall the definition of MDPs, which serve as the formal analysis model \orev{in our approach}.
\begin{definition}[Markov decision process]
A Markov decision process (MDP) is a tuple $\mdp = (S', s_0, A', \mathit{prob})$ where 
\begin{itemize}
\item $S'$ is a non-empty set of states;
\item $s_0\in S'$ is the initial state;
\item $A'$ is a non-empty set of actions;
\item $ \mathit{prob}:S'\times A'\times S '\to [0,1]$ is a {probabilistic transition function} s.t.\ 
$ \sum_{s'\in S'} \mathit{prob}(s,a,s')\in \{0,1\}$.
\end{itemize}
\end{definition}

Similar to $\ta$, we define the set of enabled actions $A'(s)$ at a state $s$ for $\mdp$ as a subset of $A'$ such that $a\in A'(s)$ if and only if $\mathit{prob}(s,a,s')>0$ for some $s'$.  
Randomised dynamic design choices for the switch (c.f., Section~\ref{sec:determine_swtich_config}) are formalised as (memoryless) schedulers in $\mdp$. 
\begin{definition}[Scheduler]
A (memoryless) \emph{scheduler} (aka.\ \emph{policy}) for $\mdp$ is the function $\pi$ that maps each $s\in S'$ to a probability distribution over $A'(s)$, namely, $\pi(s)(a) \geq 0$ for all $a\in A'(s)$ and $\sum_{a\in A'(s)} \pi(s)(a) = 1$.
Let $\mathit{Sch}(\mdp)$ denote \emph{the set of (memoryless) schedulers} for $\mdp$.
\end{definition}

The clear distinction between underdetermined enabling conditions and non-deterministic actions allows an \textit{ad hoc} conversion of $\ta$ (design model) into $\mdp$ (analysis model). In the conversion, we use an anonymous action $\tau$ to mask actions at non-choice states, and replace all enabling conditions by (underdetermined) transition probabilities. It is noteworthy that this conversion differs from the standard TA semantics, in which $\ta$ is interpreted as an (infinite) transition system (with clock valuations as states).

We interpret each guard at a branching state as a random variable.
We use ${Pr}(\varphi| s)$ where $s\in S_\mathit{BR}$ (a branching state) and $\varphi=\mathit{enb}(s,a)$ for some $a$ to denote the \emph{probability} that $\varphi$ is satisfied (and thus $a$ is enabled).
Given a TA $\ta=(S, s_{0}, A, \mathit{enb}, \mathit{tran}, \mathit{res})$, a \emph{syntactic-converted MDP} $\mdp = (S', s_0, A', \mathit{prob})$ is generated by the following procedure:
\begin{itemize}
\item Let $S' = S$ and $A' = \emptyset$ initially. 
\item For all $s\in S$:
	\begin{itemize}
	\item If $s\in S_\mathit{CH}$  (choice state) or $|A(s)| =1$, then:
		\begin{itemize}
		\item Let $A'(s) = A(s)$; 
		\item Let $\mathit{prob}(s, a, s') = 1$ for all $a\in A(s)$ and $s'=\mathit{tran}(s,a)$;
		\end{itemize}	
	\item If $s\in S_\mathit{BR}$ (branching state), then for each $ a\in A(s)$: 
		\begin{itemize}
		\item Add a \emph{new} state $\tilde{s}$ to $S'$, $\tau$ to $A'(s)$, and $a$ to $A'(\tilde{s})$;
		\item Let $\mathit{prob}(s, \tau, \tilde{s}) = \mathit{Pr}(\varphi|s)$ 
		where $\varphi= \mathit{enb}(s,a)$;
		\item Let $\mathit{prob}(\tilde{s}, a, s') =1$ where $s'=\mathit{tran}(s,a)$.
		\end{itemize}
	\end{itemize}
\end{itemize}
Figure~\ref{fig:mdp_sw} depicts a syntactic-converted MDP $\mdp_{\mathit{SW}0}$ for the base TA model $\ta_{\mathit{SW}0}$ (where 
$\mathtt{conf}(t, \iota)$ represents multiple non-deterministic actions). \rev{Note that three new states $\tilde{s}_\mathtt{p}^1$, $\tilde{s}_\mathtt{s}^1$ and $\tilde{s}_\mathtt{r}$ are added to $\mdp_{\mathit{SW}0}$ in the conversion process.}

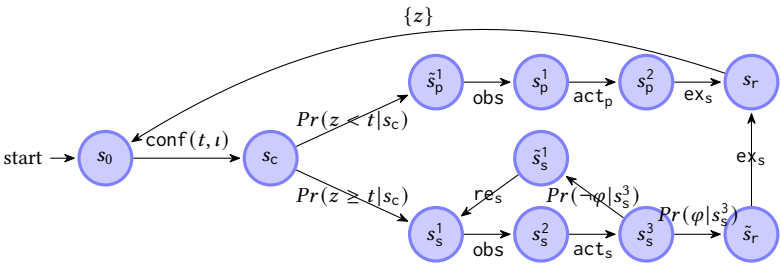
\begin{figure}[!thbp]
    \footnotesize\centering
\begin{tikzpicture}[shorten >=1pt,node distance=1.4cm,on grid,>={Stealth[round]},
    actionnode/.style={circle, draw, minimum size=1mm},
    every state/.style={draw=blue!50,very thick,fill=blue!20, inner sep=0cm}]
\node[initial left, state] (s_0){$s_0$};
\node[state] (s_c) [right=of s_0, xshift=.8cm] {$s_\mathtt{c}$};
\node[state] (s_p^11) [right=of s_c, yshift=1cm,  xshift=.8cm] {$\tilde{s}_\mathtt{p}^\mathtt{1}$};
\node[state] (s_p^1) [right=of {s_p^11}] {$s_\mathtt{p}^\mathtt{1}$};
\node[state] (s_s^1) [right=of s_c, yshift=-1cm, xshift=.8cm] {$s_\mathtt{s}^\mathtt{1}$};
\node[state] (s_s^11) [right=of s_s^1, yshift=1cm] {$\tilde{s}_\mathtt{s}^\mathtt{1}$};
\node[state] (s_p^2) [right=of s_p^1] {$s_\mathtt{p}^\mathtt{2}$};
\node[state] (s_s^2) [right=of s_s^1] {$s_\mathtt{s}^\mathtt{2}$};
\node[state] (s_s^3) [right=of s_s^2] {$s_\mathtt{s}^\mathtt{3}$};
\node[state] (s_rr) [right=of s_s^3] {$\tilde{s}_\mathtt{r}$};
\node[state] (s_r) [right=of s_p^2] {$s_\mathtt{r}$};
\path[->] 
(s_0) edge node[above]{$\mathtt{conf}(t, \iota)$ } (s_c) 
(s_c) edge node[]{${Pr}(z<t|s_\mathtt{c})$} (s_p^11) edge node[]{${Pr}(z\geq t|s_\mathtt{c})$} (s_s^1)
(s_p^11) edge node[below]{$\mathtt{obs}$} (s_p^1)
(s_p^1) edge node[below]{$\mathtt{act_p}$} (s_p^2)
(s_p^2) edge node[below] {$\mathtt{ex_s}$} (s_r)
(s_s^1) edge node[below]{$\mathtt{obs}$} (s_s^2)
(s_r) edge [bend right = 25] node[above]{$\{z\}$} (s_0)
(s_s^2) edge node[below]{$\mathtt{act_s}$} (s_s^3)
(s_s^3) edge [] node[]{${Pr}(\neg \varphi|s_\mathtt{s}^3)$} (s_s^11) 
(s_s^11) edge node[]{$\mathtt{re_s}$} (s_s^1)
(s_s^3) edge node[above]{${Pr}(\varphi|s_\mathtt{s}^\mathtt{3})$} (s_rr)
(s_rr) edge node[]{$\mathtt{ex_s}$} (s_r)
;
\end{tikzpicture}
\caption{$\mdp_{\mathit{SW}0}$: Syntactic-converted MDP of $\ta_{\mathit{SW}0}$}\label{fig:mdp_sw}
\end{figure}

The syntactic conversion does not resolve transition probabilities for branching states.
Those transition probabilities are probabilities that the corresponding guards (i.e., clock-based and Boolean conditions) are enabled.
For example, in $\mdp_{\mathit{SW}0}$, the probability of transitioning from $s_\mathtt{c}$ to $\tilde{s}_\mathtt{p}^1$ is ${Pr}(z< t|s_\mathtt{c})$, and that of transitioning from $s_\mathtt{s}^3$ to $\tilde{s}_\mathtt{r}$ is ${Pr}(\varphi | s_\mathtt{s}^3)$.
To determine them, we adopt the following two simple estimation methods:
\begin{enumerate}
\item Probabilities can be estimated by assuming a specific probability distribution. For example, if delay times have a known probability distribution, ${Pr}(z< t|s_\mathtt{c})$ and ${Pr}(z\geq t|s_\mathtt{c})$ in $\mdp_{\mathit{SW}0}$ can be calculated.
\item Statistics can be inferred by simulating the system in an environment and counting transition frequencies. For example, observing whether the exit condition $b$ in $\ta_{\mathit{SW}0}$ is satisfied can estimate ${Pr}(\varphi|s_\mathtt{s}^3)$ in $\mdp_{\mathit{SW}0}$ based on transition frequencies from $s_\mathtt{s}^3$ to  $\tilde{s}_\mathtt{s}^1$ and $\tilde{s}_\mathtt{r}$.
\end{enumerate}
Discussion of more complex methods for learning MDPs from sampling processes (such as \cite{li2006towards,allen2021learning}) is beyond the scope of this paper.

\rev{We employ TAs rather than MDPs as primary switch design models because of two reasons. First, MDPs lack the notion of clocks, which is essential for expressing conditions such as timed guards and timeouts in real-time system design. Second, the transition matrix of an MDP can only be determined after the design model has been constructed (including the separation of choice and branching states), rendering it unsuitable for the initial design.}

\subsection{Design Choice Analysis}\label{sec:formal_problem}

In our approach, the performance objectives of a TA-based switch design model are expressed by using rewards within its syntactic-converted MDP.
\begin{definition}[Reward function]\label{def:reward_function}
A \emph{reward function} for $\mdp = (S', s_0, A', \mathit{prob})$ is a function \rev{$\rho: S'\times A'\to \mathbb{Q}^+$} where $S'$ is the state space of $\mdp$. Let $\vect{\rho}= (\rho_1,\ldots,\rho_m)$ denote a \emph{reward vector function} for $\mdp$ for any $m\in \mathbb{N}$.
\end{definition}

Our primary goal of analysing $\mdp$  lies in finding a scheduler to yield optimal rewards with respect to different $\rho$.
Without loss of generality, we aim to compute optimal schedulers in \emph{multi-objective MDPs}. 
Formally, under a given scheduler $\pi$, starting from the initial state $s_0$, a trajectory (i.e., an alternating sequence of states and actions) $s_0a_0s_1a_1s_2a_2\ldots$ can be generated randomly, which yields a reward sequence \rev{$\rho(s_0,a_0)\rho(s_1,a_1)\rho(s_2,a_2)\ldots$}
\begin{definition}[Total reward]\label{def:total_reward}
Given a scheduler $\schr$ and a reward function $\rho$ for $\mdp$,  the \emph{(expected) total reward}, denoted $\tolrew{\rho}{\schr}$,  is 
the {expected sum} of these (random) rewards, \rev{namely $\rho(s_0,a_0)+\rho(s_1,a_1)+\rho(s_2,a_2)+\ldots$} Given a reward vector function $\vect{\rho}$, let $\tolrew{\vect{\rho}}{\schr} =(\tolrew{\rho_1}{\schr}\ldots, \tolrew{\rho_m}{\schr})$ denote a \emph{total reward vector}.\footnote{Rigorously speaking, as no discount factor is considered, we need to assume the existence of these total rewards (aka \emph{reward finiteness}). This assumption can be verified using the standard reachability analysis in model checking.}
\end{definition}

\rev{Total rewards formalise the performance objectives for switch analysis. In practice, total rewards can characterise a broad range of intuitive performance metrics. }
Let $\mdp_{\mathit{SW}1} $ be the syntactic-converted MDP of 
$\ta_{\mathit{SW}1}=\ta_{\mathit{SW}0} \| \ta_{\mathit{TS}}  \| \ta_{\mathit{CC}}\|\ta_{\mathit{MR}}$.
Two examples of reward functions for $\mdp_{\mathit{SW}1} $ are defined as follows:
\begin{itemize}
\item $\rho_{\mathit{CTL}}$ for $\ta_{\mathit{SW}0}$ that assigns $1$ to $s_{\mathtt{p}}^1$ (entry to the primary path), $2$ to $s_{\mathtt{s}}^1$  (entry to the secondary path ) and $0$ to all other states. This intuitively specifies that ``the cost of running a secondary controller is twice the cost for the primary controller''.
\item  $\rho_\mathit{MR}$ for $\ta_{\mathit{MR}}$ that assigns $1$ to $s_{\mathtt{mr}}^\mathtt{yes}$ (metric satisfied) and $0$ to all other states, thereby tracking \rev{``satisfaction of the proposition $\varphi_\mathtt{mr}$''}.
\end{itemize}
\rev{

We now explain how total rewards characterise three basic categories of performance metrics (i.e., variables): Boolean, discrete and continuous variables. 
\begin{itemize}
\item (\emph{Expected value for discrete metric}) The controller cost (CTL) is a discrete metric. The total reward  for $\rho_{\mathit{CTL}}$ under $\schr$ (denoted $\tolrew{\mathit{CTL}}{\schr}$) means the \emph{expected} controller cost.
\item (\emph{Probability for Boolean metric}) Suppose MR is a Boolean variable such as ``on lane''. We define that  $\varphi_\mathtt{mr}$ is true if and only if the vehicle is on lane.  The total reward  for $\rho_{\mathit{MR}}$ under $\schr$ (denoted $\tolrew{\mathit{MR}}{\schr}$)  means the on-lane \emph{probability}.
\item (\emph{Percentile rank for continuous metric}) Suppose MR is a continuous variable such as ``time to collision'' (TTC). We specify a TTC value $c$, and define that $\varphi_\mathtt{mr}$ is true if and only if the TTC is not larger than $q$. In this way, the total reward $\tolrew{\mathit{MR}}{\schr}$ means the TTC's \emph{percentile rank of score} $q$;  in other words, $q$ is the $\tolrew{\mathit{MR}}{\schr}$-th percentile of TTC. Moreover, if we specify multiple scores $q_1,\ldots q_k$, then the corresponding total rewards $\tolrew{\mathit{MR1}}{\schr},\ldots,\tolrew{\mathit{MRk}}{\schr}$ (as percentile ranks) approximate the continuous distribution of TTC.
\end{itemize}
}
Leveraging the existing methods of multi-objective model checking, we can conduct the following three kinds of analysis of total rewards:
\begin{itemize}
\item (\emph{Direct evaluation}) Given a pre-defined scheduler $\schr$, compute the total reward vector $\tolrew{\vect{\rho}}{\schr}$.
\item (\emph{Single-objective optimisation}) For each $\rho_i$, find an optimal scheduler $\schr$ that minimises $\tolrew{\rho_i}{\schr}$.
\item (\emph{Achievability query}) Given a threshold vector $(v_1, \ldots, v_m)$, find a scheduler $\schr$ which can achieve $\tolrew{\rho_1}{\schr}\leq v_1$, \ldots and $\tolrew{\rho_m}{\schr}\leq v_m$ simultaneously.
\end{itemize}
In addition, we introduce a new kind of analysis, called \emph{convex query}.
Let $\los: \mathbb{R}^m \mapsto \mathbb{R}$ be a convex function that characterises the \emph{trade-off} of objectives, namely, $\los(\tolrew{\vect{\rho}}{\schr})$ formalises their overall cost under a given scheduler $\schr$. Moreover, let $[\vect{l}, \vect{u}]$ be an $m$-dimensional (possibly unbounded) box that specifies the \emph{hard constraint} for these objectives.
The convex query is expressed as follows:
\begin{equation}\label{eq:ccq0}
\begin{aligned}
& \textit{Find an optimal scheduler $\schr^*\in\mathit{Sch}(\mdp)$}\\
& \quad \textit{ such that $\tolrew{\vect{\rho}}{\schr^*}\in [\vect{l}, \vect{u}]$ and  $\los(\tolrew{\vect{\rho}}{\schr^*})$ is minimised.}
\end{aligned}
\end{equation}
For example, for $\mdp_{\mathit{SW}1}$ we can consider the following convex query:
\begin{equation}\label{eq:ccq-sw}
\begin{aligned}
& \textit{Find an optimal scheduler $\schr^*\in\mathit{Sch}(\mdp_{\mathit{SW}1})$}\\
& \quad \textit{ such that $\tolrew{{\rho_{\mathit{CTL}}}}{\schr^*} \leq e_1$, $ \tolrew{{\rho_{\mathit{MR}}}}{\schr^*}\leq e_2$} \textit{ and $(\tolrew{{\rho_{\mathit{CTL}}}}{\schr^*}-d_1)^2 + (\tolrew{{\rho_{\mathit{MR}}}}{\schr^*}-d_2)^2$ is minimised.}
\end{aligned}
\end{equation}
\rev{Informally speaking, Query~\eqref{eq:ccq-sw} means:
\emph{Finding an optimal scheduler (i.e.\ design configuration)  such that (i) both the expected controller cost and the MR probability must not exceed their hard thresholds ($e_1, e_2$) and (ii) the deviation from their soft thresholds ($d_1, d_2$) are minimised.
}}

Our framework provides an expressive mechanism to reason about performance objectives for real-time switch design.
The formal technique to solve Problem~\eqref{eq:ccq0} is detailed in Section~\ref{sec:cq}.

\section{Convex Query on Multi-Objective MDP}\label{sec:cq}
In this section, we first establish the necessary geometric definitions for our technical approach (Section~\ref{sec:geometric}) and then review the basic querying methods for multi-objective MDPs (Section~\ref{sec:existing_queries}).
Building on this, we formulate the convex query problem for multi-objective MDPs and present our novel algorithm to solve this problem (Section~\ref{sec:convex_query}). We then describe an iterative procedure to handle a critical step in this algorithm (Section~\ref{sec:parallel_pvi}). \rev{Finally, we briefly describe the construction of random schedulers (Section~\ref{sec:scheduler_generation}).} For readability, the formal proofs of our main theorems 
are provided in Appendix~\ref{sec:proofs}.

\subsection{Geometric Preliminaries}\label{sec:geometric}
\rev{Informally, a convex query of total rewards on multi-objective MDPs finds an optimal value within a feasible solution set. This solution set is usually a polytope in a high-dimensional Euclidean space (where the dimension equals the number of objectives). Below we present  geometric definitions and fundamental properties of polytopes which underpin our convex query method.}

For a point (i.e., vector) $\vect{v}
\in \mathbb{R}^n$  for some $n$, let $v_i$ denote the $i^\mathrm{th}$ element of $\vect{v}$.
A \emph{weight vector} $\vect{w}$ is a vector such that $w_i\geq 0$ and $\sum_{i=1}^n w_i=1$.
The \emph{dot product} of $\vect{v}$ and $\vect{u}$, denoted $\vect{v}\cdot \vect{u}$, is the sum $\sum_{i=1}^n v_iu_i$.
For a set $\Phi=\{\vect{v}_1, \ldots, \vect{v}_m\} \subseteq \mathbb{R}^n$,
a \emph{convex combination} in $\Phi$ is $\sum_{i=1}^m w_i \vect{v}_i$ for some weight vector $\vect{w}\in \mathbb{R}^m$.
The set of convex combinations of vectors in $\Phi$ is a \textit{polytope}, denoted $\mathbf{V}(\Phi)$ (i.e., V-representation of polytopes). 
We say $\vect{u}$ \emph{dominates (resp., strictly dominates) $\vect{u}'$ in the direction of $\vect{v}$} if $\vect{u}\cdot \vect{v}\geq \vect{u}'\cdot \vect{v}$ (resp., $\vect{u}\cdot \vect{v}>\vect{u}'\cdot \vect{v}$).
Let $\Psi\subseteq \mathbb{R}^n$. 
A vector $\vect{u}\in \Psi$ is \emph{Pareto optimal} if there is a direction such that $\vect{u}$ dominates all points in $\Psi$.
A \emph{Pareto curve} in $\Psi$ is the set of Pareto optimal vectors in $\Psi$. 
Given a set of point-direction pairs $\Lambda=\{(\vect{u}_i, \vect{v}_i)\}_{i=1}^k$ where $\vect{u}_i, \vect{v}_i\in \mathbb{R}^n$ for all $i$, we can also define a polytope $\mathbf{H}(\Lambda)=\{\vect{x}\in \mathbb{R}^n \mid \vect{u}_i\cdot\vect{v}_i\geq \vect{x}\cdot \vect{v}_i, 1\leq i\leq k\}$ (i.e., H-representation of polytopes). 
For convenience, let $\mathbf{V}(\emptyset)=\emptyset$ and $\mathbf{H}(\emptyset)=\mathbb{R}^n$.
%
\begin{lemma}[Separating and supporting hyperplane theorems \cite{boyd2004convex}]\label{lem:hyperplane}
    Let $\Psi\subseteq \mathbb{R}^n$ be a convex set of points. For any $\vect{v}\not\in \Psi$, there is a weight vector $\vect{w}$ such that $\vect{w}\cdot\vect{v}> \vect{w}\cdot\vect{x}$ for all $\vect{x}\in \Psi$. We say that $\vect{w}$ \emph{separates} $\vect{v}$ from $\Psi$.
    Also, for any $\vect{u}$ on the Pareto curve of $\Psi$, there is a weight vector $\vect{w}'$ such that $\vect{w}'\cdot\vect{u}\geq \vect{w}'\cdot\vect{x}$ for all $\vect{x}\in \Psi$. We call the set $\{\vect{x}\in\mathbb{R}^n\mid \vect{w}'\cdot\vect{x}=\vect{w}'\cdot\vect{u}\}$ a \emph{supporting hyperplane} of $\Psi$.
\end{lemma}

\subsection{Multi-Objective MDP: Basic Queries}\label{sec:existing_queries}
    Let $\Polytope{\vect{\rho}}^\mdp = \{\tolrew{\vect{\rho}}{\schr}\in \mathbb{R}^m \mid \schr\in\mathit{Sch}(\mdp)\}$ \rev{where $m=|\vect{\rho}|$. In words, $\Polytope{\vect{\rho}}^\mdp$ is the set of \rrev{$m$}-dimensional total reward vectors that can be obtained by possible schedulers of $\mdp$ (which is associated with reward vector function $\vect{\rho}$}).
As $\mdp$ is clear in the context,  we abbreviate $\Polytope{\vect{\rho}}^\mdp$ as $\Polytope{\vect{\rho}}$.
The following lemma is a fundamental property for multi-objective MDP model checking.

\begin{lemma}[Polytope for multi-objective MDP \cite{Etessami2008,Forejt2011}]
 $\Polytope{\vect{\rho}}$ is a bounded polytope with finitely many faces.
\end{lemma}
The above proposition states a fundamental property for multi-objective MDPs, which underpins the multi-objective queries. 
\rev{For example, the set of total reward vectors for $\mdp_{\mathit{SW}1}$ is a polygon (i.e., 2-dimensional polytope). The shape of the polytope $\Polytope{\vect{\rho}}$ may be extremely complex, as the worst-case number of its vertices and faces can be exponential in the number of states in the MDP. Multi-objective model checking aims to find out optimal points in $\Polytope{\vect{\rho}}$ for given queries \emph{in an efficient and scalable manner}.}

Before presenting our novel method of convex queries, we recall the three basic types of multi-objective queries~\cite{Forejt2012}.\footnote{A discussion of other multi-objective MDP queries is presented in Section~\ref{sec:related_work:mopmc}.}
Let $\vect{t} \in \mathbb{R}^m$ be a given point, which is used a threshold.
(i) An \textit{achievability query}, which is considered most basic among all queries,  decides whether there exists $\vect{r}\in \Polytope{\vect{\rho}}$ such that either $r_i\leq t_i$ or $r_i\geq t_i$ for all $1\leq i \leq m$.
In words, achievability queries can only check hard constraints. 
(ii)  A \textit{numerical (aka quantitative) query} maximises (or minimises) $r_1$ such that $\vect{r}\in\Polytope{\vect{\rho}}$ and $r_i\leq t_i$ or $r_i\geq t_i$ for all $2\leq i\leq m$. 
Thus, numerical queries can only express constrained optimisation of one objective.
(iii) A \textit{Pareto query} computes the Pareto front of $\Polytope{\vect{\rho}}$, which is computationally expensive and only practical when the multi-objective dimension is very small (e.g., $\leq 3$). 
As our main motivation is computing an optimal scheduler to determine switch design choices, we do not consider Pareto queries.

\subsection{Convex Query}\label{sec:convex_query}
As discussed in Section~\ref{sec:formal_problem}, we use a convex function $\los: \mathbb{R}^m\mapsto \mathbb{R}$ to characterise the total reward trade-off with respective to $\vect{\rho}$ within $\mdp$. For technical reasons, we assume that $\los$ has a continuous gradient on $\mathbb{R}^m$.
Recall that $[\vect{l}, \vect{u}] = [l_1, u_1]\times \ldots\times [l_m, u_m]$ is a possibly unbounded box in $\mathbb{R}^m$ 
(note that to abuse notations, if $l_i = -\infty $ or $u_i = +\infty$ then the interval $[l_i, u_i]$ is an unbounded open interval).
We now present the convex query definition, which is a formulation of Query~\eqref{eq:ccq0}.
\begin{definition}[Convex query]
   A \emph{convex query} on $\mdp$ (associated with $\vect{\rho}$) with respect to a loss function $\los$ and a lower (resp., upper) bound $\vect{l}$ (resp., $\vect{u}$) is the following minimisation problem:
  \begin{equation}\label{eq:ccq}
  \begin{aligned}
  \text{minimse } & \quad \los(\vect{x})  \\
  \text{such that } & \quad \vect{x}\in  \Polytope{\vect{\rho}}  \cap [\vect{l}, \vect{u}]
  \end{aligned}
  \end{equation}
\end{definition}

\begin{algorithm}[t]
\caption{Convex Query\label{alg:main-convex-query}}
\SetKwFor{ParallelForEach}{parallel foreach}{do}{}
\SetKwBlock{DoParallel}{do in parallel}{}
\SetKwRepeat{DoWhile}{do}{while}
\SetKw{Break}{break}
\KwIn{$\mdp $, $\vect{\rho}$, {$\vect{l}$, $\vect{u}$,}
$\los$, $\varepsilon\geq 0$}
$\Phi:=\emptyset$; $\Lambda: =\emptyset$\; 
Initialise $\overline{\vect{v}}, \underline{\vect{v}}$\;
Choose $\vect{w}\in \mathbb{R}^{m}$ s.t.\ $\|\vect{w}\|_1=1$\tcp*[l]{$m=|\vect{\vect{\rho}}|$}
\While{true}{
\label{ln:inner_start}
Find $\vect{r}\in \Polytope{\vect{\rho}}$ s.t.\ $\{\vect{y}\mid \vect{w}\cdot\vect{y}= \vect{w}\cdot\vect{r}\}$ is a supporting hyperplane of $\Polytope{\vect{\rho}}$\;\label{ln:find-supp-hp}
$\Phi:=\Phi\cup \{\vect{r}\}$; {$\Lambda:=\Lambda\cup \{(\vect{w},\vect{r})\}$}\;\label{ln:add-vector}
{\If{$\mathbf{H}(\Lambda)\cap[\vect{l}, \vect{u}]=\emptyset$}{\Return{``infeasible''\;}}}
\If{$|\Lambda|=1$ or $\vect{w} \cdot\vect{r}  <  \vect{w} \cdot \underline{\vect{v}}$\label{ln:main-condition1}}{
{$\underline{\vect{v}}:= \argmin_{\vect{z}\in \mathbf{H}(\Lambda)\cap[\vect{l}, \vect{u}]}\los(\vect{z})$}\;\label{ln:main-find-minimum2}
}
{$\overline{\vect{v}}:= \argmin_{\vect{x}\in \mathbf{V}(\Phi)}\|\underline{\vect{v}}-\vect{x}\|$}\tcp*[l]{mini.\ norm point} \label{ln:main-find-minimum1} 
{\If{$\overline{\vect{v}}\in [\vect{l}, \vect{u}]$ and $\los(\overline{\vect{v}})- \los(\underline{\vect{v}})\leq\varepsilon$
\label{ln:main-exit-condition2}}{\Return{$\overline{\vect{v}}$}\;}}
{$\vect{w}:=  {(\underline{\vect{v}}-\overline{\vect{v}})}/{\|\underline{\vect{v}}-\overline{\vect{v}}\|_1}$}\tcp*[l]{unit length scaling}\label{ln:find-w-1}
}
\end{algorithm}

Algorithm~\ref{alg:main-convex-query} presents an algorithm to solve Problem~\eqref{eq:ccq}. 
In each iteration, the algorithm works by finding a vertex $\vect{r} \in \Polytope{\vect{\rho}}$ and a direction $\vect{w}$ such that $(\vect{w}, \vect{r})$ defines a supporting hyperplane of $\Polytope{\vect{\rho}}$.
It adds $\vect{w}$ to $\Phi$ and $(\vect{w}, \vect{r})$ to $\Lambda$. 
In this way, it iteratively constructs two structures: $\mathbf{H}(\Lambda)$ which is an over-approximation of $\Polytope{\vect{\rho}}$ and  $\mathbf{V}(\Phi)$ which is an under-approximation of $\Polytope{\vect{\rho}}$ (i.e., $\mathbf{V}(\Phi)\subseteq \Polytope{\vect{\rho}} \subseteq \mathbf{H}(\Lambda)$). 
The optimisation in Line~\ref{ln:main-find-minimum2} computes a point $\underline{\vect{v}}\in \mathbf{H}(\Lambda)\cap [\vect{l}, \vect{u}]$ minimising $\los$, while that in Line~\ref{ln:main-find-minimum1} computes $\overline{\vect{v}}\in\mathbf{V}(\Phi)$ \orev{which minimises its distance to  $\underline{\vect{v}}$. The algorithm either returns an approximate optimal value eventually or asserts that the query is ``infeasible'' (namely ``no points can meet the hard constraints'').}

\begin{theorem}[Correctness of Algorithm~\ref{alg:main-convex-query}]\label{thm:main-alt}
 Algorithm~\ref{alg:main-convex-query} terminates. It returns ``infeasible'' if and only if $\Polytope{\vect{\rho}}\cap[\vect{l}, \vect{u}] =\emptyset$. If it returns $\overline{\vect{v}}$ (i.e., $\Polytope{\vect{\rho}}\cap[\vect{l}, \vect{u}] \neq \emptyset$) then 
 $$\min_{\vect{x}\in\Polytope{\vect{\rho}}\cap[\vect{l},\vect{u}]}\los(\vect{x}) \leq \los(\overline{\vect{v}}) \leq \min_{\vect{x}\in\Polytope{\vect{\rho}}\cap[\vect{l},\vect{u}]}\los(\vect{x}) + \varepsilon.$$
\end{theorem}

The last part of Theorem~\ref{thm:main-alt} states that Algorithm~\ref{alg:main-convex-query} returns an approximate optimal value of the convex function $\los$ in $\Polytope{\vect{\rho}}\cap[\vect{l},\vect{u}]$ (with an error up to $\varepsilon$).

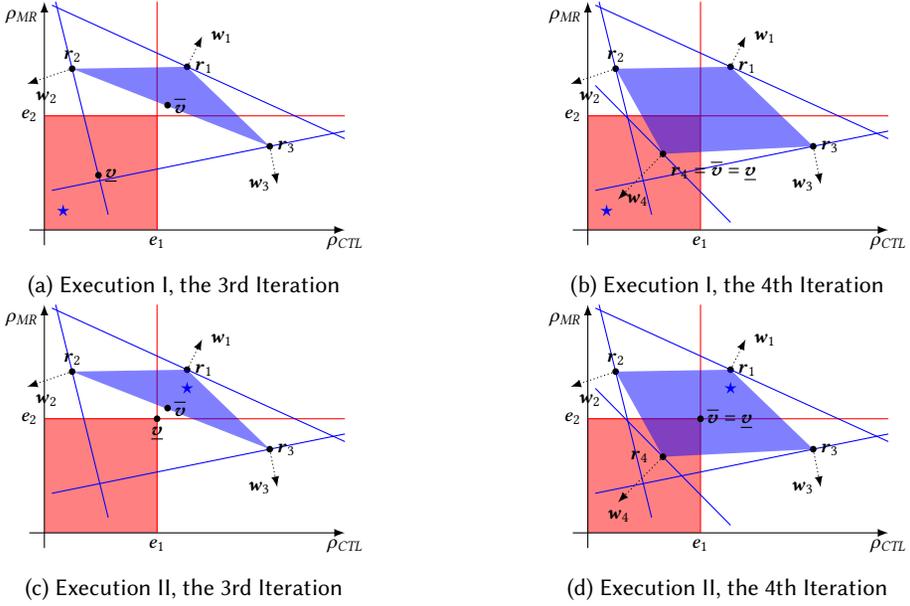
\begin{figure}[t]
\centering
\scriptsize
\def\opac{0.5}
\def\dotsize{1.25pt}
\def\redboxsize{1.5} 

\def\fa(#1){3 - #1*0.45}
\def\fb(#1){3.6 - #1*4}
\def\fc(#1){0.5 + #1*0.2}
\def\fd(#1){2 - #1}

\def\xa{1.9} \def\xb{0.37} \def\xc{3} \def\xd{1}

\tikzset{
  axis/.style={-latex},
  funcline/.style={blue},
  dashedline/.style={gray!40, dashed},
  point/.style={fill=blue, circle, radius=\dotsize},
  optstar/.style={star, star points=5, star point ratio=3,
                  minimum size=5pt, fill=blue, inner sep=0pt}
}

\newcommand{\drawCommonElements}{%
    \coordinate (Pa) at ({\xa}, {\fa(\xa)});
    \coordinate (Pb) at ({\xb}, {\fb(\xb)});
    \coordinate (Pc) at ({\xc}, {\fc(\xc)});
    \coordinate (Pd) at ({\xd}, {\fd(\xd)});
    \fill[red, opacity=\opac] (0,0) rectangle (\redboxsize,\redboxsize);
    \draw[axis] (-0.2,0) -- (4,0) node[below] {$\rho_\mathit{CTL}$};
    \draw[axis] (0,-0.2) -- (0,3) node[left, yshift=-.2cm] {$\rho_\mathit{MR}$};
    \draw[red] (0,\redboxsize) -- (4,\redboxsize);
    \node[left] at (0,\redboxsize) {$e_2$};
    \draw[red] (\redboxsize,0) -- (\redboxsize,3);
    \node[below] at (\redboxsize,0) {$e_1$};
    \draw[funcline, domain=0.1:4]   plot (\x, {\fa(\x)});
    \draw[funcline, domain=0.15:0.85]plot (\x, {\fb(\x)});
    \draw[funcline, domain=0.1:4]   plot (\x, {\fc(\x)});
    \fill (Pa) circle (\dotsize) node[right] {$\vect{r}_1$};
    \fill (Pb) circle (\dotsize) node[above] {$\vect{r}_2$};
    \fill (Pc) circle (\dotsize) node[right] {$\vect{r}_3$};
    \draw[densely dotted,-latex] (Pa) -- ++(0.2,0.4) node[right] {$\vect{w}_1$};
    \draw[densely dotted,-latex] (Pb) -- ++(-0.6,-0.2) node[below right] {$\vect{w}_2$};
    \draw[densely dotted,-latex] (Pc) -- ++(0.1,-0.5) node[left] {$\vect{w}_3$};
}

\newcommand{\drawOptimalAt}[2]{%
  \node[optstar] at (#1,#2) {};
}

\begin{subfigure}[t]{0.48\textwidth}
\centering
\begin{tikzpicture}[x=1cm,y=1cm,domain=0:4]
  \drawCommonElements
  \drawOptimalAt{0.25}{0.25}
  \fill[blue,opacity=\opac] (Pa)--(Pb)--(Pc)--cycle;
  \fill (1.64,1.64) circle (\dotsize) node[right] {$\overline{\vect{v}}$};
  \fill (0.72,0.72) circle (\dotsize) node[right] {$\underline{\vect{v}}$};
\end{tikzpicture}
\caption{Execution I, the 3rd Iteration}\label{fig:iter3-1}
\end{subfigure}
\hfill
\begin{subfigure}[t]{0.48\textwidth}
\centering
\begin{tikzpicture}[x=1cm,y=1cm,domain=0:4]
  \drawCommonElements
  \drawOptimalAt{0.25}{0.25}
  \fill[blue,opacity=\opac] (Pa)--(Pb)--(Pd)--(Pc)--cycle;
  \draw[funcline, domain=0.1:1.9] plot (\x,{\fd(\x)});
  \fill (Pd) circle (\dotsize)
       node[below right,xshift=0cm] {$\vect{r}_4=\overline{\vect{v}}=\underline{\vect{v}}$};
  \draw[densely dotted,-latex] (Pd)--++(-.6,-.6) node[right] {$\vect{w}_4$};
\end{tikzpicture}
\caption{Execution I, the 4th Iteration}\label{fig:iter4-1}
\end{subfigure}

\begin{subfigure}[t]{0.48\textwidth}
\centering
\begin{tikzpicture}[x=1cm,y=1cm,domain=0:4]
  \drawCommonElements
  \drawOptimalAt{1.9}{1.9}
  \fill[blue,opacity=\opac] (Pa)--(Pb)--(Pc)--cycle;
  \fill (1.64,1.64) circle (\dotsize) node[right] {$\overline{\vect{v}}$};
  \fill (1.5,1.5) circle (\dotsize) node[below] {$\underline{\vect{v}}$};
\end{tikzpicture}
\caption{Execution II, the 3rd Iteration}\label{fig:iter3-2}
\end{subfigure}
\hfill
\begin{subfigure}[t]{0.48\textwidth}
\centering
\begin{tikzpicture}[x=1cm,y=1cm,domain=0:4]
  \drawCommonElements
  \drawOptimalAt{1.9}{1.9}
  \fill[blue,opacity=\opac] (Pa)--(Pb)--(Pd)--(Pc)--cycle;
  \draw[funcline, domain=0.1:1.9] plot (\x,{\fd(\x)});
  \fill (Pd) circle (\dotsize) node[left,xshift=-0.1cm] {$\vect{r}_4$};
  \draw[densely dotted,-latex] (Pd)--++(-.6,-.6) node[below] {$\vect{w}_4$};
  \fill (1.5,1.5) circle (\dotsize) node[right] {$\overline{\vect{v}}=\underline{\vect{v}}$};
\end{tikzpicture}
\caption{Execution II, the 4th Iteration}\label{fig:iter4-2}
\end{subfigure}

\caption{Illustrative Executions of Alg.~\ref{alg:main-convex-query} for $\mdp_{\mathit{SW}1}$ \orev{(where $\star$ indicates the location of   $(d_1, d_2)$, the soft threshold vector in Eq.~\eqref{eq:ccq-sw})}}
\label{fig:iters}
\end{figure}

Figure~\ref{fig:iters} illustrates two possible executions of Algorithm~\ref{alg:main-convex-query} for $\mdp_{\mathit{SW}1}$. As $\rho_\mathit{CTL}$ and $\rho_\mathit{MR}$ are considered, all vectors (points) are two-dimensional. 
In the first execution, three pairs of vectors are computed: $(\vect{w}_1, \vect{r}_1)$, $(\vect{w}_2, \vect{r}_2)$, and $(\vect{w}_3, \vect{r}_3)$, as shown in Figure~\ref{fig:iter3-1}. 
Accordingly, $\Phi = \{ \vect{w}_1, \vect{w}_2, \vect{w}_3 \}$ and $\Lambda = \{ (\vect{w}_1, \vect{r}_1), (\vect{w}_2, \vect{r}_2), (\vect{w}_3, \vect{r}_3) \}$. 
The algorithm computes a point $\underline{\vect{v}}$ closest to $\star$ within the intersection of $\mathbf{H}(\Lambda)$ (the region defined by the three blue side lines) and $(-\infty, e_1]\times (-\infty, e_2]$ (the red region). 
It also computes a point $\overline{\vect{v}}$ closest to $\underline{\vect{v}}$ within $\mathbf{V}(\Phi)$ (the blue region).
Let $\vect{w} = (\underline{\vect{v}} - \overline{\vect{v}})/\|\underline{\vect{v}} - \overline{\vect{v}}\|$. 
Figure~\ref{fig:iter4-1} shows the next iteration in which the vertex $\vect{r}_4$ is computed by setting the direction $\vect{w}_4$ to $\vect{w}$. 
Then, by adding $\vect{r}_4$ to $\Phi$ and $(\vect{w}_4, \vect{r}_4)$ to $\Lambda$, the new locations of $\underline{\vect{v}}$ and $\overline{\vect{v}}$ are obtained, as shown in Figure~\ref{fig:iter4-2}.
As $\underline{\vect{v}}$ and $\overline{\vect{v}}$ coincide,  they represent the minimum location and thus the algorithm terminates.
In the second execution, we assume that the idealised minimum is attained at a different location of $\star$ (outside the red region), as shown in Figure~\ref{fig:iter3-2}.
Thus, \orev{the location of $\underline{\vect{v}}$ differs from that in Figure~\ref{fig:iter3-1} while that of $\overline{\vect{v}}$ remains unchanged.}
Figure~\ref{fig:iter4-2} illustrates the vertex $\vect{r}_4$ , direction $\vect{w}_4$, and new locations of $\underline{\vect{v}}$ and $\overline{\vect{v}}$.
Again, the algorithm terminates, and the minimum locates at $\overline{\vect{v}}$ (and equivalently $\underline{\vect{v}}$).


\subsection{Policy-Value Iteration for Hyperplane Computation}\label{sec:parallel_pvi}

Line~\ref{ln:find-supp-hp} in Algorithm~\ref{alg:main-convex-query} is the most computation-intense step in the convex query. It can be computed by using a standard iterative solution method (i.e., policy-value iteration) in MDP model checking. 
\begin{proposition}\label{thm:value_iteration}
Given any $\vect{w}\in \mathbb{R}^m$ such that $m = |\vect{\rho}|$, let $\schr_{\vect{w}} = \argmax_{\schr\in \mathit{Sch}(\mdp)} \tolrew{\vect{\rho}\cdot \vect{w}}{\schr}$. Then, $\{\vect{y} \mid \vect{w}\cdot\vect{y} = \vect{w} \cdot\tolrew{\vect{\rho}}{\schr_{\vect{w}}}\}$ is a seperation hyperplane of $\Polytope{\vect{\rho}}$.
\end{proposition}
The above proposition suggests the following two steps to compute the support vector $ \tolrew{\vect{\rho}}{\schr_{\vect{w}}}$:
\begin{enumerate}
\item Compute an optimal schedule $\schr_{\vect{w}}$ that maximises $\tolrew{\vect{\rho}\cdot \vect{w}}{\schr}$ where $\schr\in \mathit{Sch}(\mdp)$ by policy iteration.
\item Compute $ \tolrew{\rho_i}{\schr_{\vect{w}}}$ for each $i= 1,\ldots, m$ by value iteration (where $m = |\vect{\rho}|$, the number of objectives).
\end{enumerate}

A salient observation of this scheme is the inherent \textit{parallelisation} in the second step. This characteristic is a significant departure from single-objective MDP model checking \rrev{and motivates parallel-computing} solutions in our convex queries and, in general, multi-objective MDP model checking.

\subsection{\rev{Random Scheduler Construction}}\label{sec:scheduler_generation}
Algorithm~\ref{alg:main-convex-query} returns an (approximate) optimal value, but a solution to Problem~\eqref{eq:ccq0} should be an optimal scheduler. However, this scheduler can be generated in the following way:
Suppose the algorithm returns $\overline{\vect{v}}$ after $n$ iterations.
Let $\Phi = \{\tolrew{\vect{\rho}}{\pi_1} ,\ldots, \tolrew{\vect{\rho}}{\pi_n} \}$ where $\pi_1,\ldots,\pi_n$ be the corresponding optimal schedulers.
The returned vector $\overline{\vect{v}}\in \mathbf{V}(\Phi )$ is actually represented by $p_1\tolrew{\vect{\rho}}{\pi_1} +\ldots+ p_n\tolrew{\vect{\rho}}{\pi_n} $ where $(p_1,\ldots,p_n)$ is a weight vector.
The convex combination of schedulers $p_1{\pi_1} +\ldots+ p_n{\pi_n}$ yields an approximate optimal solution to Problem~\eqref{eq:ccq0}. 
In practice, viewing the convex vector as a probability distribution, we can randomly select schedulers from $\schr_1, \ldots, \schr_n$. 

\section{{MOPMC} Implementation}\label{sec:tool}

This section presents our tool, a \textbf{M}ulti-\textbf{O}bjective \textbf{P}robabilistic \textbf{M}odel \textbf{C}hecker (MOPMC).
The aims of MOPMC's development were to implement the new method for analysing design choices within our real-time switching design framework and to create a new, standalone software tool for multi-objective MDP model checking.
Compared with the existing probabilistic model checking tools~\cite{Andriushchenko2024}, MOPMC features two key innovations:
\begin{enumerate}
    \item It supports convex queries on multi-objective MDPs by implementing Algorithm~\ref{alg:main-convex-query}.
    \item It utilises GPU to accelerate the policy-value iteration (c.f. Section~\ref{sec:parallel_pvi}).
\end{enumerate}

\subsection{Model Specification}
MOPMC employs the PRISM language, a domain-specific language for high-level specification of MDP models.
For a complete guide to the PRISM language, please refer to the official PRISM manual.\footnote{
\url{https://prismmodelchecker.org/manual/ThePRISMLanguage/}}
List~\ref{list:prism_model} presents a specification snippet of $\mdp_{\mathit{SW}1}$, the syntactic-converted MDP for \orev{our running example} $\ta_{\mathit{SW}1}=\ta_{\mathit{SW}0} \| \ta_{\mathit{TS}}  \| \ta_{\mathit{CC}}\|\ta_{\mathit{MR}}$. 
The keyword \textbf{\texttt{mdp}} indicates the model type. 
The main model specification {includes} three parts.

\definecolor{prismgreen}{rgb}{0, 0.6, 0}
\lstdefinelanguage{Prism}{ 
basicstyle=\color{red}\scriptsize\ttfamily, 
keywords={bool,C,ceil,clock,const,ctmc,double,dtmc,endinit,endmodule,endrewards,
endsystem,F,false,floor,formula,G,global,I,init,int,label,max,mdp,min,
module,nondeterministic,P,Pmin,Pmax,prob,probabilistic,pta,R,rate,rewards,
Rmin,Rmax,S,stochastic,system,ta,true,U,X},
keywordstyle={\bfseries\color{black}},
numberstyle=\tiny\color{black},
comment=[l] {//}, morecomment=[s]{/*}{*/}, 
commentstyle= \color{prismgreen}, 
tabsize=4, 
captionpos=b, 
escapechar=@ 
}

\begin{figure}[t]
\begin{center}
\begin{minipage}{0.88\textwidth}
\lstset{
language={Prism}, numbers=left, frame=single,
rulesepcolor=\color{black}, rulecolor=\color{black}, breaklines=true,
breakatwhitespace=true, firstnumber=1, firstline=1, lastline=25,
caption={(\rev{Partial}) Specification of $\mdp_{\mathit{SW}1} $ in PRISM},label={list:prism_model}
}
\begin{lstlisting}
mdp // modeling a switch between pri and sec controllers under m configs
const int n; // max number of timesteps
const int k; // min number of consecutive calls to a secondary controller
const int m: // number of configurations
const double p1; // probability of entering the primary controller in config 1
// similar declarations for config 2..m
const double q1; // probability of hitting 'yes' in MetricRecorder in config 1
// similar declarations for config 2..m
// other probability parameters...
module BaseSwitch
	cfg: [0..m]; //m configs, 0 represents a null config
  	sm: [0..2] init 0; //main state variable
  	sp: [0..4] init 0; //primary control state varaible
  	ss: [0..5] init 0; //secondary control state varaible 
  	[conf] sm=0 -> (sm'=1) & (cfg'=1);
	// more config transition rules for cfg=2...m
	[] sm=1 & sp=0 & ss=0 & cfg=1 -> (1-p1) : (sp'=1) + p1 : (ss'=1);
	// more prob transition rules for cfg=2...m
	[obs] sp=1 -> (sp'=2);
	[obs] ss=1 -> (ss'=2);
	[act_p] sp=2 -> (sp'=3);
	[act_s] ss=2 -> (ss'=3);
	// other transition rules...
endmodule
module TimestepCounter
 	tc: [0..n] init 0;
	[conf] tc<MAX_TS -> true;
 	[obs] tc<MAX_TS -> true;
 	[act_p] tc<n -> (tc'=tc+1);
  	[act_s] tc<n -> (tc'=tc+1);
	// other transition rules...
endmodule
module ConsecCallCounter
	cc: [0..k]; //consecutive call counts
 	[re_s] cc<k -> (cc'=cc+1);
  	[re_s] cc=k -> (cc'=k);
  	[ex_s] cc=k -> (cc'=0);
endmodule
module MetricRecorder
  	mr: [0..3] init 0; //2: yes, 3: no
  	[obs] mr=0 -> (mr'=1);
	[] cfg=1 & mr=1 -> q1 : (mr'=2) + (1-q1) : (mr'=3);
	// more prob transition rules for cfg=2...m and other transition rules...
endmodule
rewards "EXP_CTL"
	[act_p] true : 1; // assigns cost 1 to each call to primary controller
	[act_s] true : 2; // assigns cost 2 to each call to sec. controller
endrewards
rewards "MR" // mr=2 indicates ``yes'' is recorded 
  [ex_p] mr=2 : 1/n; 
  [ex_s] mr=2 : 1/n; 
  [re_s] mr=2 : 1/n; 
endrewards
\end{lstlisting}
\end{minipage}
\vspace{-2em}
\end{center}
\end{figure}

The first part defines three constant integers \texttt{n} (the maximum number of timesteps), \texttt{k} (the minimum number of consecutive calls to a secondary controller), \rev{ \texttt{m} (the number of configurations)}, and transition probability parameters (e.g., \rev{\texttt{p1}, \ldots, \texttt{pm}}).
All these constants, whose the meanings are annotated in the specification snippet, need to be instantiated by the model designer or using parameter estimation. 

The second part defines four modules \texttt{BaseSwitch}, 
\texttt{TimestepCounter}, \texttt{ConsecCallCounter} and \texttt{MetricRecorder} corresponding to the four simple TAs: $\ta_{\mathit{SW}0}$, $\ta_{\mathit{TS}} $, $\ta_{\mathit{CC}}$ and $\ta_{\mathit{MR}}$, respectively.
Each module contains (integer) state variables and transition rules, whose syntax is self-explanatory. 
For example:
\rev{
\begin{itemize}
\item \texttt{cfg:[0...m]} defines that the variable \texttt{cfg} ranges from 0 to m;
\item \texttt{[] sm=1 \& sp=0 \& ss=0 \& cfg=1 -> (1-p1):(sp'=1) + p1:(ss'=1)} reads as: ``if \texttt{sm=1 \& sp=0 \& ss=0 \& cfg=1} and an implicit action \texttt{[]} is performed, then one of these two cases occurs: (i) with probability \texttt{1-p1}, \texttt{sp} is updated to 1, and (ii) with probability \texttt{p1}, \texttt{ss} is updated to 1'';
\item  \texttt{[obs] sp=1 -> (sp'=2)} means: ``if \texttt{sp=1} and \texttt{[obs]} is performed then \texttt{sp} is updated to 2''.
\end{itemize}}
The transition labels (i.e., actions) are relevant to an important semantic convention in PRISM models: Transition rules that are defined in \emph{different} modules but share the same action must be synchronised. Thus, \orev{the transition rule that is labelled by \texttt{[obs]} in \texttt{BaseSwitch} must be synchronised with transition rules in other modules which are also labelled by \texttt{[obs]} (if any).}
By contrast, \orev{transitions rules with an empty label \texttt{[]}} are executed independently.

The last part of the model specification includes two reward functions \texttt{``\rev{EXP\_CTL}''} and \texttt{``\rev{MR}''}, which correspond to $\rho_{\mathit{CTL}}$ and $\rho_{\mathit{MR}}$, respectively.
These two reward functions are used to define total rewards in the convex query objective and constraints (c.f., Query~\eqref{eq:ccq-sw}).

\subsection{Tool Architecture} 
MOPMC utilises Storm's C++ interface to parse, build and pre-process MDP models specified in the PRISM language. 
In addition to total rewards, it also accepts accumulated rewards and temporal properties as multi-objective property specifications. 
The architectural overview of MOPMC is depicted in Figure~\ref{fig:mopmc_arch}.
For input, MOPMC processes MDPs described in the PRISM modelling language, accepting properties including total rewards, accumulated rewards, and temporal properties. Output includes optimal values and, for total reward queries, an exportable scheduler. The tool employs Storm's parser and model builder, along with a customised model exporter.  The constructed models are fed into MOPMC's query algorithms, which currently handle both constrained/unconstrained convex queries and achievability queries. 
These queries call the standard or GPU-accelerated policy-value iteration routines.
The tool's  implementation uses C++ and CUDA, and its source code is available in the GitHub repository \url{https://github.com/gxshub/mopmc}. 
 
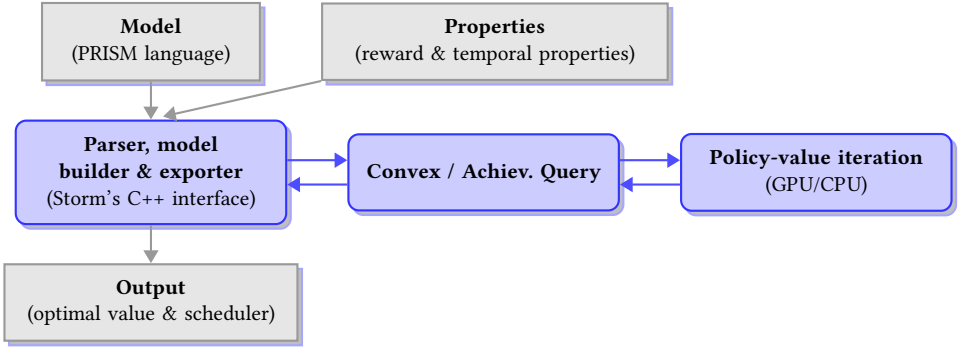
\begin{figure}[!tbh]
\ifacm\newcommand{\blocklength}{3.2cm}\fi
\ifacm\newcommand{\propertyblocklength}{3.5cm}\fi
    \centering
    \footnotesize
   \begin{tikzpicture}[
        scale=0.8, 
        node distance= 0.5cm and 0.8cm, 
        every path/.style={draw=blue!70, text=black, thick}, 
        block/.style={
            rectangle,
            rounded corners,
            draw=blue!80,
            fill=blue!20,
            text=black,
            inner sep=2mm,
            text width=\blocklength, 
            align=center,
            minimum height=1cm,
            drop shadow
        },
        inputoutput/.style={rectangle, aspect=1.5, draw=gray!80, fill=gray!20, text=black, inner sep=2mm, align=center, minimum height=0.8cm, text width=\propertyblocklength, drop shadow}, 
        arrowlable/.style={-Triangle,font=\small, text width=2cm, align=center, midway, above},
    ]
    \node (model) [inputoutput,  text width=\propertyblocklength-1cm] {\textbf{Model} \\ (PRISM language)};
    \node (property) [inputoutput, right=of model, text width=\propertyblocklength+0.7cm ] {\textbf{Properties} \\ (reward \& temporal properties)};
    \node (frontend) [block, below=of model] {\textbf{Parser, model builder \& exporter} \\ (Storm's C++ interface)};
    \node (cq) [block, right=of frontend] {\textbf{Convex / Achiev.\ Query}};
    \node (vi) [block, right=of cq] {\textbf{Policy-value iteration} \\ (GPU/CPU)};
    \node (output) [inputoutput, below=of frontend] {\textbf{Output} \\ (optimal value \& scheduler)};
    \draw[-Triangle,draw=gray!80] (model.south) -- (frontend.north); 
    \draw[-Triangle,draw=gray!80] (property.south west) -- ([xshift=.2cm,yshift=.05cm]frontend.north); 
    \draw[-Triangle,draw=gray!80] (frontend.south) -- (output.north);
    \draw[-Triangle,] ([yshift=.2cm]frontend.east) -- ([yshift=.2cm]cq.west);
    \draw[-Triangle,] ([yshift=-.2cm]cq.west) -- ([yshift=-.2cm]frontend.east);
    \draw[-Triangle,] ([yshift=.2cm]cq.east) -- ([yshift=.2cm]vi.west); 
    \draw[-Triangle,] ([yshift=-.2cm]vi.west) -- ([yshift=-.2cm]cq.east);

    \end{tikzpicture}
    \caption{MOPMC Architecture}
    \label{fig:mopmc_arch}
\end{figure}

\section{Evaluation}\label{sec:case_study}

Our evaluation aims to address the following research questions (RQs): 
\begin{description}

\item[\textbf{RQ1}:] What is the impact of time delay on the performance of DRL agents with respect to critical evaluation metrics?

\item[\textbf{RQ2}:] To what extent can our real-time switch architecture, combined with multi-objective model checking and the proposed convex query method, effectively balance and mitigate the impact of time delay in DRL systems?


\item [{\textbf{RQ3}}:] How does the computational performance of MOMPC, in terms of runtime efficiency, compare to that of existing probabilistic model checking tools?

\end{description}

\subsection{Experiment Setup and Performance Metrics}\label{sec:exp_setup}

In most DRL environments for continuous control problems, the agent's observation-action rate is typically set \rrev{to} $k/\Delta t$, where $k$ is a positive integer and $\Delta t$ is the time interval between two consecutive updates conducted by the simulator. 
This means that after an agent observes the environment state and performs an action, the simulator updates the environment $k$ times before the agent receives its next observation.
However, 
in real-time scenarios, the time elapsed between an agent's observation and action must be explicitly accounted for, and the environment and the DRL agent operate asynchronously. Agents trained in synchronous\rev{,} deterministic tick-tock simulations may manifest significant performance degradation and even result in safety hazards when deployed in a realistic environment with variable delays.

For our evaluation, we implemented a real-time version of Gymnasium HighwayEnv~\cite{highway-env}, a DRL environment designed for behavioural planning in autonomous driving.
In this \rev{modified} setup, we injected a fixed or randomised \emph{time delay} between the autonomous driving agent's observation and action. Crucially, during this delay period, the simulated highway environment continues to update \rev{asynchronously}, reflecting real-world dynamics where the world does not pause for the agent's decision-making.

Since our focus is on the controller switch mechanism \rev{rather than} DRL training itself, we utilised off-the-shelf DRL libraries, specifically Stable Baselines3~\cite{stable-baselines3}, to train a DRL agent in the standard HighwayEnv (without any induced time delay). 
The experimental data presented \rrev{in} this case study is based on the Deep Deterministic Policy Gradient (DDPG) algorithm, a widely used classic DRL approach. 
We also implemented simple rule-based controllers that serve as the fallback options activated by our switch when significant time delay occurs. 

We \rev{considered} the performance of the DRL agent (i.e., DDPG controller) and rule-based controllers against three common autonomous driving metrics: \emph{time-to-collision} (TTC), and \emph{headway} (HW), and \emph{lane-departure} (LD).
These metrics, while not necessarily used as reward signals during DRL training, represent safety-critical performance indicators. 
In addition, as the rule-based controllers have more restricted behaviours than the DDPG controller, \orev{they are usually not a preferable  option in normal conditions. In view of this, we associated a higher cost value to the former and a lower cost value to the latter.  The \emph{overall controller cost} is denoted CTL}. 
\rev{Another important reason of selecting these four metrics is that they represent \emph{three basic types of variables}: LD is a \emph{Boolean variable}, {CTL} is a \emph{discrete variable}, and TTC and HW are \emph{continuous variables}.
In the next section we will show how to use total rewards to characterise them.}
The source code and experimental data for this case study are available in the GitHub repository \url{https://github.com/gxshub/realtime-drl-switch}.
All the evaluations were performed on a Ubuntu 20.04 LTS system with an Intel i9-12900 CPU, 24GB of RAM and 16GB of GPU memory.

\subsection{Time Delay Impact on DRL Agent (RQ1)}\label{sec:time_delay_impact}

Table~\ref{tab:time_delay_impact} summarises the evaluation results for \textbf{RQ1}, namely the impact of time delay on the DDPG agent across the three safety-critical metrics.
\rev{The data are collected by evaluating the agent in the real-time HighwayEnv. The field $\pcttc$ denotes the \emph{10th percentile of TTC}, $\pchw$ refers to the \emph{50th percentile (i.e.\ median) of headway}, and $\probld$ is the probability of lane departure.}

\begin{table}[htb]
       \centering
    \caption{\rev{Impact of Fixed (Left) and Randomised (Right) Time Delays on a DDPG Agent (Time in Seconds)}}
    \label{tab:time_delay_impact}
\fontsize{9}{9}\selectfont 
\setlength{\tabcolsep}{3pt} 
    \begin{minipage}[t]{0.48\textwidth} 
        \centering
        \label{tab:fixed_delay}
        \begin{tabular}{cccc}
   	 \toprule
   	 \textbf{Fixed delay} & \rev{$\pcttc$} & \rev{$\pchw$} & \rev{$\probld$} \\
	\midrule
	0.00 & 28.0064 & 4.8403 & 0.0001 \\
	0.05 & 27.8275 & 4.7022 & 0.0000 \\
	0.10 & 24.9783 & 4.0851 & 0.0001 \\
	0.15 & 20.6544 & 4.2097 & 0.0328 \\
	0.20 & {16.8771} & {4.6250} & {0.2316} \\
	0.25 & 16.5499 & 4.6607 & 0.2444 \\
	0.30 & 11.7091 & 4.7424 & 0.5416 \\
	0.35 & 9.5141 & 4.6305 & 0.5924 \\
	0.40 & 8.7159 & 4.7036 & 0.6324 \\
	0.45 & 8.8419 & 4.7278 & 0.6378 \\
	0.50 & 9.2667 & 5.0252 & 0.6308 \\
   	 \bottomrule
    \end{tabular}
    \end{minipage}
    \hfill
    \begin{minipage}[t]{0.48\textwidth} 
        \centering
        \label{tab:randomised_delay}
   \begin{tabular}{cccc}
    \toprule
    \textbf{Avg delay} & $\pcttc$ & $\pchw$ & $\probld$ \\
	\midrule
	0.00 & 28.0064 & 4.8403 & 0.0001 \\
	0.05 & 18.8983 & 4.3366 & 0.0003 \\
	0.10 & 17.5975 & 3.8275 & 0.0063 \\
	0.15 & 15.4452 & 3.7491 & 0.0354 \\
	{0.20} & {14.4638} & {3.9710} & {0.1152} \\
	0.25 & 13.3526 & 4.1480 & 0.1923 \\
	0.30 & 12.7368 & 4.2904 & 0.2849 \\
	0.35 & 12.3128 & 4.5370 & 0.3526 \\
	0.40 & 11.6652 & 4.5911 & 0.4057 \\
	0.45 & 11.9127 & 4.6872 & 0.4375 \\
	0.50 & 12.1273 & 4.8036 & 0.4655 \\
    \bottomrule
   \end{tabular}
    \end{minipage}
\end{table}
The results in Table~\ref{tab:time_delay_impact} clearly illustrate the detrimental effect of increasing time delay on the performance of the DRL agent. 
With negligible delays (0.00s to 0.05s), the agent maintains a very low lane-departure probability and healthy TTC values, indicating stable and safe operation. 

\rev{For both fixed and randomised delay impacts, as the delay time increases beyond $0.10 \text{s}$, an obvious degradation in safety is observed. The $\pcttc$ values decrease dramatically, signifying a greater propensity for dangerously close interactions and potential collisions. The $\pchw$ values remain relatively stable but still \rrev{show} a decrease initially before rising again. The $\probld$ metric also indicates a dramatically increasing chance of lane departure.} 

\rrev{These} data strongly motivates the need for mechanisms, such as the proposed controller switch, to enhance robustness of DRL-driven systems against such time-sensitive performance degradation.

\subsection{Switch Design and Analysis (RQ2)}\label{sec:eval:swtich_design_analysis}

\subsubsection{Switch Design Modelling}

\rev{In this case study, we used a composed TA $\ta_{\mathit{SW}2} = \ta_{\mathit{SW}0} \| \ta_{\mathit{TS}} $ $ \| \ta_{\mathit{CC}}\|\ta_{\mathit{LD}}\| \ta_{\mathit{TTC}} \|\ta_{\mathit{HW}}
$ (which is an extension to $\ta_{\mathit{SW}1}$).
Recall that $\ta_{\mathit{SW}0}$ (depicted in Figure~\ref{fig:ta_sw}) is the base model which represents the most basic behaviours of the switch, $\ta_{\mathit{TS}}$ (Figure~\ref{fig:ra1}) records the timesteps, and $\ta_{\mathit{CC}}$ (Figure~\ref{fig:ra2}) records the consecutive calls to a secondary controller. 
}
The three additional TAs $\ta_{\mathit{LD}}$, $\ta_{\mathit{TTC}}$ and $\ta_{\mathit{HW}}$ are concrete instances of $\ta_{\mathit{MR}}$ (Figure~\ref{fig:ra3}), which correspond to the aforementioned three performance metrics LD, TTC, and HW, respectively.
\rrev{Semantically,  $\ta_{\mathit{SW}2}$ is specified by defining the edges for each of its TA components.} 
\rev{The interactions of these atomic TAs are through the synchronisation of the shared actions among them.}

\rev{We considered two secondary controllers and set the time threshold to trigger those controllers to 0.1s or 0.2s.  Therefore, for the action representing an abstract configuration $\mathtt{conf}(t, \iota)$ in $\ta_{\mathit{SW}0}$, $t$ is $0.1$ or $0.2$, and  $\iota$  is $1$ or $2$.
This leads to four (concrete) configurations (i.e., basic design choices) for $\ta_{\mathit{SW}2}$, denoted C1--C4.} Based on the four basic configurations, we analysed two advance configuration strategies: (i) randomised configurations (i.e., determining the four decision choices randomly), and (ii) dynamic configurations (i.e., determining different design choices at different timesteps of $\ta_{\mathit{SW}2}$).

Following the methodology detailed in Section \ref{sec:syntactic_mdp}, \rrev{$\ta_{\mathit{SW}2}$ can be converted into an MDP, denoted $\mdp_{\mathit{SW}2}$}. 
\rrev{Similar to $\ta_{\mathit{SW}2}$, $\mdp_{\mathit{SW}2}$ can be specified both graphically and semantically. However, in order to analyse it with MOPMC and other existing probabilistic model checking tools, we specify $\mdp_{\mathit{SW}2}$ syntactically in the PRISM language. }  \rev{For readability, the complete PRISM specification of $\mdp_{\mathit{SW}2}$ is included in Appendix~\ref{sec:complete_prism_spec}. 
This MDP has two integer parameters representing the maximum timestep and the minimum number of consecutive calls to a second controller. Adjusting the maximum timestep, in particular, allows us to scale the model size. It has four non-deterministic actions each of its choice states and includes initially \num{20} undetermined transition probabilities.} Among these probabilities, we then estimated the time-delay probabilities by using an approximately exponential distribution with scale parameter $0.2$ (seconds), and conducted simulations within our real-time HighwayEnv to derive values for other transition probabilities.

\rev{
Two further points are worth noting:  
First, we used the TA $\ta_{\mathit{SW}2}$ (but not the MDP $\mdp_{\mathit{SW}2}$) as our formal switch design model because the latter lacks the clock notion, which is essential for real-time switch modelling.  
Second, the construction of both $\ta_{\mathit{SW}2}$ and $\mdp_{\mathit{SW}2}$ can be incremental owing to their modular nature.
}

 \rev{
\subsubsection{Defining Reward Functions and Total Rewards}
In Section~\ref{sec:time_delay_impact}, we have adopted the LD probability ($\probld$) and the TTC/HW percentiles ($\pcttc$ and $\pchw$) as specific performance metrics. Here we show how to use total rewards (c.f., Definition~\ref{def:total_reward}) to define $\probld$ as well as reverse concepts of $\pcttc$ and $\pchw$ (referred to as \emph{percentile ranks}).
Note that 
$\prttc$ (resp., $\prhw$) of score $q$ is the probability of $\prttc\leq q$ (resp., $\prhw\leq q$). For example, $\prttc$ of score $\pcttc$  (the 10th percentile) is $0.1$, and \rrev{$\pchw$ of score $\prhw$ (the median) is $0.5$}.
As we also considered the cost of controllers (CTL), we show how to define the expected controller cost (denoted $\expctl$) as a total reward.
Table~\ref{tab:rewards_metrics} presents the definitions and meanings of reward functions and total rewards corresponding total rewards are summarised in Table~\ref{tab:rewards_metrics}. 
}

\begin{table}[tbh!] 
\centering
\caption{\rev{Total Rewards for CTL, LD, TTC and HW}}\label{tab:rewards_metrics}
\fontsize{9}{9.5}\selectfont 
\setlength{\tabcolsep}{3.5pt} 
\begin{tabular}{ccp{3.5cm}p{3.5cm}cc} 
\toprule
\textbf{Metric} & \textbf{Type} & \textbf{Reward function} & \textbf{Meaning of total reward} & \rrev{\textbf{Notation}}  \\
\midrule
CTL & discrete & assign a cost to each controller &  CTL \emph{expected value}  & \rrev{$\expctl$}   \\
LD & Boolean &assign 1 if $\text{LD}$ is true and 0 otherwise & LD \emph{probability} &  \rrev{$\probld$}  \\
TTC & continuous &  assign 1 (0) if $\text{TTC}$ is below (above) $q$ seconds & TTC \emph{percentile rank} of score $q$ \rrev{(dual of percentile)} &  \rrev{$\prttc$} \\
HW & continuous & assign 1 (0) if $\text{HW}$ is below (above) $q'$ seconds &  HW \emph{percentile rank} of score $q'$ \rrev{(dual of percentile)} &  \rrev{$\prhw$} \\
\bottomrule
\end{tabular}
\end{table}

\rev{
The alignment of formal definitions based on total rewards with intuitive quantitative metrics is an important advantage of our approach.
This facilitates a switch designer to understand the numerical analysis results returned from multi-objective model checking (c.f., Section~\ref{sec:switch_analysis_in_case_study}). 
The PRISM specification of switch models in Listing~\ref{list:prism_model} and Appendix~\ref{list:prism_model_added} also include reward functions.
}

\subsubsection{\rev{Switch Design Analysis}}\label{sec:switch_analysis_in_case_study}
\rev{To analyse the trade-off between the four performance metrics (i.e., CTL, LD, TTC and HW), 
we conducted different types of analysis using multi-objective model checking: Direction evaluation, single-objective optimisation, achievability query and convex query (see Section~\ref{sec:formal_problem} for a description of these analyses).}


\rev{Table~\ref{tab:case_study_data} presents the analysis results obtained from a specific randomised delay setting, where the average delayed time is 0.2s.}  
The numerical data is computed by using Storm~\cite{Hensel2021}, PRISM~\cite{Kwiatkowska2011} and our new tool MOPMC\orev{, all of which produced identical results (subject to numerical rounding) for the same model checking problems (if supported).}
%

\begin{table}[tbh!] 
\centering
\caption{\orev{Analysis Results of Switch with Randomised Delays Averaging 0.2 Seconds (where $\prttc$ and $\prhw$ scores are $\pcttc$ (i.e.\ 14.4638s) and $\pchw$  (i.e.\ 3.9710s)}}
\label{tab:case_study_data}
\fontsize{9}{9.5}\selectfont 
\setlength{\tabcolsep}{3pt} 
\begin{tabular}{lp{6.3cm}ccccc} 
\toprule
\textbf{Analysis Type} & \textbf{Configuration} & $\expctl$ & $\prttc$ & $\prhw$ & $\probld$ \\
\midrule
Direct evaluation & (C1) switch to controller 1 at 0.1s & +0.0000 & +0.0369 & +0.1320 & +0.0004 \\
Direct evaluation & (C2) switch to controller 2 at 0.1s & +0.0843 & +0.0002 & +0.0018 & +0.0111 \\
Direct evaluation & (C3) switch to controller 1 at 0.2s & +0.0408 & +0.0355 & +0.1293 & +0.0000 \\
Direct evaluation & (C4) switch to controller 2 at 0.2s & +0.0851 & +0.0000 & +0.0067 & +0.0111 \\
Direct evaluation & Uniform combination of C1--C4       & +0.0631 & +0.0173 & +0.0637 & +0.0059 \\
Single-obj.\ optim. & Config.\ minimising \rrev{$\expctl$} & {1.8485} & - & - & - \\
Single-obj.\ optim. & Config.\ minimising \rrev{$\prttc$} & - & {0.0308} & - & - \\
Single-obj.\ optim. & Config.\ minimising \rrev{$\prhw$} & - & - & {0.1839} & - \\
Single-obj.\ optim. & Config.\ minimising \rrev{$\probld$} & - & - & - & {0.1639} \\
Achiev.\ query & \underline{Some} config.\ satisfying thresholds on the right. & $\leq$2.0 & $\leq$0.040 & $\leq$0.23 & $\leq$0.174 \\
Achiev.\ query & \underline{No} config.\ satisfying thresholds on the right. & $\leq$1.9 & $\leq$0.035 & $\leq$0.22 & $\leq$0.169 \\
Convex query & Config.\ minimising MSE of the (non-satisfiable) thresholds in the second achiev.\ query & +0.0515 & +0.0092 & +0.0350 & +0.0084\\
\bottomrule
\end{tabular}
\end{table}



\rev{
First, we used single-objective optimisations to compute the lowest possible values (i.e., baselines) of the four metrics across all combinations of the four basic configurations C1–C4. Direct evaluations using C1–C4 and their uniform combination reveal the inherent trade-off among these metrics. Specifically, C1 achieves the lowest values for $\expctl$ and $\probld$, C2 yields the lowest value for $\prhw$, and C4 minimises $\prttc$. This occurs because the optimal value of each individual metric can always be attained by one of the basic configurations. 
More formally, C1--C4 correspond to the vertices in the (four-dimensional) polytope of $\mdp_{\mathit{SW}2}$, which includes all possible value  for $\expctl$, $\prttc$, $\prhw$ and $\probld$.
Multi-objective model checking pursues their suitable trade-off within the polytope.
For example, by using a uniform combination of configurations C1 to C4, we observed a more balanced performance across the four metrics.

Next, to understand the thresholds and trade-off of those metrics, we selected two sets of thresholds for the four metrics and performed multi-objective achievability queries. The first achievability query asserts that there is a configuration to satisfy $\expctl\leq2.0$, $\prttc\leq0.04$, $\prhw\leq0.23$ and $\probld\leq0.174$ simultaneously, while the second query confirms that $\expctl\leq1.9$, $\prttc\leq0.035$, $\prhw\leq0.22$ and $\probld\leq0.169$ is non-achievable.

Third, in view of the above two achievability queries, we further examined how closely the four metric values can approach the above non-achievability thresholds. To this end, we performed a convex query using MOPMC (which implements Algorithm~\ref{alg:main-convex-query}) to compute the minimal Mean Squared Error (MSE) from those thresholds.
The minimal MSE is achieved at $\expctl = 1.9$, $\prttc = 0.0377$, $\prhw = 0.2189$, and $\probld = 0.1720$, obtained through a weighted combination of configurations C1 to C4. 
In other words, these four performance values correspond to a minimum-MSE point in the polytope of $\mdp_{\mathit{SW}2}$ (subject to numerical rounding errors).
Therefore, schedulers exported from convex queries usually achieve better design choices than arbitrary configuration strategies. 
Note that we adopted MSE as it is a common function to express trade-off, but other convex functions (such as variance and weighted MSE) are also applicable.

Overall, the results in Table~\ref{tab:case_study_data} demonstrate that our approach supports rigorous analysis of trade-off  between diverse performance metrics.}


\subsubsection{\rev{Delay Impact Mitigation Evaluation}}

\rev{The primary goal of implementing our switch-based control system is to mitigate the detrimental impact of delays on critical vehicle safety metrics. We applied switching configurations obtained from our design analysis to the DRL system and evaluated their effectiveness across a range of increasing average delays, specifically within the interval of $0.20 \pm 0.10$ seconds. The two representative configurations which we examined are: the uniform combination of the four basic switch configurations C1--C4, and the ``CQ-exported'' configuration exported from the multi-objective convex query in Table $\text{\ref{tab:case_study_data}}$. Performance differences are calculated relative to the ``no switch'' baseline, where only the standalone DDPG agent operates without secondary controllers.
For performance metrics, we considered $\pcttc$.  
For TTC, we considered percentiles $\pcttc$ and $\prttc$, percentile ranks $\prttc$ and $\prhw$, and the LD probability $\probld$. The corresponding scores of percentile ranks are indicated parenthetically.
These metrics are calculated in the same way as in Section~\ref{sec:time_delay_impact} (i.e., using statistics on simulation data directly), but different from that in Section~\ref{sec:switch_analysis_in_case_study} (through switch modelling and probabilistic model checking).
}

\rev{The results, detailed in Table~\ref{tab:case_study_data_more}, reveal a consistent trade-off employed by the switching policies to manage increased latency. The baseline (``no switch'') data is adopted and (for $\prttc$ and $\prhw$) derived from Table~\ref{tab:time_delay_impact}. As the average delay increases, both the {C1--C4 uniform combination} and the $\text{CQ-exported}$ configurations demonstrate a significant improvement in longitudinal safety. The risk metrics $\prttc$ and $\prhw$ show substantial decrement (e.g., up to {-0.0867} for $\prttc$ and {-0.3575} for $\prhw$ at {0.30s} delay). This improved longitudinal safety is evidenced by the large positive change in the average headway $\pchw$ and the corresponding large positive change in $\pcttc$ at high delays (e.g., +\num{24.9239} for $\pcttc$ at {0.30s}). The cost of using the switch in the current configurations is a systematic increase in the lateral risk, as shown by the consistent positive difference in the $\probld$ metric across all delay levels. It is also important to observe that, except for average delay with \num{0.1} seconds, the CQ-exported configuration achieves a better trade-off across the five performance metrics compared to the simple $\text{C1--C4 uniform}$ configuration.}

\begin{table}[tbh!] 
\centering
\caption{\rev{Performance Comparison w/o Secondary Controllers with Varying Average Delay Time (where ``CQ-exported'' refers to the config.\ exported from the convex query in the last row of Table~\ref{tab:case_study_data})}}
\label{tab:case_study_data_more}
\fontsize{9}{9}\selectfont 
\setlength{\tabcolsep}{3pt} 
\begin{tabular}{ccccccc} 
\toprule
\textbf{Config.\ } & \textbf{Avg delay} & $\pcttc$ & $\prttc$ (\textbf{score}) & $\pchw$ & $\prhw$ (\textbf{score}) & $\probld$ \\
\midrule 
No switch & \num{0.10} & \num{17.5975} & \num{0.1000} (\num{17.5975}) & \num{3.8275} & \num{0.5000} (\num{3.8275}) & \num{0.0063} \\
C1--C4 unif.\ com. & \num{0.10} & \textbf{-\num{0.7655}} & \textbf{+\num{0.0080}} (\num{17.5975}) & \textbf{+\num{0.3448}} & \textbf{-\num{0.0703}} (\num{3.8275}) & \textbf{+\num{0.0305}} \\
CQ-exported & \num{0.10} & \textbf{-\num{0.3378}} & \textbf{+\num{0.0036}} (\num{17.5975}) & \textbf{+\num{0.4255}} & \textbf{-\num{0.0869}} (\num{3.8275}) & \textbf{+\num{0.0348}} \\
\midrule
No switch & \num{0.15} & \num{15.4452} & \num{0.1000} (\num{15.4452}) & \num{3.7491} & \num{0.5000} (\num{3.7491}) & \num{0.0354} \\
C1--C4 unif.\ com.& \num{0.15} & \textbf{+\num{2.8626}} & \textbf{-\num{0.0255}} (\num{15.4452}) & \textbf{+\num{1.0036}} & \textbf{-\num{0.1873}} (\num{3.7491}) & \textbf{+\num{0.0586}} \\
CQ-exported & \num{0.15} & \textbf{+\num{3.9778}} & \textbf{-\num{0.0346}} (\num{15.4452}) & \textbf{+\num{1.2021}} & \textbf{-\num{0.2185}} (\num{3.7491}) & \textbf{+\num{0.0642}} \\
\midrule
No switch & \num{0.20} & \num{14.4638} & \num{0.1000} (\num{14.4638}) & \num{3.9710} & \num{0.5000} (\num{3.9710}) & \num{0.1152} \\
C1--C4 unif.\ com. & \num{0.20} & \textbf{+\num{7.1221}} & \textbf{-\num{0.0511}} (\num{14.4638}) & \textbf{+\num{1.6159}} & \textbf{-\num{0.2496}} (\num{3.9710}) & \textbf{+\num{0.0543}} \\
CQ-exported & \num{0.20} & \textbf{+\num{9.2413}} & \textbf{-\num{0.0598}} (\num{14.4638}) & \textbf{+\num{1.8840}} & \textbf{-\num{0.2804}} (\num{3.9710}) & \textbf{+\num{0.0570}} \\
\midrule
No switch & \num{0.25} & \num{13.3526} & \num{0.1000} (\num{13.3526}) & \num{4.1480} & \num{0.5000} (\num{4.1480}) & \num{0.1923} \\
C1--C4 unif.\ com. & \num{0.25} & \textbf{+\num{12.6027}} & \textbf{-\num{0.0717}} (\num{13.3526}) & \textbf{+\num{2.2751}} & \textbf{-\num{0.2924}} (\num{4.1480}) & \textbf{+\num{0.0707}} \\
CQ-exported & \num{0.25} & \textbf{+\num{16.0426}} & \textbf{-\num{0.0790}} (\num{13.3526}) & \textbf{+\num{2.6359}} & \textbf{-\num{0.3224}} (\num{4.1480}) & \textbf{+\num{0.0757}} \\
\midrule
No switch & \num{0.30} & \num{12.7368} & \num{0.1000} (\num{12.7368}) & \num{4.2904} & \num{0.5000} (\num{4.2904}) & \num{0.2849} \\
C1--C4 unif.\ com. & \num{0.30} & \textbf{+\num{19.7276}} & \textbf{-\num{0.0827}} (\num{12.7368}) & \textbf{+\num{3.0990}} & \textbf{-\num{0.3362}} (\num{4.2904}) & \textbf{+\num{0.0713}} \\
CQ-exported& \num{0.30} & \textbf{+\num{24.9239}} & \textbf{-\num{0.0867}} (\num{12.7368}) & \textbf{+\num{3.4980}} & \textbf{-\num{0.3575}} (\num{4.2904}) & \textbf{+\num{0.0718}} \\
\bottomrule
\end{tabular}
\end{table}

\rev{
\subsubsection{Sensitivity Analysis}

As we estimated the transition probabilities based on prescribed distributions and using simulation data, the probabilities may be inaccurate. To evaluation the robustness of our probabilistic model checking (particularly the convex query) results against the transition probability perturbations, we performed an empirical sensitivity analysis. We perturbed each transition probability by up to $\pm15\%$. More specifically, the perturbation levels range from $0.015$ to $0.15$ (corresponding to $\pm 1.5\%$ to $\pm 15\%$ change in probabilities).

The performance results for the four metrics  $\pcttc$, $\prttc$ (with a {score} of 14.4638s), $\pchw$ and $\prhw$ (with a {score} of 3.9710s) and $\probld$ returned from convex queries are presented in Figure~\ref{fig:sens_analy}.
Each box plot summarises a set of 40 data points (derived from increasing or decreasing 20 transition probabilities by the perturbation level). The dashed red line indicates the unperturbed baseline performance for direct comparison.

Overall, there is a general trend of increasing variability (i.e., wider boxes and longer whiskers) proportional to the perturbation scale for most metrics. However, the main finding is that the convex query results are still stable, as the median performance across all metrics never significantly deviates from the baseline.
}

\begin{figure}[!htbp]
\begin{center}
\includegraphics[width=0.95\textwidth]{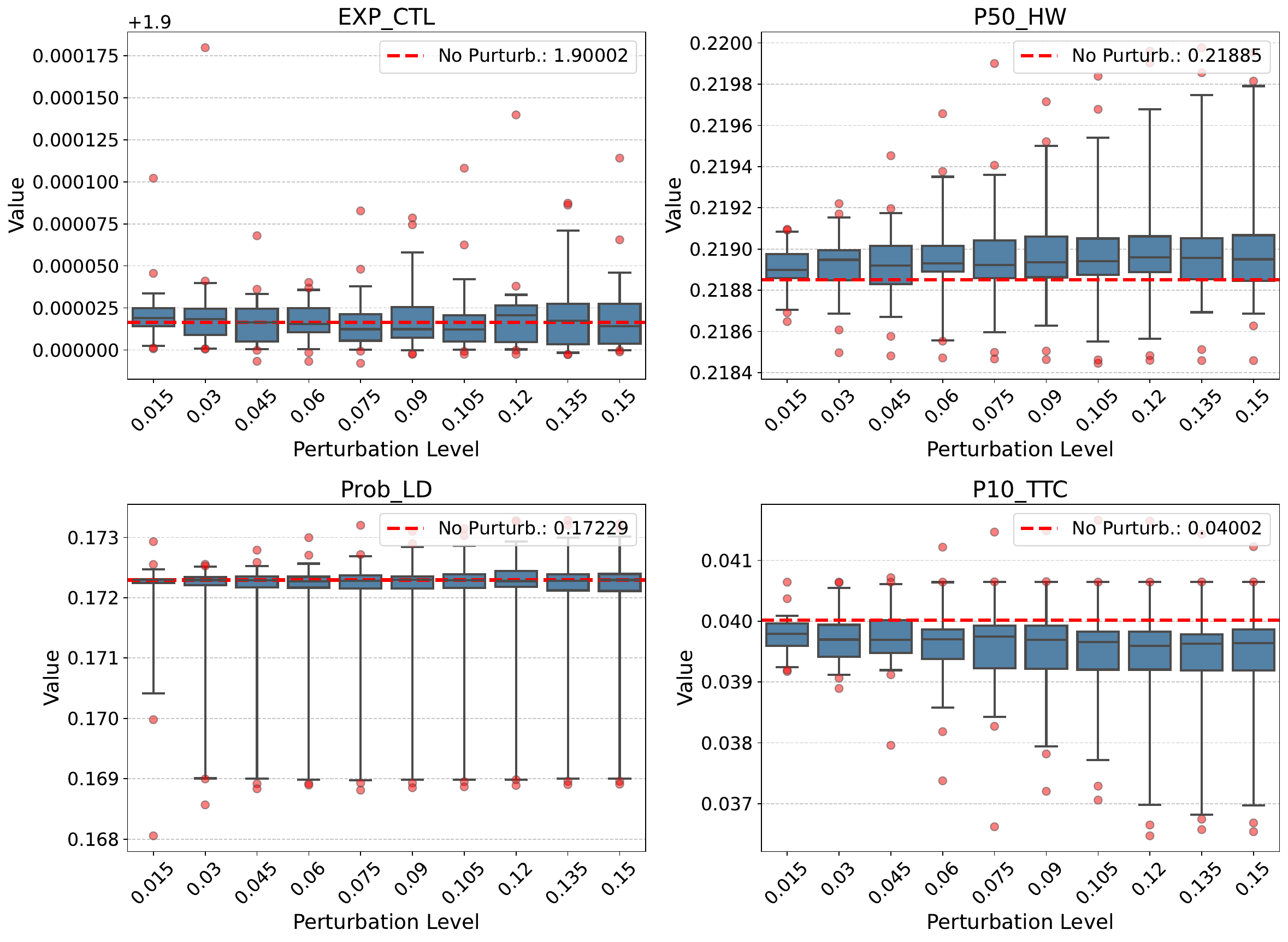}
\caption{\rev{Sensitivity Analysis of Convex Query Outputs}}
\label{fig:sens_analy}
\end{center}
\end{figure}

\subsection{MOPMC Runtime Performance {(RQ3)} }

\begin{table*}[t] 
\centering
\caption{Runtime Performance of MOPMC, Storm, and PRISM (Timeout: 1000s)}\label{tab:performance_results}
\fontsize{8}{9}\selectfont 
\setlength{\tabcolsep}{2pt} 
\newcolumntype{P}[1]{>{\centering\arraybackslash}p{#1}} 
\begin{tabular}{@{} cccc cccc c c  c c cc @{}}%
\toprule
\multicolumn{12}{l}{\textbf{Model A (Switch Model): Varying Model Size (Number of Objectives: 4)}} \\ 
\midrule
\multirow{4}{*}{\makecell{Time- \\ step \\ (max)}} & \multirow{4}{*}{States} & \multirow{4}{*}{Choices} & \multirow{4}{*}{\makecell{Trans.}} & \multicolumn{4}{c}{Achievability Query} & \multicolumn{4}{c}{Convex Query} \\ 
\cmidrule(lr){5-8} \cmidrule(lr){9-12} %
& & & & \multicolumn{4}{c}{\makecell{Runtime}} & \multicolumn{2}{c}{\makecell{Runtime (Loop count)}} & \multicolumn{2}{c}{GPU vs CPU}\\ %
\cmidrule(lr){5-8} \cmidrule(lr){9-10} \cmidrule(lr){11-12} %
& & & & \makecell{Storm} & \makecell{PRISM} & \makecell{MOPMC \\ GPU} & \makecell{MOPMC \\ CPU} & \makecell{MOPMC \\ GPU} & \makecell{MOPMC \\ CPU} & {\makecell{E2E\\ time}} & {\makecell{per-iter\\ time }} \\
\midrule
20  & 19317   & 19377   & 68033   & 0.23 & 1.32 & 0.35 & 0.17 & 0.68 (42) & 1.00 (40) & $32\%\downarrow$ & $35\%\downarrow$  \\ 
40  & 76937   & 77057   & 276673  & 0.61 & 1.45 & 0.45 & 0.49 & 2.09 (52) & 6.83 (40) & $69\%\downarrow$ & $76\%\downarrow$ \\ 
60  & 172957  & 173137  & 626113  & 1.31 & 3.41 & 0.94 & 1.25 & 3.93 (33) & 100 (200) & $96\%\downarrow$ & $76\%\downarrow$ \\ 
80  & 307377  & 307617  & 1116353 & 2.43 & 7.33 & 1.70 & 2.70 & 12.48 (55) & 96.41 (71) & $87\%\downarrow$ & $83\%\downarrow$ \\ 
100 & 480197  & 480497  & 1747393 & 3.91 & 13.20 & 2.73 & 5.08 & 75.54 (200) & 576.02 (200) & $87\%\downarrow$ & $87\%\downarrow$ \\ 
120 & 691417  & 691777  & 2519233 & 6.19 & 22.37 & 4.27 & 8.62 & 122.60 (200) & 122.50 (21) & $0\%$ & $89\%\downarrow$ \\ 
140 & 941037  & 941457  & 3431873 & 8.77 & 34.12 & 6.14 & 13.73 & 28.57 (24) & 187.79 (20) & $85\%\downarrow$ & $87\%\downarrow$ \\
160 & 1229057 & 1229537 & 4485313 & 12.01 & 50.16 & 8.82 & 20.65 & 36.87 (21) & 267.56 (18) & $86\%\downarrow$ & $88\%\downarrow$ \\
180 & 1555477 & 1556017 & 5679553 & 16.06 & 70.68 & 12.07 & 29.16 & 73.93 (33) & 703.05 (35) & $89\%\downarrow$ & $89\%\downarrow$ \\ 
200 & 1920297 & 1920897 & 7014593 & 20.93 & 95.87 & 16.17 & 40.37 & 92.32 (30) & 740.68 (26) & $88\%\downarrow$ & $89\%\downarrow$ \\
\midrule
\midrule 
\multicolumn{12}{l}{\textbf{Model B: Varying Number of Objectives (States: 15051, Choices: 30001, Trans: 89203)}} \\ 
\midrule
\multirow{4}{*}{\makecell{Num.\\ obj.}} & \multirow{4}{*}{States} & \multirow{4}{*}{Choices} & \multirow{4}{*}{\makecell{Trans.}} & \multicolumn{4}{c}{Achievability Query} & \multicolumn{4}{c}{Convex Query} \\ 
\cmidrule(lr){5-8} \cmidrule(lr){9-12} %
& & & & \multicolumn{4}{c}{\makecell{Runtime}} & \multicolumn{2}{c}{\makecell{Runtime (Loop count)}} & \multicolumn{2}{c}{GPU vs CPU}\\ %
\cmidrule(lr){5-8} \cmidrule(lr){9-10} \cmidrule(lr){11-12} %
& & & & \makecell{Storm} & \makecell{PRISM} & \makecell{MOPMC \\ GPU} & \makecell{MOPMC \\ CPU} & \makecell{MOPMC \\ GPU} & \makecell{MOPMC \\ CPU} & {\makecell{E2E\\ time}} & {\makecell{per-iter\\ time }} \\
\midrule
6   & \multicolumn{3}{c}{} & 0.81 & 1.52 & 0.35 & 0.43 & 0.24 (3) & 0.24 (3) & $0\%$ & $0\%$ \\
7   & \multicolumn{3}{c}{} & 3.94 & 1.95 & 0.41 & 0.53 & 0.67 (22) & 0.29 (3) & $56\%\uparrow$ & $69\%\downarrow$ \\ 
8   & \multicolumn{3}{c}{} & 24.63 & 2.31 & 0.37 & 0.58 & 0.35 (3) & 1.45 (20) & $76\%\downarrow$ & $61\%\uparrow$ \\ 
9   & \multicolumn{3}{c}{} & 150.78 & 2.74 & 0.38 & 0.81 & 0.50 (15) & 0.32 (3) & $56\%\uparrow$ & $69\%\downarrow$ \\
10  & \multicolumn{3}{c}{} & \textsc{t.o.} & 3.02 & 0.38 & 0.80 & 0.72 (24) & 1.89 (23) & $62\%\downarrow$ & $63\%\downarrow$ \\ 
20  & \multicolumn{3}{c}{} & \textsc{t.o.} & 10.49 & 0.81 & 2.29 & 0.96 (18) & 2.86 (22) & $66\%\downarrow$ & $59\%\downarrow$ \\ 
40  & \multicolumn{3}{c}{} & \textsc{t.o.} & 29.79 & 1.38 & 1.47 & 2.78 (45) & 10.44 (46) & $73\%\downarrow$ & $73\%\downarrow$ \\
60  & \multicolumn{3}{c}{} & \textsc{t.o.} & 24.93 & 1.42 & 4.08 & 3.78 (45) & 12.60 (40) & $70\%\downarrow$ & $73\%\downarrow$ \\
80  & \multicolumn{3}{c}{} & \textsc{t.o.} & 41.58 & 1.80 & 2.52 & 7.33 (67) & 34.96 (82) & $79\%\downarrow$ & $74\%\downarrow$ \\
100 & \multicolumn{3}{c}{} & \textsc{t.o.} & 49.72 & 2.21 & 7.82 & 26.49 (200) & 103.39 (194) & $74\%\downarrow$ & $75\%\downarrow$ \\
\midrule
\midrule\multicolumn{12}{l}{\textbf{Model C (Team Formulation Protocol): Varying Model Size and Objective Number}} \\ 
\midrule
\multirow{4}{*}{\makecell{Num.\\ obj.}} & \multirow{4}{*}{States} & \multirow{4}{*}{Choices} & \multirow{4}{*}{\makecell{Trans.}} & \multicolumn{4}{c}{Achievability Query} & \multicolumn{4}{c}{Convex Query} \\ 
\cmidrule(lr){5-8} \cmidrule(lr){9-12} %
& & & & \multicolumn{4}{c}{\makecell{Runtime}} & \multicolumn{2}{c}{\makecell{Runtime (Loop count)}} & \multicolumn{2}{c}{GPU vs CPU}\\ %
\cmidrule(lr){5-8} \cmidrule(lr){9-10} \cmidrule(lr){11-12} %
& & & & \makecell{Storm} & \makecell{PRISM} & \makecell{MOPMC \\ GPU} & \makecell{MOPMC \\ CPU} & \makecell{MOPMC \\ GPU} & \makecell{MOPMC \\ CPU} & {\makecell{E2E\\ time}} & {\makecell{per-iter\\ time }} \\
\midrule
2  & 12475 & 14935 & 15228 & 0.10 & 2.27 & 0.44 & 0.16 & 0.45 (16) & 0.12 (16) & $275\%\uparrow$ & $275\%\uparrow$ \\
4  & 12475 & 14935 & 15228 & 0.16 & 4.41 & 0.38 & 0.11 & 0.38 (19) & 0.12 (19) & $217\%\uparrow$ & $217\%\uparrow$ \\
7  & 12475 & 14935 & 15228 & 0.14 & 5.03 & 0.30 & 0.15 & 0.39 (17) & 0.14 (17) & $179\%\uparrow$ & $179\%\uparrow$ \\
10 & 12475 & 14935 & 15228 & 0.20 & 5.61 & 0.37 & 0.16 & 0.69 (200) & 0.55 (200) & $25\%\uparrow$ & $25\%\uparrow$ \\
2  & 96665 & 115289 & 116464 & 0.42 & 62.07 & 0.73 & 0.42 & 0.63 (12) & 0.49 (12) & $29\%\uparrow$ & $29\%\uparrow$ \\
4  & 96665 & 115289 & 116464 & 0.52 & 131.29 & 0.71 & 0.49 & 0.73 (16) & 1.28 (81) & $43\%\downarrow$ & $189\%\uparrow$ \\
7  & 96665 & 115289 & 116464 & 8.03 & 174.34 & 0.87 & 0.67 & 1.04 (98) & 2.32 (122) & $55\%\downarrow$ & $44\%\downarrow$ \\
10 & 96665 & 115289 & 116464 & \textsc{t.o.} & 190.68 & 1.37 & 1.46 & 1.46 (40) & 2.20 (57) & $34\%\downarrow$ & $5\%\downarrow$ \\
2  & 907993 & 1078873 & 1084752 & 5.60 & \textsc{t.o.} & 5.38 & 5.19 & 5.45 (14) & 6.35 (14) & $14\%\downarrow$ & $14\%\downarrow$ \\
4  & 907993 & 1078873 & 1084752 & 6.83 & \textsc{t.o.} & 5.98 & 6.03 & 6.09 (15) & 7.45 (14) & $18\%\downarrow$ & $24\%\downarrow$ \\
7  & 907993 & 1078873 & 1084752 & 23.03 & \textsc{t.o.} & 7.66 & 9.93 & 8.78 (63) & 13.31 (31) & $34\%\downarrow$ & $68\%\downarrow$ \\
10 & 907993 & 1078873 & 1084752 & \textsc{t.o.} & \textsc{t.o.} & 12.76 & 16.00 & 14.16 (60) & 46.13 (130) & $69\%\downarrow$ & $33\%\downarrow$ \\
\bottomrule
\end{tabular}
\end{table*}

\rev{To evaluate the runtime performance of our tool MOPMC, we designed two MDP models. Model A is derived from the switch model in Section~\ref{sec:eval:swtich_design_analysis} by adjusting the maximum timestep. Model B is specialised model that enables the specification of up to 100 objectives.
We also adopted the Team Formulation Protocol model (Model C) which is a multi-objective MDP benchmark used by both PRISM and Storm previously. All the model specifications in the PRISM language are available in MOPMC's GitHub repository.}

Table~\ref{tab:performance_results} showcases the performance data of MOPMC, benchmarked against Storm and PRISM across \rev{the three models}. 
\rev{Since MOPMC employed the sparse engine of Storm for model building, for a fair comparison the results, we set the engines of PRISM and Storm to ``sparse engine'' in the evaluations.} 

For achievability queries on Model A, MOPMC (GPU) consistently outperforms MOPMC (CPU), Storm, and PRISM, especially with larger models. This suggests MOPMC's GPU-accelerated implementation superior scalability with increasing model size.
%
For model B, both MOPMC (GPU) and MOPMC (CPU) demonstrate significantly better scalability than Storm and PRISM. 
\rev{For Model C, the performance data show that, except for small model size and objective number, both MOPMC (GPU) and MOPMC (CPU) can outperform both Storm and PRISM.}

Convex queries, whilst generally taking longer than achievability queries due to its more complex nature. Their runtime is influenced by the number of iterations required by the main query loop (with a maximum of 200 iterations), as well as the computational cost within each iteration. The most computationally intensive part of each main query algorithm is the computation Line~\ref{ln:find-supp-hp} of Algorithm~\ref{alg:main-convex-query}. Given that value iterations can be performed in parallel in our algorithm (as explained previously in Section~\ref{sec:parallel_pvi}), the computational time can directly benefit from GPU acceleration. \rrev{Except for small models}, MOPMC (GPU) are evident in faster execution \rrev{in items of both the end-to-end time and the time per iteration}.
\rev{Note that the time per iteration is a rough runtime metric. This is because, prior to  policy-value iterations, probabilistic model checkers execute graph algorithms, which can be more computationally intensive than the iterations themselves.}
Future work will focus on designing more sophisticated and adaptive termination conditions for MOPMC to further optimise its performance, building upon evaluation results observed.

It is worthy of notice that, compared with model size scalability, objective number scalability has not been evaluated sufficiently the multi-objective MDP model checking literature. 
\rev{However, the presence of numerous evaluation metrics reported in the literature~\cite{Delavari2025} underscores the need to address scalability with respect to the number of objectives.}  
This highlights a key strength and novel contribution of MOPMC in handling problems in which numerous objectives must be considered.

\subsection{Discussion}
\rev{
\subsubsection{Threats to Validity for \textbf{RQ1} Evaluation}
Before discussing possible threats to validity for \textbf{RQ1}, we first revisit the two domain assumptions that motivate our approach: (i) the occurrence of time delay in a real-time environment, and (ii) the detrimental impact of time delay on control performance (both of which are explained in Section~\ref{sec:rt_drl_attributes}). 
Although our evaluation uses autonomous driving as a representative example, we believe that these assumptions are realistic across other dynamic domains. Factors such as sensor acquisition lag, network congestion, and computational bottlenecks can all contribute to delay. 

To further clarify these assumptions and their implications, we consider two scenarios for an DRL agent, one with deterministic (tick-tock) timing and the other with random delay.
In both scenarios, the agent produces some action at time $t$, and after a fixed time interval $\Delta$, an observation $o$ is generated. 
In the tick-tock scenario, the agent processes $o$ and produces another optimal action $a$ before the environment evolves. 
In the delayed scenario, the action processes $o$ at time $t+\Delta+\epsilon$ where $\epsilon$ is an unknown data processing delay, and 
produces an action $a'$ at time $t+\Delta+\epsilon+\epsilon'$
Because the agent is trained under tick-tock conditions, $a'=a$.
But the environment may differ significantly from that at time $t+\Delta$, rendering $a$ sub-optimal or even unsafe.

Our evaluation for \textbf{RQ1} (Table~\ref{tab:time_delay_impact}) demonstrates the detrimental impact of time delay on a trained DDPG agent in the HighwayEnv environment. Several \emph{external threats} may qualify this finding:
\begin{itemize}
\item (Delay magnitude) Our evaluation shows that the detrimental impact becomes obvious from \num{10}\% of $\Delta$ (\num{0.1}s) as the (average) delay. This setting is invalid for domains where the delayed time $\epsilon+\epsilon'$ is negligible.
\item (Domain dynamics) If the environment is less dynamic (such that the environment change is minimal), then $a$ is still optimal in the delayed setting and thus performance degradation will not be observed.
\item (Agent adaptability)  Standard DRL algorithms are not latency-aware; they train agents with respective to loss functions defined for discrete training steps. However, an agent capable of adapting to elapsed time would mitigate delay effects by itself.
\item (Agent quality) If the trained agent is not optimal even in the non-delayed setting, further deterioration under delay may not be evident. Hence, a highly trained agent is essential for meaningful evaluation.
\end{itemize}

One \emph{construct threat} to \textbf{RQ1} concerns the choice of performance metrics. We used the 10th percentile as the TTC threshold,  the median as the HW threshold and the lane departure probability. Different metrics and threshold selections may lead to variations in observed performance.
}

\rev{
\subsubsection{Threats to Validity for \textbf{RQ2} Evaluation}
For the evaluation of RQ2, in addition to the threats discussed earlier, an additional \emph{external threat} arises from the presence of \emph{secondary controllers} that can compensate for critical performance aspects. In practice, it is important to consider whether such controllers are straightforward to design, and whether they require extensive performance tuning.
In our evaluation, we implemented several variants of a rule-based controller that only utilise HW and TTC for decision making. These controllers effectively enhance longitudinal safety but at the expense of lateral stability, as evidenced in Table~\ref{tab:case_study_data_more}.

\textbf{RQ2} is also subject to \emph{internal threats} related to model validity, particularly the structural and parameter accuracy of the underlying MDP model. In Section~\ref{sec:syntactic_mdp}, our syntactic conversion from a TA to an MDP relies on the well-known Markovian assumption, which states that the next state depends only on the current state and not on the sequence of previous state. For example, In Figure~\ref{fig:mdp_sw} the transition probabilities from $s_\mathit{c}$ are conditional solely on $s_\mathit{c}$ regardless of the transition history. Furthermore, our approach determines transition probabilities using a prescribed probability distribution and sampling, both of which may introduce inaccuracies.  These factors constitute internal threats to the accuracy of the model checking results presented in Table~\ref{tab:case_study_data} and Table~\ref{tab:case_study_data_more}.
To understand the impact by these inaccuracies, we have conducted a sensitivity analysis of the convex queries. Findings in Figure~\ref{fig:sens_analy} demonstrated stability of query outputs against model perturbations.

Rigorous methods for sensitivity analysis in (single-objective) probabilistic model checking are well established. As optimal schedulers are synthesised, sensitivity analysis for MDPs reduces to that for discrete-time Markov chains (DTMCs). The first method computes the first- and second-order derivatives of uncertain transition probabilities for a given DTMC model checking problem~\cite{Su2022,Su2016}. The second method, called parametric model checking, derives closed-form algebraic expressions for DTMC model-checking outputs~\cite{Daws2005,Hahn2010}. The third method computes worst-case bounds by assuming interval-valued transition probabilities in DTMCs~\cite{Puggelli2013}. 
Incorporating one or more of these rigorous methods into our approach is left for future work.
}


It is noteworthy that, despite these threats, the core value of our framework lies in its ability to provide a formal and expressive approach to supporting the switch design and optimisation process. The analysis results can furnish the switch designer with evidence to verify their assumptions, and offer clear metrics to determine and balance multiple factors to achieve desirable design solutions.

\rev{
\subsubsection{Simulation of Real-Time Behaviour in Other DRL Environments}
Although HighwayEnv is selected for our evaluations addressing \textbf{RQ1} and \textbf{RQ2}, our method of injecting fixed or randomised time delays, as described in Section~\ref{sec:exp_setup}, is in principle transferable to other control simulators (e.g., CARLA). We acknowledge that employing a DRL environment from a different domain may influence evaluation outcomes; nevertheless, implementing real-time interactions between an agent and environments subject to fixed or randomised delays remains highly feasible. In our setup, elapsed time for the agent’s observations ($\epsilon$) and actions ($\epsilon'$) is simulated by updating the environment accordingly. Alternative strategies for modelling real-time behaviour in DRL environments have also been explored in prior work~\cite{Ramstedt2019,Thodoroff22a}.
}

\section{Related Work}\label{sec:related_work}

\subsection{Fault-tolerance for DRL Systems}

Fault-tolerance for Deep Reinforcement Learning (DRL) systems is primarily achieved through two architectural approaches: Simplex-based designs that switch between controllers, and a protection component (often referred to as a ``shield'') that correct unsafe actions.

The Simplex paradigm leverages controller redundancy. 
Phan~\textit{et~al.}~\cite{Phan2020} introduced the Neural Simplex Architecture, which uses a verified safe controller to correct a DRL agent's behaviour at runtime. This concept has been applied extensively to autonomous driving. 
Chen~\textit{et~al.}~\cite{Shengduo2022} proposed Simplex-Drive, which combines a high-performance DRL controller with a provably safe, Velocity-Obstacle-based controller, using a verified management unit to switch between them and guarantee safety. 
Maderbacher~\textit{et~al.}~\cite{Maderbacher2025} focused on formally synthesising the switching logic itself,
ensuring correctness against temporal specifications while maximising the use of the high-performance controller. 
Cai~\textit{et~al.}~\cite{Cai2025} presented a runtime learning machine with a Student-Teacher architecture, where a high-assurance Teacher controller provides real-time corrections, enabling the Student DRL agent to learn safely and continuously while addressing the Sim-to-Real gap.


The second paradigm involves synthesising a protective shield or enforcer. Alshiekh~\textit{et~al.}~\cite{Alshiekh2018}  first proposed synthesising shields from temporal logic specifications to monitor and minimally correct an agent's actions to prevent safety violations. Jansen~\textit{et~al.}~\cite{Jansen2020} extended this to stochastic environments by developing probabilistic shields that provide safety guarantees under uncertainty. 
More recently, Yang~\textit{et~al.}~\cite{Yang2023} introduced Probabilistic Logic Shields, which use probabilistic logic programming to create a differentiable shield that can be integrated directly into the DRL training process. 
In a similar vein, Vuppala~\textit{et~al.}~\cite{Vuppala2024} used a runtime enforcer synthesised from formal policies to secure a DRL-based pacemaker, while Dunlap~\cite{Dunlap2025} employed a Run Time Assurance (RTA) system using control barrier functions to filter commands for an autonomous spacecraft inspection controller.

However, the existing literature does not address the unique challenges of fault-tolerance within real-time computing settings, nor does it provide a formal analysis for managing potentially conflicting safety and performance objectives under such constraints.

\subsection{Real-Time DRL}

Adapting DRL for real-time settings is a significant research area, which focuses on mitigating performance loss from computational and environmental latencies. 
Ramstedt and Pal~\cite{Ramstedt2019} developed a Real-Time Actor-Critic algorithm to counteract the sub-optimality of standard DRL methods where system states evolve faster than action computations. Similarly, Bouteille~\textit{et al.}~\cite{bouteiller2020reinforcement} proposed the Delay-Correcting Actor-Critic to specifically handle random environmental delays. 
To systematically evaluate these challenges, Thodoroff~\textit{et al.}~\cite{Thodoroff22a} introduced a benchmark for analysing the performance trade-offs of control and RL algorithms in time-sensitive robotics.
Li~\textit{et al.}~\cite{Li2023a} created the $R^3$ framework, which uses dynamic batch sizing and a runtime coordinator to balance timing, memory, and performance. Nhu~\textit{et al.}~\cite{nhu2025timeaware} introduced the Time-Aware World Model, a model-based DRL approach that explicitly incorporates temporal dynamics by training over a diverse range of time-step sizes. This method enables learning both high- and low-frequency task dynamics, leading to improved performance and data efficiency across various control problems.

These approaches all focus on enhancing real-time performance by modifying the internal DRL algorithm or its training process. By contrast, our work, along with the fault-tolerant methods previously discussed, employs external architectural mechanisms to provide formal performance guarantees and optimality.


\subsection{Multi-Objective Model Checking}\label{sec:related_work:mopmc}

Multi-objective model checking is a well-established field within probabilistic model checking~\cite{Baier2019,Andriushchenko2024}. 
Early methods from Chatterjee~\textit{et~al.}~\cite{Chatterjee2006} and Etessami~\textit{et~al.}~\cite{Etessami2008} relied on Linear Programming (LP) to simultaneously optimise MDPs with multiple discounted rewards and Linear Temporal Logic (LTL) properties, respectively.
 A significant advance came from Forejt~\textit{et~al.}~\cite{Forejt2012}, who introduced a more efficient value-iteration technique solve
achievability, numerical and Pareto queries over total rewards, LTL properties, and step-bounded properties.

Subsequent research expanded the scope of solvable problems. Multiple works focused on handling mean-payoff objectives, which were later extended by Quatmann~\textit{et~al.}~\cite{Quatmann2021} to continuous-time MDPs (Markov Automata). 
Hahn~\textit{et~al.}~\cite{Hahn2017} and Delgrange~\textit{et~al.}~\cite{Delgrange2020} addressed MDPs with uncertainty, developing methods for interval-valued and uncertain transition probabilities. Other specialised queries that were explored include solving inverse queries for multiple cost bounds by Hartmanns~\textit{et al.}~\cite{Hartmanns2020} and analysing multiple queries with deterministic and finite-memory schedulers by Delgrange~\textit{et~al.}~\cite{Delgrange2020}. 
Robinson and Su~\cite{Robinson2023} proposed a method to find the nearest achievable point to a target objective vector for multi-agent task allocation and planning.
More recently, Watanabe~\textit{et~al.}~\cite{Watanabe2024a,Watanabe2024b} proposed new compositional verification methods for MDPs that utilise multi-objective model checking to build Pareto curves (similar to Pareto queries).
Bals~\textit{et~al.}~\cite{Bals2024} introduced MultiGain, a major update of the tool presented by Brazdil~\textit{et~al.}~\cite{Brazdil2015}, which handles mean pay-off queries with LTL and steady-state probability constraints.

Joining the effort of the aforementioned literature, our novel approach introduces constrained convex optimisation of total rewards, which advances the multi-objective model checking techniques and implementation for MDPs.

\section{Conclusions and Future Work}\label{sec:conclusions}

In this work, we proposed a formal framework to design and analyse dependable real-time DRL systems. Our approach uses Timed Automata (TAs) to formally model a controller switching mechanism, which is then converted to a Markov Decision Process (MDP) syntactically for rigorous analysis. We introduced a novel convex query technique for multi-objective MDPs that optimises performance objectives while ensuring hard safety constraints. To demonstrate the practicality and scalability of our framework, we presented MOPMC, a GPU-accelerated tool that implements our proposed technique.

For future work, we plan to extend our convex query technique to a broader range of temporal properties, such as mean-payoff objectives. Another promising direction is to employ learning-based approaches to construct the MDP model from data, avoiding the need for explicit parameter estimation in the model. 
\rev{Additionally, we plan to conduct experiments in other environments and domains (such as the CARLA simulator). Another direction is to develop a sensitivity analysis method for our approach, which quantifies how perturbations in the model's transition probabilities influence the convex query output.}
Finally, we aim to integrate our framework into the reinforcement learning lifecycle, providing a formal feedback loop to retrain and improve the DRL agent's performance.

\ifacm
\begin{acks}
This work was partially supported by the National Natural Science Foundation of China (Grant Nos. 62477004, 62377040), Chongqing Municipal Economy and Information Technology Commission (YJX-2025001001008).
\end{acks}
\fi


\ifacm
\bibliographystyle{ACM-Reference-Format}
\fi
\bibliography{bib1,bib2}

@InCollection{Hahn2017,
  author    = {Ernst Moritz Hahn and Vahid Hashemi and Holger Hermanns and Morteza Lahijanian and Andrea Turrini},
  booktitle = {Quantitative Evaluation of Systems},
  publisher = {Springer International Publishing},
  title     = {Multi-objective Robust Strategy Synthesis for~Interval Markov Decision Processes},
  year      = {2017},
  pages     = {207--223},
  doi       = {10.1007/978-3-319-66335-7_13},
}

@Article{Etessami2008,
  author    = {Kousha Etessami and Marta Kwiatkowska and Moshe Vardi and Mihalis Yannakakis},
  journal   = {Logical Methods in Computer Science},
  title     = {Multi-Objective Model Checking of Markov Decision Processes},
  year      = {2008},
  month     = {nov},
  number    = {4},
  volume    = {4},
  doi       = {10.2168/lmcs-4(4:8)2008},
  editor    = {Michael Huth},
  publisher = {Centre pour la Communication Scientifique Directe ({CCSD})},
}

@InCollection{Chatterjee2006,
  author    = {Krishnendu Chatterjee and Rupak Majumdar and Thomas A. Henzinger},
  booktitle = {{STACS} 2006},
  publisher = {Springer Berlin Heidelberg},
  title     = {Markov Decision Processes with Multiple Objectives},
  year      = {2006},
  pages     = {325--336},
  doi       = {10.1007/11672142_26},
}

@InCollection{Quatmann2021,
  author    = {Tim Quatmann and Joost-Pieter Katoen},
  booktitle = {Tools and Algorithms for the Construction and Analysis of Systems},
  publisher = {Springer International Publishing},
  title     = {Multi-objective Optimization of Long-run Average and Total Rewards},
  year      = {2021},
  pages     = {230--249},
  doi       = {10.1007/978-3-030-72016-2_13},
}

@InCollection{Forejt2011,
  author    = {Vojt{\v{e}}ch Forejt and Marta Kwiatkowska and Gethin Norman and David Parker and Hongyang Qu},
  booktitle = {Tools and Algorithms for the Construction and Analysis of Systems},
  publisher = {Springer Berlin Heidelberg},
  title     = {Quantitative Multi-objective Verification for Probabilistic Systems},
  year      = {2011},
  pages     = {112--127},
  doi       = {10.1007/978-3-642-19835-9_11},
}

@InCollection{Brazdil2015,
  author    = {Tom{\'{a}}{\v{s}} Br{\'{a}}zdil and Krishnendu Chatterjee and Vojt{\v{e}}ch Forejt and Anton{\'{\i}}n Ku{\v{c}}era},
  booktitle = {Tools and Algorithms for the Construction and Analysis of Systems},
  publisher = {Springer Berlin Heidelberg},
  title     = {{MultiGain}: A Controller Synthesis Tool for {MDPs} with Multiple Mean-Payoff Objectives},
  year      = {2015},
  pages     = {181--187},
  doi       = {10.1007/978-3-662-46681-0_12},
}

@Article{Hartmanns2020,
  author    = {Hartmanns, Arnd and Junges, Sebastian and Katoen, Joost-Pieter and Quatmann, Tim},
  journal   = {Journal of Automated Reasoning},
  title     = {Multi-cost Bounded Tradeoff Analysis in {MDP}},
  year      = {2020},
  issn      = {1573-0670},
  month     = jul,
  number    = {7},
  pages     = {1483--1522},
  volume    = {64},
  doi       = {10.1007/s10817-020-09574-9},
  publisher = {Springer Science and Business Media LLC},
}

@InBook{Delgrange2020,
  author    = {Delgrange, Florent and Katoen, Joost-Pieter and Quatmann, Tim and Randour, Mickael},
  pages     = {346--364},
  publisher = {Springer International Publishing},
  title     = {Simple Strategies in Multi-Objective MDPs},
  year      = {2020},
  isbn      = {9783030451905},
  booktitle = {Tools and Algorithms for the Construction and Analysis of Systems},
  doi       = {10.1007/978-3-030-45190-5_19},
  issn      = {1611-3349},
}

@InProceedings{Bals2024,
  author     = {Bals, Severin and Evangelidis, Alexandros and Křetínský, Jan and Waibel, Jakob},
  booktitle  = {Proceedings of the 27th ACM International Conference on Hybrid Systems: Computation and Control},
  title      = {MULTIGAIN 2.0: MDP controller synthesis for multiple mean-payoff, {LTL} and steady-state constraints},
  year       = {2024},
  month      = may,
  pages      = {1--7},
  publisher  = {ACM},
  series     = {HSCC ’24},
  collection = {HSCC ’24},
  doi        = {10.1145/3641513.3650135},
}

@InBook{Watanabe2024b,
  author    = {Watanabe, Kazuki and Vegt, Marck van der and Junges, Sebastian and Hasuo, Ichiro},
  pages     = {467--491},
  publisher = {Springer Nature Switzerland},
  title     = {Compositional Value Iteration with Pareto Caching},
  year      = {2024},
  isbn      = {9783031656330},
  booktitle = {Computer Aided Verification},
  doi       = {10.1007/978-3-031-65633-0_21},
  issn      = {1611-3349},
}

@InBook{Forejt2012,
  author    = {Forejt, Vojtěch and Kwiatkowska, Marta and Parker, David},
  pages     = {317--332},
  publisher = {Springer Berlin Heidelberg},
  title     = {Pareto Curves for Probabilistic Model Checking},
  year      = {2012},
  isbn      = {9783642333866},
  booktitle = {Automated Technology for Verification and Analysis},
  doi       = {10.1007/978-3-642-33386-6_25},
  issn      = {1611-3349},
}

@InBook{Watanabe2024a,
  author    = {Watanabe, Kazuki and van der Vegt, Marck and Hasuo, Ichiro and Rot, Jurriaan and Junges, Sebastian},
  pages     = {279--298},
  publisher = {Springer Nature Switzerland},
  title     = {Pareto Curves for Compositionally Model Checking String Diagrams of MDPs},
  year      = {2024},
  isbn      = {9783031572494},
  booktitle = {Tools and Algorithms for the Construction and Analysis of Systems},
  doi       = {10.1007/978-3-031-57249-4_14},
  issn      = {1611-3349},
}

@InBook{Kwiatkowska2011,
  author    = {Kwiatkowska, Marta and Norman, Gethin and Parker, David},
  pages     = {585--591},
  publisher = {Springer Berlin Heidelberg},
  title     = {{PRISM} 4.0: Verification of Probabilistic Real-Time Systems},
  year      = {2011},
  isbn      = {9783642221101},
  booktitle = {Computer Aided Verification},
  doi       = {10.1007/978-3-642-22110-1_47},
  issn      = {1611-3349},
}

@Article{Hensel2021,
  author    = {Hensel, Christian and Junges, Sebastian and Katoen, Joost-Pieter and Quatmann, Tim and Volk, Matthias},
  journal   = {International Journal on Software Tools for Technology Transfer},
  title     = {The probabilistic model checker Storm},
  year      = {2021},
  issn      = {1433-2787},
  month     = jul,
  number    = {4},
  pages     = {589--610},
  volume    = {24},
  doi       = {10.1007/s10009-021-00633-z},
  publisher = {Springer Science and Business Media LLC},
}

@InCollection{Koenighofer2022,
  author    = {Bettina Könighofer and Roderick Bloem and Rüdiger Ehlers and Christian Pek},
  booktitle = {Lecture Notes in Computer Science},
  publisher = {Springer Nature Switzerland},
  title     = {Correct-by-construction runtime enforcement in~{AI} {\textendash} a survey},
  year      = {2022},
  pages     = {650--663},
  doi       = {10.1007/978-3-031-22337-2_31},
}

@InProceedings{Elsayed2021,
  author    = {ElSayed-Aly, Ingy and Bharadwaj, Suda and Amato, Christopher and Ehlers, R\"{u}diger and Topcu, Ufuk and Feng, Lu},
  booktitle = {Proceedings of the 20th International Conference on Autonomous Agents and MultiAgent Systems},
  title     = {Safe Multi-Agent Reinforcement Learning via Shielding},
  year      = {2021},
  address   = {Richland, SC},
  pages     = {483–491},
  publisher = {International Foundation for Autonomous Agents and Multiagent Systems},
  series    = {AAMAS '21},
  isbn      = {9781450383073},
  location  = {Virtual Event, United Kingdom},
  numpages  = {9},
  ranking   = {rank3},
  url       = {https://www.ifaamas.org/Proceedings/aamas2021/pdfs/p483.pdf},
}

@Article{Alshiekh2018,
  author    = {Alshiekh, Mohammed and Bloem, Roderick and Ehlers, Rüdiger and Könighofer, Bettina and Niekum, Scott and Topcu, Ufuk},
  journal   = {Proceedings of the AAAI Conference on Artificial Intelligence},
  title     = {Safe Reinforcement Learning via Shielding},
  year      = {2018},
  issn      = {2159-5399},
  month     = apr,
  number    = {1},
  volume    = {32},
  doi       = {10.1609/aaai.v32i1.11797},
  publisher = {Association for the Advancement of Artificial Intelligence (AAAI)},
}

@InProceedings{Ramstedt2019,
  author    = {Ramstedt, Simon and Pal, Chris},
  booktitle = {Advances in Neural Information Processing Systems},
  title     = {Real-Time Reinforcement Learning},
  year      = {2019},
  editor    = {H. Wallach and H. Larochelle and A. Beygelzimer and F. d\textquotesingle Alch\'{e}-Buc and E. Fox and R. Garnett},
  publisher = {Curran Associates, Inc.},
  volume    = {32},
  url       = {https://proceedings.neurips.cc/paper_files/paper/2019/file/54e36c5ff5f6a1802925ca009f3ebb68-Paper.pdf},
}

@InProceedings{Thodoroff22a,
  author    = {Thodoroff, Pierre and Li, Wenyu and Lawrence, Neil D.},
  booktitle = {NeurIPS 2021 Workshop on Pre-registration in Machine Learning},
  title     = {Benchmarking Real-Time Reinforcement Learning},
  year      = {2022},
  editor    = {Albanie, Samuel and Henriques, Jo\~{a}o F. and Bertinetto, Luca and Hern\'{a}ndez-Garc\'{i}a, Alex and Doughty, Hazel and Varol, G\"{u}l},
  month     = {13 Dec},
  pages     = {26--41},
  publisher = {PMLR},
  series    = {Proceedings of Machine Learning Research},
  volume    = {181},
  url       = {https://proceedings.mlr.press/v181/thodoroff22a.html},
}

@Misc{highway-env,
  author       = {Leurent, Edouard},
  howpublished = {\url{https://github.com/eleurent/highway-env}},
  title        = {An Environment for Autonomous Driving Decision-Making},
  year         = {2018},
  comment      = {Idealised 
http://highway-env.farama.org/dynamics/vehicle/controller/},
  journal      = {GitHub repository},
  publisher    = {GitHub},
  url          = {https://github.com/eleurent/highway-env},
}

@InProceedings{Jansen2020,
  author    = {Nils Jansen and Bettina K{\"o}nighofer and Sebastian Junges and Alex Serban and Roderick Bloem},
  booktitle = {31st International Conference on Concurrency Theory (CONCUR 2020)},
  title     = {{Safe Reinforcement Learning Using Probabilistic Shields (Invited Paper)}},
  year      = {2020},
  address   = {Dagstuhl, Germany},
  editor    = {Igor Konnov and Laura Kov{\'a}cs},
  pages     = {3:1--3:16},
  publisher = {Schloss Dagstuhl--Leibniz-Zentrum f{\"u}r Informatik},
  series    = {Leibniz International Proceedings in Informatics (LIPIcs)},
  volume    = {171},
  doi       = {10.4230/LIPIcs.CONCUR.2020.3},
  isbn      = {978-3-95977-160-3},
  issn      = {1868-8969},
  url       = {https://drops.dagstuhl.de/opus/volltexte/2020/12815},
  urn       = {urn:nbn:de:0030-drops-128155},
}

@InProceedings{Yang2023,
  author    = {Yang, Wen-Chi and Marra, Giuseppe and Rens, Gavin and De Raedt, Luc},
  booktitle = {The 32nd International Joint Conference on Artificial Intelligence (IJCAI)},
  title     = {Safe Reinforcement Learning via Probabilistic Logic Shields},
  year      = {2023},
  copyright = {Creative Commons Attribution 4.0 International},
  doi       = {10.48550/ARXIV.2303.03226},
  ranking   = {rank2},
}

@InCollection{Bacci2022,
  author    = {Edoardo Bacci and David Parker},
  booktitle = {Lecture Notes in Computer Science},
  publisher = {Springer International Publishing},
  title     = {Verified Probabilistic Policies for Deep Reinforcement Learning},
  year      = {2022},
  pages     = {193--212},
  doi       = {10.1007/978-3-031-06773-0_10},
}

@Article{Dunlap2023,
  author    = {Kyle Dunlap and Mark Mote and Kaiden Delsing and Kerianne L. Hobbs},
  journal   = {Journal of Aerospace Information Systems},
  title     = {Run Time Assured Reinforcement Learning for Safe Satellite Docking},
  year      = {2023},
  month     = {jan},
  number    = {1},
  pages     = {25--36},
  volume    = {20},
  doi       = {10.2514/1.i011126},
  publisher = {American Institute of Aeronautics and Astronautics ({AIAA})},
}

@Article{Alur1994,
  author    = {Rajeev Alur and David L. Dill},
  journal   = {Theoretical Computer Science},
  title     = {A theory of timed automata},
  year      = {1994},
  month     = {apr},
  number    = {2},
  pages     = {183--235},
  volume    = {126},
  doi       = {10.1016/0304-3975(94)90010-8},
  publisher = {Elsevier {BV}},
}

@InProceedings{Katoen2016,
  author    = {Katoen, Joost-Pieter},
  title     = {The Probabilistic Model Checking Landscape},
  year      = {2016},
  address   = {New York, NY, USA},
  pages     = {31–45},
  publisher = {Association for Computing Machinery},
  series    = {LICS '16},
  doi       = {10.1145/2933575.2934574},
  isbn      = {9781450343916},
  location  = {New York, NY, USA},
  numpages  = {15},
  url       = {https://doi.org/10.1145/2933575.2934574},
}

@InCollection{Riley2022,
  author    = {Joshua Riley and Radu Calinescu and Colin Paterson and Daniel Kudenko and Alec Banks},
  booktitle = {Lecture Notes in Computer Science},
  publisher = {Springer International Publishing},
  title     = {Assured Deep Multi-Agent Reinforcement Learning for Safe Robotic Systems},
  year      = {2022},
  pages     = {158--180},
  doi       = {10.1007/978-3-031-10161-8_8},
}

@InCollection{Phan2020,
  author    = {Dung T. Phan and Radu Grosu and Nils Jansen and Nicola Paoletti and Scott A. Smolka and Scott D. Stoller},
  booktitle = {Lecture Notes in Computer Science},
  publisher = {Springer International Publishing},
  title     = {Neural Simplex Architecture},
  year      = {2020},
  pages     = {97--114},
  doi       = {10.1007/978-3-030-55754-6_6},
}

@InBook{Kopetz2022,
  author    = {Kopetz, Hermann},
  pages     = {43--52},
  publisher = {Springer International Publishing},
  title     = {Data in Real-Time Control Systems},
  year      = {2022},
  isbn      = {9783030963293},
  booktitle = {SpringerBriefs in Computer Science},
  doi       = {10.1007/978-3-030-96329-3_8},
  issn      = {2191-5776},
}

@InProceedings{Liu2023,
  author    = {Liangkai Liu and Yanzhi Wang and Weisong Shi},
  booktitle = {Proceedings of the Fourth Workshop on Benchmarking Machine Learning Workloads on Emerging Hardware (MLBench)},
  title     = {Understanding Time Variations of {DNN} Inference in Autonomous Driving},
  year      = {2023},
  url       = {https://arxiv.org/abs/2209.05487},
}

@InBook{Vuppala2024,
  author    = {Vuppala, Sai Rohan Harshavardhan and Allen, Nathan and Pinisetty, Srinivas and Roop, Partha},
  pages     = {3--21},
  publisher = {Springer Nature Switzerland},
  title     = {A Formal Approach for Safe Reinforcement Learning: A Rate-Adaptive Pacemaker Case Study},
  year      = {2024},
  isbn      = {9783031742347},
  month     = oct,
  booktitle = {Runtime Verification},
  doi       = {10.1007/978-3-031-74234-7_1},
  issn      = {1611-3349},
}

@Article{Cai2025,
  author    = {Cai, Yihao and Mao, Yanbing and Sha, Lui and Cao, Hongpeng and Caccamo, Marco},
  journal   = {ACM Transactions on Cyber-Physical Systems},
  title     = {Runtime Learning Machine},
  year      = {2025},
  issn      = {2378-9638},
  month     = jun,
  doi       = {10.1145/3744351},
  publisher = {Association for Computing Machinery (ACM)},
  ranking   = {rank4},
}

@Article{Maderbacher2025,
  author    = {Maderbacher, Benedikt and Schupp, Stefan and Bartocci, Ezio and Bloem, Roderick and Ničković, Dejan and Könighofer, Bettina},
  journal   = {International Journal on Software Tools for Technology Transfer},
  title     = {An adaptive, provable correct simplex architecture},
  year      = {2025},
  issn      = {1433-2787},
  month     = feb,
  doi       = {10.1007/s10009-025-00779-0},
  publisher = {Springer Science and Business Media LLC},
}

@InProceedings{bouteiller2020reinforcement,
  author    = {Bouteiller, Yann and Ramstedt, Simon and Beltrame, Giovanni and Pal, Christopher and Binas, Jonathan},
  booktitle = {International conference on learning representations},
  title     = {Reinforcement learning with random delays},
  year      = {2021},
  url       = {https://openreview.net/forum?id=QFYnKlBJYR},
}

@Misc{Da2025,
  author    = {Da, Longchao and Turnau, Justin and Kutralingam, Thirulogasankar Pranav and Velasquez, Alvaro and Shakarian, Paulo and Wei, Hua},
  title     = {A Survey of Sim-to-Real Methods in {RL}: Progress, Prospects and Challenges with Foundation Models},
  year      = {2025},
  copyright = {arXiv.org perpetual, non-exclusive license},
  doi       = {10.48550/ARXIV.2502.13187},
  publisher = {arXiv preprint},
}

@InBook{Koenighofer2020,
  author    = {Könighofer, Bettina and Lorber, Florian and Jansen, Nils and Bloem, Roderick},
  pages     = {290--306},
  publisher = {Springer International Publishing},
  title     = {Shield Synthesis for Reinforcement Learning},
  year      = {2020},
  isbn      = {9783030613624},
  booktitle = {Leveraging Applications of Formal Methods, Verification and Validation: Verification Principles},
  doi       = {10.1007/978-3-030-61362-4_16},
  issn      = {1611-3349},
}

@InProceedings{Li2023a,
  author    = {Li, Zexin and Samanta, Aritra and Li, Yufei and Soltoggio, Andrea and Kim, Hyoseung and Liu, Cong},
  booktitle = {2023 IEEE Real-Time Systems Symposium (RTSS)},
  title     = {$\mathrm{R}^{3}$: On-Device Real-Time Deep Reinforcement Learning for Autonomous Robotics},
  year      = {2023},
  month     = dec,
  pages     = {131--144},
  publisher = {IEEE},
  doi       = {10.1109/rtss59052.2023.00021},
}

@InProceedings{nhu2025timeaware,
  author    = {Anh N Nhu and Sanghyun Son and Ming Lin},
  booktitle = {Forty-second International Conference on Machine Learning},
  title     = {Time-Aware World Model for Adaptive Prediction and Control},
  year      = {2025},
  url       = {https://openreview.net/forum?id=gZ5N3TLjwv},
}

@InProceedings{Shengduo2022,
  author    = {Chen, Shengduo and Sun, Yaowei and Li, Dachuan and Wang, Qiang and Hao, Qi and Sifakis, Joseph},
  booktitle = {2022 International Conference on Robotics and Automation (ICRA)},
  title     = {Runtime Safety Assurance for Learning-enabled Control of Autonomous Driving Vehicles},
  year      = {2022},
  publisher = {IEEE Press},
  doi       = {10.1109/ICRA46639.2022.9812177},
  location  = {Philadelphia, PA, USA},
  numpages  = {7},
  url       = {https://doi.org/10.1109/ICRA46639.2022.9812177},
}

@InProceedings{Dunlap2025,
  author    = {Dunlap, Kyle and Hamilton, Nathaniel P. and Lippay, Zachary and Shubert, Matthew and Phillips, Sean and Hobbs, Kerianne L.},
  booktitle = {AIAA SCITECH 2025 Forum},
  title     = {Demonstrating Reinforcement Learning and Run Time Assurance for Spacecraft Inspection Using Unmanned Aerial Vehicles},
  year      = {2025},
  month     = jan,
  publisher = {American Institute of Aeronautics and Astronautics},
  doi       = {10.2514/6.2025-0756},
  url       = {https://arxiv.org/html/2405.06770v1},
}

@InProceedings{zhao2020sim,
  author       = {Zhao, Wenshuai and Queralta, Jorge Pe{\~n}a and Westerlund, Tomi},
  booktitle    = {2020 IEEE symposium series on computational intelligence (SSCI)},
  title        = {Sim-to-Real transfer in deep reinforcement learning for robotics: a survey},
  year         = {2020},
  organization = {IEEE},
  pages        = {737--744},
}

@InProceedings{Haider2024,
  author     = {Haider, Tom and Roscher, Karsten and Herd, Benjamin and Schmoeller Roza, Felippe and Burton, Simon},
  booktitle  = {Proceedings of the 39th ACM/SIGAPP Symposium on Applied Computing},
  title      = {Can you trust your Agent? The Effect of Out-of-Distribution Detection on the Safety of Reinforcement Learning Systems},
  year       = {2024},
  month      = apr,
  pages      = {1569--1578},
  publisher  = {ACM},
  series     = {SAC ’24},
  collection = {SAC ’24},
  doi        = {10.1145/3605098.3635931},
}

@InProceedings{Hasanbeig2020,
  author    = {Hasanbeig, Mohammadhosein and Abate, Alessandro and Kroening, Daniel},
  booktitle = {Proceedings of the 19th International Conference on Autonomous Agents and MultiAgent Systems},
  title     = {Cautious Reinforcement Learning with Logical Constraints},
  year      = {2020},
  address   = {Richland, SC},
  pages     = {483–491},
  publisher = {International Foundation for Autonomous Agents and Multiagent Systems},
  series    = {AAMAS '20},
  isbn      = {9781450375184},
  location  = {Auckland, New Zealand},
  numpages  = {9},
}

@InProceedings{Rigter2021,
  author    = {Rigter, Marc and Lacerda, Bruno and Hawes, Nick},
  booktitle = {Advances in Neural Information Processing Systems},
  title     = {Risk-Averse Bayes-Adaptive Reinforcement Learning},
  year      = {2021},
  editor    = {M. Ranzato and A. Beygelzimer and Y. Dauphin and P.S. Liang and J. Wortman Vaughan},
  pages     = {1142--1154},
  publisher = {Curran Associates, Inc.},
  volume    = {34},
  url       = {https://proceedings.neurips.cc/paper_files/paper/2021/file/08f90c1a417155361a5c4b8d297e0d78-Paper.pdf},
}

@Article{zanon2020safe,
  author    = {Zanon, Mario and Gros, S{\'e}bastien},
  journal   = {IEEE Transactions on Automatic Control},
  title     = {Safe reinforcement learning using robust MPC},
  year      = {2020},
  number    = {8},
  pages     = {3638--3652},
  volume    = {66},
  publisher = {IEEE},
}

@Article{Reghenzani2023,
  author    = {Reghenzani, Federico and Guo, Zhishan and Fornaciari, William},
  journal   = {ACM Computing Surveys},
  title     = {Software Fault Tolerance in Real-Time Systems: Identifying the Future Research Questions},
  year      = {2023},
  issn      = {1557-7341},
  month     = jul,
  number    = {14s},
  pages     = {1--30},
  volume    = {55},
  doi       = {10.1145/3589950},
  publisher = {Association for Computing Machinery (ACM)},
}

@InBook{Andriushchenko2024,
  author    = {Andriushchenko, Roman and {et al.}},
  pages     = {90--146},
  publisher = {Springer Nature Switzerland},
  title     = {Tools at the Frontiers of Quantitative Verification: QComp 2023 Competition Report},
  year      = {2024},
  isbn      = {9783031676956},
  month     = nov,
  booktitle = {TOOLympics Challenge 2023},
  doi       = {10.1007/978-3-031-67695-6_4},
  issn      = {1611-3349},
}

@Book{boyd2004convex,
  author    = {Boyd, Stephen and Boyd, Stephen P and Vandenberghe, Lieven},
  publisher = {Cambridge university press},
  title     = {Convex optimization},
  year      = {2004},
}

@InProceedings{li2006towards,
  author  = {Li, Lihong and Walsh, Thomas J and Littman, Michael L},
  title   = {Towards a unified theory of state abstraction for MDPs.},
  year    = {2006},
  number  = {2},
  pages   = {3},
  volume  = {1},
  journal = {International Symposium on Artificial Intelligence and Mathematics},
}

@Article{allen2021learning,
  author  = {Allen, Cameron and Parikh, Neev and Gottesman, Omer and Konidaris, George},
  journal = {Advances in Neural Information Processing Systems},
  title   = {Learning {Markov} state abstractions for deep reinforcement learning},
  year    = {2021},
  pages   = {8229--8241},
  volume  = {34},
}

@Article{stable-baselines3,
  author  = {Raffin, Antonin and Hill, Ashley and Gleave, Adam and Kanervisto, Anssi and Ernestus, Maximilian and Dormann, Noah},
  journal = {Journal of Machine Learning Research},
  title   = {Stable-Baselines3: Reliable Reinforcement Learning Implementations},
  year    = {2021},
  number  = {268},
  pages   = {12348-12355},
  volume  = {22},
  url     = {http://jmlr.org/papers/v22/20-1364.html},
}

@InBook{Robinson2023,
  author    = {Robinson, Thomas and Su, Guoxin},
  pages     = {260--277},
  publisher = {Springer Nature Switzerland},
  title     = {Multi-objective Task Assignment and Multiagent Planning with Hybrid GPU-CPU Acceleration},
  year      = {2023},
  isbn      = {9783031331701},
  booktitle = {NASA Formal Methods},
  doi       = {10.1007/978-3-031-33170-1_16},
  issn      = {1611-3349},
}

@Misc{Delavari2025,
  author    = {Delavari, Elahe and Khanzada, Feeza Khan and Kwon, Jaerock},
  title     = {A Comprehensive Review of Reinforcement Learning for Autonomous Driving in the CARLA Simulator},
  year      = {2025},
  copyright = {Creative Commons Attribution 4.0 International},
  doi       = {10.48550/ARXIV.2509.08221},
  publisher = {arXiv},
}

@InBook{Baier2019,
  author    = {Baier, Christel and Hermanns, Holger and Katoen, Joost-Pieter},
  pages     = {420--451},
  publisher = {Springer International Publishing},
  title     = {The 10,000 Facets of MDP Model Checking},
  year      = {2019},
  isbn      = {9783319919089},
  booktitle = {Computing and Software Science},
  doi       = {10.1007/978-3-319-91908-9_21},
  issn      = {1611-3349},
}

@Article{Su2022,
  author    = {Su, Guoxin and Liu, Li and Zhang, Minjie and Rosenblum, David S.},
  journal   = {IEEE Transactions on Software Engineering},
  title     = {Quantitative Verification for Monitoring Event-Streaming Systems},
  year      = {2022},
  issn      = {2326-3881},
  month     = feb,
  number    = {2},
  pages     = {538--550},
  volume    = {48},
  doi       = {10.1109/tse.2020.2996033},
  publisher = {Institute of Electrical and Electronics Engineers (IEEE)},
}

@Article{Su2016,
  author    = {Su, Guoxin and Feng, Yuan and Chen, Taolue and Rosenblum, David S.},
  journal   = {IEEE Transactions on Software Engineering},
  title     = {Asymptotic Perturbation Bounds for Probabilistic Model Checking with Empirically Determined Probability Parameters},
  year      = {2016},
  issn      = {2326-3881},
  month     = jul,
  number    = {7},
  pages     = {623--639},
  volume    = {42},
  doi       = {10.1109/tse.2015.2508444},
  publisher = {Institute of Electrical and Electronics Engineers (IEEE)},
}

@InBook{Daws2005,
  author    = {Daws, Conrado},
  pages     = {280--294},
  publisher = {Springer Berlin Heidelberg},
  title     = {Symbolic and Parametric Model Checking of Discrete-Time Markov Chains},
  year      = {2005},
  isbn      = {9783540318620},
  booktitle = {Theoretical Aspects of Computing - ICTAC 2004},
  doi       = {10.1007/978-3-540-31862-0_21},
  issn      = {1611-3349},
}

@Article{Hahn2010,
  author    = {Hahn, Ernst Moritz and Hermanns, Holger and Zhang, Lijun},
  journal   = {International Journal on Software Tools for Technology Transfer},
  title     = {Probabilistic reachability for parametric Markov models},
  year      = {2010},
  issn      = {1433-2787},
  month     = apr,
  number    = {1},
  pages     = {3--19},
  volume    = {13},
  doi       = {10.1007/s10009-010-0146-x},
  publisher = {Springer Science and Business Media LLC},
}

@InBook{Puggelli2013,
  author    = {Puggelli, Alberto and Li, Wenchao and Sangiovanni-Vincentelli, Alberto L. and Seshia, Sanjit A.},
  pages     = {527--542},
  publisher = {Springer Berlin Heidelberg},
  title     = {Polynomial-Time Verification of PCTL Properties of MDPs with Convex Uncertainties},
  year      = {2013},
  isbn      = {9783642397998},
  booktitle = {Computer Aided Verification},
  doi       = {10.1007/978-3-642-39799-8_35},
  issn      = {1611-3349},
}


\ifacm\appendix\fi

\section{Formal Proofs}\label{sec:proofs}
\subsection{Proofs for Section~\ref{sec:convex_query}}
We present several technical lemma.
Recall that 
a \emph{face} of a bounded polytope $C$ is a subset $H\subseteq C$ such that there is a vector $\vect{v}$ such that $\vect{u}\cdot \vect{v}\geq \vect{u}'\cdot \vect{v}$ for all $\vect{u}\in H,\vect{u'}\in C$.
We recall the following property for polytopes.
\begin{lemma}[\cite{Forejt2012}] \label{lem:polytope}
Let $C$ be a bounded polytope. For any vector $\vect{v}$, there is a face $H$ such that $\vect{u}\cdot \vect{v}=\vect{u}'\cdot \vect{v}$ and $\vect{u}\cdot\vect{v}> \vect{u}''\cdot \vect{v}$ for all $\vect{u},\vect{u'}\in H$ and $\vect{u}''\in C\backslash H$.   
\end{lemma}

\begin{lemma}\label{lem:minimum-angle-points}
Let $D\subset \{\vect{x}\in \mathbb{R}^n \mid \|\vect{x}\|_1=1\}$ such that if $\vect{x}\neq \vect{x}'$ then $\vect{x}\cdot\vect{x}'\leq 0$ or  for all $\vect{x},\vect{x}'\in S$. 
Then, $D$ is a finite set.
\end{lemma}
\begin{proof}
The condition ``if $\vect{x}\neq \vect{x}'$ then $\vect{x}\cdot\vect{x}'\leq 0$'' implies that the minimum angle between $\vect{x}, \vect{x}'$ is $90$ degree. There are at most finitely many (actually $2n$) those points scattering on the $n$-sphere's surface. 
\end{proof}

\begin{lemma}\label{lem:v-h-relationship}
    Throughout the execution of Algorithm~\ref{alg:main-convex-query}, 
$\mathbf{V}(\Phi)\subseteq \Polytope{\vect{\rho}}\subseteq \mathbf{H}(\Lambda)$ holds.
\end{lemma}
\begin{proof}
    $\mathbf{V}(\Phi)\subseteq \Polytope{\vect{\rho}}$ holds as $\Phi\subseteq \Polytope{\vect{\rho}}$. $\Polytope{\vect{\rho}}\subseteq \mathbf{H}(\Lambda)$ holds as each $(\vect{w}',\vect{r}')\in\Lambda$ defines a supporting hyperplane of $\Polytope{\vect{\rho}}$.
\end{proof}

We now show the proof for Theorem~\ref{thm:main-alt}.

\begin{proof}[Proof of Theorem~\ref{thm:main-alt}]
We first show the termination of Theorem~\ref{thm:main-alt}.  
Assume that Algorithm~\ref{alg:main-convex-query} completes in the $\ell^\mathrm{th}$ iteration for any $\iota>1$ but does not terminate yet. 
As $\|\overline{\vect{v}} - \underline{\vect{v}}\|_1\neq 0$ (otherwise the algorithm terminates according to Line~\ref{ln:main-exit-condition2}), $\vect{w}$ is well-defined. 
The proof relies on several claims.
\begin{claim}\label{claim1}
$\vect{w}\cdot\vect{x}=\vect{w}\cdot\overline{\vect{v}}
$ for any $\vect{x}$ defines a supporting hyperplane of $\mathbf{V}(\Phi\backslash\{\vect{r}\})$.
\end{claim}
Let $\Phi'=\Phi\backslash\{\vect{r}\}$. Assume that $\vect{w}\cdot\vect{x}=\vect{w}\cdot\overline{\vect{v}}
$ is \textit{not} a supporting hyperplane of $\mathbf{V}(\Phi')$. 
Then, there is $\vect{y}\in \mathbf{V}(\Phi')$ such that $\vect{w}\cdot(\vect{y}-\overline{\vect{v}}) >0
$.
Let $g(\vect{x})= \|\vect{x}-\underline{\vect{v}}\|^2$ and so $\nabla g(\vect{x})=2(\vect{x}-\underline{\vect{v}})$.
Thus, $\nabla g(\overline{\vect{v}}) = 2(\overline{\vect{v}}-\underline{\vect{v}})= -2 \vect{w}\|\underline{\vect{v}}-\overline{\vect{v}}\|_1$. 
This implies $\nabla g(\overline{\vect{v}})\cdot(\vect{y}-\overline{\vect{v}})<0$. 
As $\mathbf{V}(\Phi')$ is convex, 
this contradicts to $\overline{\vect{v}}$ being a minimiser of $g$ in $\mathbf{V}(\Phi')$.
\begin{claim}\label{claim2}
If $\vect{w}\cdot\vect{r}>\vect{w}\cdot\overline{\vect{v}}$, then $\vect{r}$ is on a new face of $\Polytope{\vect{\rho}}$.
\end{claim}
Actually, by Claim~\ref{claim1}, $\vect{w}\cdot\vect{r}>\vect{w}\cdot\overline{\vect{v}}$ implies $\vect{w}\cdot\vect{r}>\vect{w}\cdot\vect{u}$ for all $\vect{u}\in \mathbf{V}(\Phi\backslash\{\vect{r}\})$. Thus, by Lemma~\ref{lem:polytope}, 
there is a face $H$ of $\Polytope{\vect{\rho}}$ such that $\vect{r}\in H$ and $\vect{r}'\notin H$ for all $\vect{r'}\in \Phi\backslash\{\vect{r}\}$; in other words, $\vect{r}$ is on a new face of $\Polytope{\vect{\rho}}$.
\begin{claim}\label{claim3}
If $\vect{w}\cdot\vect{r}\leq\vect{w}\cdot\overline{\vect{v}}$, then $H_{\vect{w},\overline{\vect{v}}} = \{ \vect{x}\in \Polytope{\vect{\rho}} \mid \vect{w}\cdot\vect{x}=\vect{w}\cdot\overline{\vect{v}}\}$ is a face of $\Polytope{\vect{\rho}}$.
\end{claim}
Actually, as $\mathbf{V}(\Phi)\subseteq \Polytope{\vect{\rho}}$ and $(\vect{w},\vect{r})$ defines a supporting hyperplane for $\Polytope{\vect{\rho}}$, $\vect{w}\cdot\vect{r} \geq \vect{w}\cdot\overline{\vect{v}}$. Thus $\vect{w}\cdot\vect{r}\leq\vect{w}\cdot\overline{\vect{v}}$ equals to $\vect{w}\cdot\vect{r}=\vect{w}\cdot\overline{\vect{v}}$. 

{Now suppose the algorithm does not terminate}. 
In each iteration either the condition of Claim~\ref{claim2} or that of Claim~\ref{claim3} holds. But Claim~\ref{claim2}'s condition cannot hold infinitely often as the number of faces for $\Polytope{\vect{\rho}}$ is finite. Thus, Claim~\ref{claim3}'s condition holds infinitely. 
Also due to the finite number of faces for $\Polytope{\vect{\rho}}$, there must be one infinite sequence of faces $H_{\vect{w}_{1},\overline{\vect{v}}_1},H_{\vect{w}_{2},\overline{\vect{v}}_2},\ldots$ such that all of them coincide with one face of $\Polytope{\vect{\rho}}$.
Then, for each tuple $(i, j)$ such that $i< j$,
\begin{equation*}\begin{aligned}
\vect{w}_i\cdot\vect{w}_j \ & \propto\ \vect{w}_i\cdot(\underline{\vect{v}}_j - \overline{\vect{v}}_j) \\
& =\ \vect{w}_i\cdot(\underline{\vect{v}}_j -\overline{\vect{v}}_i + \overline{\vect{v}}_i-  \overline{\vect{v}}_j)  \\
& =\ \vect{w}_i\cdot(\underline{\vect{v}}_j-\overline{\vect{v}}_i) \quad \text{(as $\overline{\vect{v}}_i, \overline{\vect{v}}_j\in H_{\vect{w}_i,\overline{\vect{v}}_i}$)}\\
& \leq\ \vect{w}_i\cdot(\underline{\vect{v}}_j-{\vect{r}}_i) \quad \text{(as $\vect{w}_i\cdot\vect{r}_i\leq\vect{w}_i\cdot \overline{\vect{v}}_i$}\\
& \qquad \text{is the condition of Claim~\ref{claim3})}\\
& \leq\ 0 
\end{aligned}\end{equation*}
The last inequality holds because Lines~\ref{ln:main-condition1}-\ref{ln:main-find-minimum2} of Alg.~\ref{alg:main-convex-query} ensures that $w_{i}\cdot \underline{\vect{v}}_{k} \leq w_{i}\cdot \vect{r} $ for all $k>i$ (note that we have assumed $|\Lambda|>1$ as $\iota>1$).
This contradicts to Lemma~\ref{lem:minimum-angle-points}. Thus, the algorithm terminates eventually.

As shown above, the algorithm either returns ``infeasible'' or $\overline{\vect{v}}$.  
If it returns ``infeasible'' then $\mathbf{H}(\Lambda)\cap[\vect{l}, \vect{u}]=\emptyset$ and thus $\Polytope{\vect{\rho}}\cap[\vect{l}, \vect{u}]=\emptyset$ by Lemma~\ref{lem:v-h-relationship}.
If it does not return ``infeasible'', then $\overline{\vect{v}}\in [\vect{l},\vect{u}]$ must hold in some future iteration. As $\overline{\vect{v}}\in\Polytope{\vect{\rho}}$ by Lemma~\ref{lem:v-h-relationship}, $\Polytope{\vect{\rho}}\cap[\vect{l}, \vect{u}]\neq\emptyset$. 

Lastly, if the algorithm returns $\overline{\vect{v}}$,
then $\vect{v}\in\Polytope{\vect{\rho}}\cap[\vect{l}, \vect{u}]$ and $\los(\overline{\vect{v}}) - \los(\underline{\vect{v}}) \leq\epsilon$.
Thus, the minised value $\min_{\vect{x}\in\Polytope{\vect{\rho}}\cap[\vect{l},\vect{u}]}\los(\vect{x})$ exits.
Also, $\min_{\vect{x}\in\Polytope{\vect{\rho}}\cap[\vect{l},\vect{u}]}\los(\vect{x})\geq\los(\underline{\vect{v}})$ as $\Polytope{\vect{\rho}}\cap[\vect{l},\vect{u}]\subseteq\mathbf{H}(\Lambda)\cap[\vect{l},\vect{u}]$. Thus,
$\los(\overline{\vect{v}}) - \min_{\vect{x}\in\Polytope{\vect{\rho}}\cap[\vect{l},\vect{u}]}\los(\vect{x}) \leq \varepsilon$ immediately follows.
\end{proof}

\subsection{Proofs for Section~\ref{sec:parallel_pvi}}

\begin{proof}[Proof of Proposition~\ref{thm:value_iteration}]
Let $\schr_*$ be the optimal scheduler computed in Line~\ref{ln:find-supp-hp} of Algorithm~\ref{alg:main-convex-query}, namely, $\schr_* = \argmax_{\schr\in {Sch}(\mdp)} \tolrew{\vect{w}\cdot\vect{\rho}}{\schr}$. Then, for any $\schr\in {Sch}(\mdp)$,
\begin{equation*}
\begin{aligned}
\vect{w}\cdot\tolrew{\vect{\rho}}{\schr} =\ & w_1 \tolrew{\rho_1}{\schr} +\ldots w_m \tolrew{\rho_m}{\schr} \\
=\ & \tolrew{\vect{w}\cdot\vect{\rho}}{\schr}  \quad \text{(Linearity of total rewards)}\\
\leq\ & \tolrew{\vect{w}\cdot\vect{\rho}}{\schr_*}\\
=\ & \vect{w}\cdot\tolrew{\vect{\rho}}{\schr_*}\\
\end{aligned}    
\end{equation*}
This shows that $\vect{w}\cdot \vect{y}=\vect{w}\cdot \tolrew{\vect{\rho}}{\schr_*}$ defines a supporting hyperplane of $\Polytope{\vect{\rho}}$.
\end{proof}

\rev{\section{Complete PRISM Specification of $\mdp_{\mathit{SW}2}$}\label{sec:complete_prism_spec}}
\centering
\lstset{
    language={Prism}, 
    numbers=left,
    frame=single, 
    rulesepcolor=\color{black}, 
    rulecolor=\color{black},
    breaklines=true,
    breakatwhitespace=true,
    firstnumber=1,
    stepnumber=1,
    breakindent=10pt,
    postbreak=\mbox{\textvisiblespace\space}, 
}

\begin{lstlisting}[
    firstline=1, 
    %caption={\rev{Additional Modules and Rewards Definitions for $\mdp_{\mathit{SW}2} $}},
    label={list:prism_model_added},
    linewidth=0.88\textwidth
]
//Switch model with four configurations for primary and secondary controllers 
mdp
const int MAX_TS = 40; // max number of timesteps
const int MIN_CC = 4; // min number of consec calls to a sec contr
//prob to entering sec contr under four confs
const double p1et; 
const double p2et; 
const double p3et; 
const double p4et; 
//prob of existing sec contr under four confs (subject to MIN_CC)
const double p1ex; 
const double p2ex;
const double p3ex;
const double p4ex;
//prob of observing low headway
const double p1hw;
const double p2hw;
const double p3hw;
const double p4hw;
//prob of observing low ttc
const double p1ttc;
const double p2ttc;
const double p3ttc;
const double p4ttc;
//prob of being on lane
const double p1ol;
const double p2ol;
const double p3ol;
const double p4ol;

module BaseSwitch
  conf: [0..4]; //four configurations
  sm: [0..2] init 0; //main state variable
  sp: [0..4] init 0; //primary control state variable
  ss: [0..5] init 0; //secondary control state variable
  // select configurations 
  [conf] sm=0 -> (sm'=1) & (conf'=1);
  [conf] sm=0 -> (sm'=1) & (conf'=2);
  [conf] sm=0 -> (sm'=1) & (conf'=3);
  [conf] sm=0 -> (sm'=1) & (conf'=4);
  [] sm=1 & (sp=0) & (ss=0) & conf=1 -> (1-p1et) : (sp'=1) + p1et : (ss'=1);
  [] sm=1 & (sp=0) & (ss=0) & conf=2 -> (1-p2et) : (sp'=1) + p2et : (ss'=1);
  [] sm=1 & (sp=0) & (ss=0) & conf=3 -> (1-p3et) : (sp'=1) + p3et : (ss'=1);
  [] sm=1 & (sp=0) & (ss=0) & conf=4 -> (1-p4et) : (sp'=1) + p4et : (ss'=1);
  [obs] sp=1 -> (sp'=2);
  [act_p] sp=2 -> (sp'=3);
  [ex_p] sp=3 -> (sm'=2) & (sp'=0);
  [obs] ss=1 -> (ss'=2);
  [act_s] ss=2 -> (ss'=3);
  [] ss=3 & conf=1 -> (1-p1ex) : (ss'=4) + p1ex : (ss'=5);
  [] ss=3 & conf=2 -> (1-p2ex) : (ss'=4) + p2ex : (ss'=5);
  [] ss=3 & conf=3 -> (1-p3ex) : (ss'=4) + p3ex : (ss'=5);
  [] ss=3 & conf=4 -> (1-p4ex) : (ss'=4) + p4ex : (ss'=5);
  [re_s] ss=4 -> (ss'=1);
  [re_s] ss=5 -> (ss'=1);
  [ex_s] ss=5 -> (sm'=2) & (ss'=0);
  [] sm=2 -> (sm'=0) & (conf'=0);
endmodule

module ConsecCallCounter
  cc: [0..MIN_CC]; //consecutive call counts
  [re_s] cc<MIN_CC -> (cc'=cc+1);
  [re_s] cc=MIN_CC -> (cc'=MIN_CC);
  [ex_s] cc=MIN_CC -> (cc'=0);
endmodule

module HWRecorder
  shw: [0..3] init 0; //2: low headway, 3: high headway
  [obs] shw=0 -> (shw'=1);
  [] conf=1 & shw=1 -> p1hw : (shw'=2) + (1-p1hw) : (shw'=3);
  [] conf=2 & shw=1 -> p2hw : (shw'=2) + (1-p2hw) : (shw'=3);
  [] conf=3 & shw=1 -> p3hw : (shw'=2) + (1-p3hw) : (shw'=3);
  [] conf=4 & shw=1 -> p4hw : (shw'=2) + (1-p4hw) : (shw'=3);
  [ex_p] shw=2 -> (shw'=0);
  [ex_s] shw=2 -> (shw'=0);
  [re_s] shw=2 -> (shw'=0);
  [ex_p] shw=3 -> (shw'=0);
  [ex_s] shw=3 -> (shw'=0);
  [re_s] shw=3 -> (shw'=0);
endmodule

module TTCRecorder
  sttc: [0..3] init 0; //2: low ttc, 3: high ttc
  [obs] sttc=0 -> (sttc'=1);
  [] conf=1 & sttc=1 -> p1ttc : (sttc'=2) + (1-p1ttc) : (sttc'=3);
  [] conf=2 & sttc=1 -> p2ttc : (sttc'=2) + (1-p2ttc) : (sttc'=3);
  [] conf=3 & sttc=1 -> p3ttc : (sttc'=2) + (1-p3ttc) : (sttc'=3);
  [] conf=4 & sttc=1 -> p4ttc : (sttc'=2) + (1-p4ttc) : (sttc'=3);
  [ex_p] sttc=2 -> (sttc'=0);
  [ex_s] sttc=2 -> (sttc'=0);
  [re_s] sttc=2 -> (sttc'=0);
  [ex_p] sttc=3 -> (sttc'=0);
  [ex_s] sttc=3 -> (sttc'=0);
  [re_s] sttc=3 -> (sttc'=0);
endmodule

module OLRecorder
  sol: [0..3] init 0; //2: on lane, 3: not on lane
  [obs] sol=0 -> (sol'=1);
  [] conf=1 & sol=1 -> p1ol : (sol'=2) + (1-p1ol) : (sol'=3);
  [] conf=2 & sol=1 -> p2ol : (sol'=2) + (1-p2ol) : (sol'=3);
  [] conf=3 & sol=1 -> p3ol : (sol'=2) + (1-p3ol) : (sol'=3);
  [] conf=4 & sol=1 -> p4ol : (sol'=2) + (1-p4ol) : (sol'=3);
  [ex_p] sol=2 -> (sol'=0);
  [ex_s] sol=2 -> (sol'=0);
  [re_s] sol=2 -> (sol'=0);
  [ex_p] sol=3 -> (sol'=0);
  [ex_s] sol=3 -> (sol'=0);
  [re_s] sol=3 -> (sol'=0);
endmodule

module TimestepCounter
  tc: [0..MAX_TS] init 0;
  //trigger conf and obs only if tc<MAX_TS
  [conf] tc<MAX_TS -> true;
  [obs] tc<MAX_TS -> true;
  //increase timestep counter
  [act_p] tc<MAX_TS -> (tc'=tc+1);
  [act_s] tc<MAX_TS -> (tc'=tc+1);
  //add self-loop to avoid deadlocks
  [] sm<2 & shw=0 & sttc=0 & sol=0 & tc=MAX_TS -> true;
endmodule

rewards "ctrl_cost"
  [act_p] true : 1/MAX_TS; 
  [act_s] true : 2/MAX_TS; 
endrewards

rewards "headway_cost" // low headway (shw=2) has cost 1/MAX_TS
  [ex_p] shw=2 : 1/MAX_TS; 
  [ex_s] shw=2 : 1/MAX_TS; 
  [re_s] shw=2 : 1/MAX_TS; 
endrewards

rewards "lane_dep_cost" // not on lane (sol=3) has cost 1/MAX_TS
  [ex_p] sol=3 : 1/MAX_TS; 
  [ex_s] sol=3 : 1/MAX_TS; 
  [re_s] sol=3 : 1/MAX_TS; 
endrewards

rewards "ttc_cost" // low ttc (sttc=2) has cost 1/MAX_TS
  [ex_p] sttc=2 : 1/MAX_TS; 
  [ex_s] sttc=2 : 1/MAX_TS; 
  [re_s] sttc=2 : 1/MAX_TS; 
endrewards
\end{lstlisting}

\end{document}